\newif\ifital   \italfalse  
\newif\ifdraft  \draftfalse    
\newif\ifhdbk   \hdbkfalse    
\newif\ifcheat  \cheatfalse 
\newif\iflong   \longfalse    
\longtrue 
\documentclass{irmaart}
\usepackage{vaucanson-g}
\renewcommand{\EdgeR}[4][{\EdgeLabelPos}]%
   {%
   \setlength{\EdgeOff}{-\EdgeOff}%
   \EdgeStyle \ncline{#2}{#3} \nbput[npos=#1]{\VaucEdgeLabel{#4}}%
   \setlength{\EdgeOff}{-\EdgeOff}%
   }
\renewcommand{\DimState}%
   {\ChgStateLineStyle{\DimStateLineStyle}%
    \ChgStateLineWidth{\DimStateLineCoef}%
    \ChgStateLineColor{\DimStateLineColor}%
    \ChgStateFillColor{\DimStateFillColor}%
    \ChgStateLabelColor{\DimStateLabelColor}}
\renewcommand{\DimEdge}%
   {\ChgEdgeLineStyle{\DimEdgeLineStyle}%
    \ChgEdgeLineWidth{\DimEdgeLineCoef}%
    \ChgEdgeLineColor{\DimEdgeLineColor}%
    \ChgEdgeLabelColor{\DimEdgeLabelColor}}
\newcommand{\UnDimState}%
   {\RstStateLineStyle%
    \RstStateLineWidth%
    \RstStateLineColor%
    \RstStateFillColor%
    \RstStateLabelColor}
\newcommand{\UnDimEdge}%
   {\RstEdgeLineStyle%
    \RstEdgeLineWidth%
    \RstEdgeLineColor%
    \RstEdgeLabelColor}
\newcommand{\GoldMeanRatio}{1.618 1}% pour l'appel \`a \scalebox
\newcommand{\SQRTwoRatio}{1.414 1}% pour l'appel \`a \scalebox
\newcommand{\ChgEdgeLabelPosit}[1]{\renewcommand{\EdgeLabelPos}{#1}}  
\newcommand{\RstEdgeLabelPosit}{\ChgEdgeLabelPosit{\EdgeLabelPosit}}
\newcommand{\SetEdgeLabelPosit}[1]%
   {\renewcommand{\EdgeLabelPosit}{#1}\RstEdgeLabelPosit}
\newcommand{\ChgTransLabelSep}[1]% coefficient !
   {\setlength{\TransLabelSP}{#1\TransLabelSep}}
\newcommand{\RstTransLabelSep}{\ChgTransLabelSep{1}}
\newcommand{\SetTransLabelSep}[1]% length !
   {\setlength{\TransLabelSep}{#1}\RstTransLabelSep}
\newcommand{\ChgRelTransLabelSep}[1]% coefficient !
      {\setlength{\TransLabelSP}{#1\TransLabelSP}}
\newlength{\NodeSep}\newlength{\NodeSP}% 
\newcommand{\ChgNodeSep}[1]{\setlength{\NodeSP}{#1}} % !! length !!
\newcommand{\RstNodeSep}{\ChgNodeSep{\NodeSep}} % !! length !!
\newcommand{\SetNodeSep}[1]{\setlength{\NodeSep}{#1}\ChgNodeSep{\NodeSep}} % !! length !!
\RstNodeSep% Initial setting
\newcommand{\ChgCoilSize}[1]{\ChgZZSize{#1}}
\newcommand{\RstCoilSize}{\RstZZSize}
\newcommand{\ChgCoilLineWidth}[1]{\ChgZZLineWidth{#1}}

\newcommand{\EdgeZZ}[2]{\ZZEdge{#1}{#2}}%
\newcommand{\EdgeZZL}[4][{\EdgeLabelRevPosit}]%
   {\ZZEdgeL[#1]{#2}{#3}{#4}}
\newcommand{\EdgeZZR}[4][{\EdgeLabelRevPosit}]%
   {\ZZEdgeR[#1]{#2}{#3}{#4}}
\newpsstyle{VaucEdgeStyle}%
    {arrows=\EdgeArrowSty,%
     arrowsize=\EdgeArrowSZDim,%
     arrowlength=\EdgeArrowSZNum,%
     arrowinset=\EdgeArrowIns,%
     linewidth=\EdgeLineWid,%
     linecolor=\EdgeLineCol,%
     linestyle=\EdgeLineSty,%
     doubleline=false,%
     bordercolor=\EdgeLineBorderColor,%
     border=\EdgeLineBord,%
     fillstyle=none,%
     offset=\EdgeOff,%
     labelsep=\TransLabelSP,%
     nodesep=\NodeSP}%\EdgeStateBord
\newpsstyle{VaucEdgeDblStyle}%
    {arrows=\EdgeArrowSty,%
     arrowsize=\EdgeArrowSZDim,%
     arrowlength=\EdgeArrowSZNum,%
     arrowinset=\EdgeArrowIns,%
     linewidth=\EdgeLineDblCoef\EdgeLineWid,%
     linecolor=\EdgeLineCol,%
     linestyle=\EdgeLineSty,%
     doubleline=true,%
     doublesep=\EdgeLineDblSep\EdgeLineWid,%
     bordercolor=\EdgeLineBorderColor,%
     border=\EdgeLineBord,%
     fillstyle=none,%
     offset=\EdgeOff,%
     labelsep=\TransLabelSP,
     nodesep=\NodeSP}%\EdgeStateBord
\newcommand{\VCPutLabel}[4][1]{
\LargeState%  side effect prone
\ChgStateLineStyle{none}%
\ChgStateFillStatus{none}%
\ChgStateLabelScale{#1}%
\StateVar[#4]{#2}{#3}%
\RstStateLabelScale%
\RstStateFillStatus%
\RstStateLineStyle%
}
\newcommand{\VCPutText}[3][1]{
\ChgStateLabelScale{#1}%
\VCPut{#2}{\VCStateLabel{#3}}%
\RstStateLabelScale%
}
\newcommand{\TransSpont}%
   {\psset{dotsep=1.5pt}%
    \ChgEdgeLineStyle{dotted}%
    \ChgEdgeLineWidth{2.2}%
    \ChgEdgeLabelScale{0.85}}
\newcommand{\BigAuto}[1]%
  {\FixStateDiameter{2.8cm}%
   \scalebox{1.272 1}{\State{(0,0)}{XQ}}%
   \ChgStateLabelScale{2}%
   \VCPutStateLabel{(0,0)}{#1}%
   \RstStateLabelScale\MediumState}
\newlength{\sttd}
\newcommand{\ChgStateDiameter}[1]%
    {\setlength{\sttd}{\StateDiam}\FixStateDiameter{#1\sttd}}
\newcommand{\ChgRelStateLineWidth}[1]%
   {\setlength{\StateLineWid}{#1\StateLineWid}}%  
\newcommand{\ChgRelEdgeLineWidth}[1]% coefficient !
   {\setlength{\EdgeLineWid}{#1\EdgeLineWid}}
\newlength{\DotSize}% 
\newlength{\SideLabelDist}
\newcommand{\SideLabelAngle}{135}
\newcommand{\DotState}[4][\SideLabelAngle]%
    {\psdots#2\pnode#2{#3}%\SpecialCoor%\NormalCoor%
     \uput{\SideLabelDist}[#1]{0}#2{\VCStateLabel{#4}}%
    }    
\newcommand{\SetAOSLengthCoef}[1]%
   {\renewcommand{\ArrowOnStateCoef}{#1}\SetAOS}
\newlength{\InitialDist}
\newcommand{\InitialAngle}{180}
\newcommand{\InitialDot}[2][\InitialAngle]%
   {\SpecialCoor%
    \uput{\InitialDist}[#1](#2){\pnode(0,0){#2zz}}%
   \NormalCoor%
   \ChgNodeSep{0.5\DotSize}%
   \EdgeStyle\ncline{#2zz}{#2}
   \RstNodeSep%
   }
\newlength{\FinalDist}
\newcommand{\FinalAngle}{-90}
\newcommand{\FinalDot}[2][\FinalAngle]%
   {\SpecialCoor%
    \uput{\FinalDist}[#1](#2){\pnode(0,0){#2zz}}%
   \NormalCoor%
   \ChgNodeSep{0.5\DotSize}%
   \EdgeStyle\ncline{#2}{#2zz}%
   \RstNodeSep%
   }
\usepackage{xspace}
\usepackage[loose]{subfigure}
\usepackage{ha-js}
\usepackage[switch,pagewise]{lineno}
\ifdraft{\linenumbers}\else{\nolinenumbers}\fi     
\usepackage[hypertex,hyperindex,pagebackref,final]{hyperref}
\usepackage{calc}
\newcounter{hours}\newcounter{minutes}
\newcommand\printtime%
  {\setcounter{hours}{\time/60}%
   \setcounter{minutes}{\time-\value{hours}*60}%
  \thehours\,h\,\,\theminutes}
\newcommand\dateandtime{\today\quad\printtime}
\makeatletter
\def\@evenfoot{{\footnotesize
\texttt{\longshort{Long}{Short} version}  
\hfil 
\texttt{\ifdraft{Not to be circulated}\else{of Ch.~2 of AutoMAthA 
Handbook}\fi}  
}}
\def\@oddfoot{{\footnotesize
\texttt{\longshort{Long}{Short} version}%  
\hfil 
\texttt{\ifdraft{\texttt{Compiled \today \e at \printtime}}%
        \else{\today}\fi}%
}}
\makeatother
\newcommand{\longonly}[1]{\iflong{#1}\else{}\fi}
\newcommand{\shortlong}[2]{\iflong{#2}\else{#1}\fi}
\newcommand{\longshort}[2]{\iflong{#1}\else{#2}\fi}
\newcommand{\shortonly}[1]{\iflong{}\else{#1}\fi}
\newcommand{\hdbkonly}[1]{\ifhdbk{#1}\else{}\fi}
\newcommand{\ArXivonly}[1]{\ifhdbk{}\else{#1}\fi}
\newcommand{\iesf}{\iflong\else\small\fi}%
\newcommand{\JSLandscape}%
   {\ifital \thispagestyle{empty}\mbox{ }\clearpage%
       \addtocounter{page}{-1}\fi}
\newlength{\DefiTest}
\newlength{\DefiHeightu}\newlength{\DefiHeightd}% 
\newlength{\DefiDepthu}\newlength{\DefiDepthd}% 
\newcommand{\Defi}[2]%
    {%
     \settoheight{\DefiHeightu}{${\displaystyle #1}$}%
     \settodepth{\DefiDepthu}{${\displaystyle #1}$}%
     \addtolength{\DefiHeightu}{\DefiDepthu}%
     \settoheight{\DefiHeightd}{${\displaystyle #2}$}%
     \settodepth{\DefiDepthd}{${\displaystyle #2}$}%
     \addtolength{\DefiHeightd}{\DefiDepthd}%
     \left\{#1%
     \rule[-\DefiDepthd]{\DefiTest}{\DefiHeightd}%
     \xmd\right|%
     \left.% 
     \rule[-\DefiDepthu]{\DefiTest}{\DefiHeightu}%
      #2\right\}%
     }
\newlength{\ColoText}% largeur de la colonne "de texte"
\newlength{\ColoFigu}% largeur de la colonne "de figure"
\newlength{\TextFiguSpace}% intervalle entre les deux colonnes
\newlength{\parindenttemp} % for indentation in minipage
\newlength{\parskiptemp} % for alinea spacing in minioage
\newlength{\fboxseptemp} % pour memoriser \fboxsep
\newcommand{\TFBoxing}{}
\newcommand{\TFVertAlig}{}
\newcommand{\LeftLarg}{}
\renewcommand{\LeftLarg}{.66}
\ifdraft\renewcommand{\TFBoxing}{\fbox}\fi
\newcommand{\TxtFg}[3]%
   {%
    \setlength{\ColoText}{#1\textwidth}%
    \setlength{\ColoFigu}{\textwidth}%
    \addtolength{\ColoFigu}{-\ColoText}%
    \addtolength{\ColoText}{-.5\TextFiguSpace}%
    \addtolength{\ColoFigu}{-.5\TextFiguSpace}%
    \ifdraft\setlength{\fboxsep}{0pt}\fi% % mod 000912, 050822
    \noi
    \TFBoxing{%
       \begin{minipage}[\TFVertAlig]{\ColoText}%
          \setlength{\parindent}{\parindenttemp}%
          \setlength{\parskip}{\parskiptemp}%
          \par\vspace*{0mm}% pour l'alignement sur le haut
             #2
       \end{minipage}%
             }%
    \hspace*{\TextFiguSpace}%
    \TFBoxing{%
       \begin{minipage}[\TFVertAlig]{\ColoFigu}%
          \par\vspace*{0mm}%
             #3%
       \end{minipage}%
             }%
    \ifdraft\setlength{\fboxsep}{\fboxseptemp}\fi%050822
   }%
\newcommand{\TextFigu}[3][\LeftLarg]%
   {\renewcommand{\TFVertAlig}{t}\TxtFg{#1}{#2}{#3}}
\newcommand{\TextFiguC}[3][\LeftLarg]%
   {\renewcommand{\TFVertAlig}{c}\TxtFg{#1}{#2}{#3}}
\newcommand{\TextFiguX}[3][\LeftLarg]
   {%
    \setlength{\ColoText}{#1\textwidth}%
    \setlength{\ColoFigu}{\textwidth}%
    \addtolength{\ColoFigu}{-\ColoText}%
    \addtolength{\ColoText}{-.5\TextFiguSpace}%
    \addtolength{\ColoFigu}{-.5\TextFiguSpace}%
    \addtolength{\ColoFigu}{\ETAExtendedLineWidth}% mod 000912,050822
    \ifdraft\setlength{\fboxsep}{0pt}\fi% % mod 000912, 050822
    \noi
    \ifodd\value{page}%
       \TFBoxing{%
          \begin{minipage}[t]{\ColoText}%
             \RstBLS% 050822
             \setlength{\parindent}{\parindenttemp}%
             \setlength{\parskip}{\parskiptemp}%
             \par\vspace*{0mm}% pour l'alignement sur le haut
                #2
          \end{minipage}%
                }%
       \hspace*{\TextFiguSpace}%
       \TFBoxing{%
          \begin{minipage}[t]{\ColoFigu}%
             \par\vspace*{0mm}%
                #3%
          \end{minipage}%
                }%
    \else
       \hspace*{-\ETAExtendedLineWidth}% mod 000912
       \TFBoxing{%
          \begin{minipage}[t]{\ColoFigu}%
             \par\vspace*{0mm}%
                #3%
          \end{minipage}%
                }%
       \hspace*{\TextFiguSpace}%
       \TFBoxing{%
          \begin{minipage}[t]{\ColoText}%
             \RstBLS% 050822
             \setlength{\parindent}{\parindenttemp}%
             \setlength{\parskip}{\parskiptemp}%
             \par\vspace*{0mm}% pour l'alignement sur le haut
                #2
          \end{minipage}%
                }%
    \fi%
    \ifdraft\setlength{\fboxsep}{\fboxseptemp}\fi%050822
   }
\newcommand{\TextText}[2]{\TextFigu[.5]{#1}{#2}}
\newcommand{\medskipneg}{\vspace*{-2ex}} %960728
\newcommand{\smallskipneg}{\vspace*{-1ex}} %960728
\newcommand{\miniskipneg}{\vspace*{-0.25ex}} %960728
\newcommand{\miniskip}{\vspace*{0.1ex}} %
\newcommand{\noi}{\noindent}
\newcommand{\findent}%
    {\ifenglish {}
        \else                                     % en prevision de
        \iffrench {}\else \hspace*{\parindent}\fi
     \fi}                                         % l'utilisation
\newlength{\jsparmini}       % indentation dans les minipages 000118
\newcommand{\ETAEnumLbl}[1]%
    {\iffrench  {\rm  #1)} \fi%
     \ifenglish {\rm (#1)} \fi}
\newcommand{\jsListe}[1]%
    {\noindent\makebox[\parindent][r]{\ETAEnumLbl{#1}}%
     \hspace*{.8em}\ignorespaces}
\newcommand{\jsListb}[1]%
    {\noindent\makebox[\retraitb][r]{\ETAEnumLbl{#1}}%
     \hspace*{.3em}\ignorespaces}
\newcommand{\jsList}[1]%
    {\noindent\makebox[\retraita][r]{\ETAEnumLbl{#1}}%
     \hspace*{.5em}\ignorespaces}
\newcommand{\tha}{\jsListe{a}}
\newcommand{\thb}{\jsListe{b}}

\newcommand{\thi}{\jsListe{i}}
\newcommand{\thii}{\jsListe{ii}}
\newcommand{\thiii}{\jsListe{iii}}
\newcommand{\thiv}{\jsListe{iv}}
\newcommand{\thv}{\jsListe{v}}
\newcommand{\thvi}{\jsListe{vi}}
\newcommand{\thvii}{\jsListe{vii}}
\newcommand{\pointn}{\noindent \makebox[1.2em]{$\bullet$}\ignorespaces}

\newcommand{\Axio}[1]%
   {\pointn #1\hspace*{.1em}\jspointtiret\hspace*{.4em}\ignorespaces}
\newcommand{\redmatu}[1]{\scalebox{0.84}{#1}}

\newcommand{\fa}{\forall}
\newcommand{\ext}{\exists}
\newcommand{\es}{\emptyset}

\newcommand{\OL}{\overline}

\newcommand{\bk}{\mathrel{\backslash }}

\newcommand{\e}{\text{\quad}}                 % un moins petit espace
\newcommand{\ee}{\text{\qquad}}               % un espace
\newcommand{\eee}{\text{\qquad \qquad}} % et un grand
\newsavebox{\InterSymbolSpace}
\savebox{\InterSymbolSpace}{\hspace{0.125em}}
\newsavebox{\SideFormulaSpace}
\savebox{\SideFormulaSpace}{\hspace{0.2em}}
\newcommand{\msp}{\usebox{\SideFormulaSpace}} % espace pour faire ressortir
\newcommand{\xmd}{\usebox{\InterSymbolSpace}} % espace entre les symboles
\newcommand{\eqpnt}{\makebox[0pt][l]{\: .}}
\newcommand{\eqpntvrg}{\makebox[0pt][l]{\: ;}}
\newcommand{\eqvrg}{\makebox[0pt][l]{\: ,}}
\newcommand{\EqVrgInt}{\: , \e }

\newcommand{\EqVrg}{\: ,}
\newcommand{\EqPnt}{\: .}

\newcommand{\quantvrg}{\, , \;}
\newcommand{\quantsp}{\ee }
\newcommand{\quantsmsp}{\e }
\newcommand{\TextInFormula}[1]{\text{\quad #1 \quad }}
\newcommand{\et}{\TextInFormula{and}}

\newcommand{\LatinLocution}[1]{{\itshape #1}\xspace}

\newcommand{\cf}{\LatinLocution{cf.}}

\newcommand{\eg}{\LatinLocution{e.g.}}

\newcommand{\etc}{\LatinLocution{etc.}} % pas de blanc laiss\'e apr\`es etc.

\newcommand{\ie}{{that is, }}

\newcommand{\via}{via\xspace}
\newcommand{\mathjsu}[1]{\mathsf{#1}}
\newcommand{\Kmbb}{\mathbb{K}}
\newcommand{\Lmbb}{\mathbb{L}}

\newcommand{\Tmbb}{\mathbb{T}}

\newcommand{\UNmbb}{{\mathchoice
{\hbox{$\textstyle\rm 1\kern-0.2em I$}}%
{\hbox{$\textstyle\rm 1\kern-0.2em I$}}%
{\hbox{$\scriptstyle\rm 1\kern-0.15em I$}}%
{\hbox{$\scriptscriptstyle\rm 1\kern-0.1em I$}}%
}}
\newcommand{\Ac}{\mathcal{A}}
\newcommand{\Bc}{\mathcal{B}}

\newcommand{\Dc}{\mathcal{D}}
\newcommand{\Ec}{\mathcal{E}}

\newcommand{\Pc}{\mathcal{P}}
\newcommand{\Qc}{\mathcal{Q}}
\newcommand{\Rc}{\mathcal{R}}
\newcommand{\Sc}{\mathcal{S}}
\newcommand{\Tc}{\mathcal{T}}

\newcommand{\Ambf}{\mathbf{A}}

\newcommand{\Cmbf}{\mathbf{C}}
\newcommand{\Dmbf}{\mathbf{D}}

\newcommand{\Imbf}{\mathbf{I}}
\newcommand{\Jmbf}{\mathbf{J}}

\newcommand{\Nmbf}{\mathbf{N}}

\newcommand{\Pmbf}{\mathbf{P}}

\newcommand{\Smbf}{\mathbf{S}}
\newcommand{\Tmbf}{\mathbf{T}}
\newcommand{\Umbf}{\mathbf{U}}

\newcommand{\Emsf}{\mathsf{E}}
\newcommand{\Fmsf}{\mathsf{F}}
\newcommand{\Gmsf}{\mathsf{G}}
\newcommand{\Hmsf}{\mathsf{H}}

\newcommand{\Kmsf}{\mathsf{K}}
\newcommand{\Lmsf}{\mathsf{L}}
\newcommand{\Mmsf}{\mathsf{M}}

\newcommand{\Umsf}{\mathsf{U}}
\newcommand{\Vmsf}{\mathsf{V}}

\newcommand{\Xmsf}{\mathsf{X}}
\newcommand{\Ymsf}{\mathsf{Y}}

\newcommand{\etainfty}{\scalebox{0.85}{\scalebox{0.85}{+}\infty}}

\newlength{\ArrowDiagSize}
\newlength{\ArrowDiagWidth}
\newpsstyle{SLDiagStyle}%
   {colsep=6ex,rowsep=5ex,nodesep=1ex,npos=.45,%
    arrows=->,linewidth=\ArrowDiagWidth,arrowsize=\ArrowDiagSize,%
        linestyle=solid,linecolor=\ArrowDiagColor}

\newenvironment{SLDiag}%
   {\psset{style=SLDiagStyle}\begin{psmatrix}}%
   {\end{psmatrix}}%
\newcommand{\CDSL}{\begin{SLDiag}}
\newcommand{\CDSLF}{\end{SLDiag}}
\newenvironment{DiagraBig}%
{\psmatrix[colsep=7ex,rowsep=6ex,arrows=->,nodesep=1ex,npos=.45]}%
{\endpsmatrix}
\newcommand{\CDB}{\begin{DiagraBig}}
\newcommand{\CDBF}{\end{DiagraBig}}
\newenvironment{DiagraSmall}%
{\psmatrix[colsep=3ex,rowsep=3ex,arrows=->,nodesep=1ex,npos=.45]}%
{\endpsmatrix}
\newcommand{\CDS}{\begin{DiagraSmall}}
\newcommand{\CDSF}{\end{DiagraSmall}}
\newcommand{\matriceuu}[1]%
    {\begin{pmatrix} #1 \end{pmatrix}}
\newcommand{\matricedd}[4]%
    {\begin{pmatrix} #1 & #2 \\ #3 & #4 \end{pmatrix}}
\newcommand{\vecteurd}[2]%
    {\begin{pmatrix} #1 \\ #2 \end{pmatrix}}
\newcommand{\ligned}[2]%
    {\begin{pmatrix} #1 & #2 \end{pmatrix}}
\newcommand{\matricett}[9]%
    {\begin{pmatrix}  #1 & #2 & #3 \\
                      #4 & #5 & #6 \\
                      #7 & #8 & #9 \end{pmatrix}}
\newcommand{\vecteurt}[3]%
    {\begin{pmatrix} #1 \\ #2 \\ #3 \end{pmatrix}}
\newcommand{\lignet}[3]%
    {\begin{pmatrix} #1 & #2 & #3 \end{pmatrix}}
\newlength{\jsWidthCol}
\newlength{\blocinterligne}
\newlength{\blocinterligned}
\newlength{\temparraycolsep}
\newlength{\longueurbloc}
\newlength{\hauteurbloc}
\newlength{\centragebloc}
\newlength{\longueurblc}
\newlength{\hauteurblc}
\newlength{\centrageblc}
\newcommand{\blocligne}[1]%
    {\framebox[\longueurbloc]{$#1$}}
\newcommand{\blocmatrice}[1]%
    {\framebox[\longueurbloc]{\rule[\centragebloc]{0mm}{\hauteurbloc}$#1$}}
\newcommand{\blocvecteur}[1]%
    {\framebox{\rule[\centragebloc]{0mm}{\hauteurbloc}$#1$}}
\newcommand{\blcligne}[1]%
    {\framebox[\longueurblc]{$#1$}}
\newcommand{\blcmatrice}[1]%
    {\framebox[\longueurblc]{\rule[\centrageblc]{0mm}{\hauteurblc}$#1$}}
\newcommand{\blcvecteur}[1]%
    {\framebox{\rule[\centrageblc]{0mm}{\hauteurblc}$#1$}}
\newcommand{\matriceddblvs}[4]%%
   {\setlength{\temparraycolsep}{\arraycolsep}%
    \setlength{\arraycolsep}{1.3pt}%
        \left (%
    \begin{array}{cc}%
                #1  & \blcligne{#2} \\
            \blcvecteur{#3} & \blcmatrice{#4}
        \end{array}%
        \right )%
    \setlength{\arraycolsep}{\temparraycolsep}%
   }%
\newcommand{\vecteurdblvs}[2]%
   {\setlength{\temparraycolsep}{\arraycolsep}%
    \setlength{\arraycolsep}{1.5pt}%
        \left (%
    \begin{array}{c}%
                #1  \\
            \blcvecteur{#2}
        \end{array}%
        \right )%
    \setlength{\arraycolsep}{\temparraycolsep}%
   }%
\newcommand{\lignedblvs}[2]%
   {\setlength{\temparraycolsep}{\arraycolsep}%
    \setlength{\arraycolsep}{1.5pt}%
        \left (%
    \begin{array}{cc}%
                #1  & \blcligne{#2}
        \end{array}%
        \right )%
    \setlength{\arraycolsep}{\temparraycolsep}%
   }%
\newcommand{\matricettblvs}[9]%%
   {\setlength{\temparraycolsep}{\arraycolsep}%
    \setlength{\arraycolsep}{1.5pt}%
        \left (%
    \begin{array}{ccc}%
                #1  & \blcligne{#2} & #3\\
            \blcvecteur{#4} & \blcmatrice{#5} & \blcvecteur{#6}\\
                #7  & \blcligne{#8} & #9\\
        \end{array}%
        \right )%
    \setlength{\arraycolsep}{\temparraycolsep}%
   }%
\newcommand{\vecteurtblvs}[3]%
   {\setlength{\temparraycolsep}{\arraycolsep}%
    \setlength{\arraycolsep}{1.5pt}%
        \left (%
    \begin{array}{c}%
                #1  \\
            \blcvecteur{#2}\\
                #3
        \end{array}%
        \right )%
    \setlength{\arraycolsep}{\temparraycolsep}%
   }%
\newcommand{\lignetblvs}[3]%
   {\setlength{\temparraycolsep}{\arraycolsep}%
    \setlength{\arraycolsep}{1.5pt}%
        \left (%
    \begin{array}{ccc}%
                #1  & \blcligne{#2} & #3
        \end{array}%
        \right )%
    \setlength{\arraycolsep}{\temparraycolsep}%
   }%
\newcommand{\matricettblblvs}[9]%%
   {\setlength{\temparraycolsep}{\arraycolsep}%
    \setlength{\arraycolsep}{1.5pt}%
        \left (%
    \begin{array}{ccc}%
                #1  & \blcligne{#2} & \blcligne{#3}\\
            \blcvecteur{#4} & \blcmatrice{#5} & \blcmatrice{#6}\\
                \blcvecteur{#7}  & \blcmatrice{#8} & \blcmatrice{#9}\\
        \end{array}%
        \right )%
    \setlength{\arraycolsep}{\temparraycolsep}%
   }%
\newcommand{\vecteurtblblvs}[3]%
   {\setlength{\temparraycolsep}{\arraycolsep}%
    \setlength{\arraycolsep}{1.5pt}%
        \left (%
    \begin{array}{c}%
                #1  \\
            \blcvecteur{#2}\\
                \blcvecteur{#3}
        \end{array}%
        \right )%
    \setlength{\arraycolsep}{\temparraycolsep}%
   }%
\newcommand{\lignetblblvs}[3]%
   {\setlength{\temparraycolsep}{\arraycolsep}%
    \setlength{\arraycolsep}{1.5pt}%
        \left (%
    \begin{array}{ccc}%
                #1  & \blcligne{#2} & \blcligne{#3}
        \end{array}%
        \right )%
    \setlength{\arraycolsep}{\temparraycolsep}%
   }%
\newcommand{\jsCard}[1]{{\|{#1}\|}}
\newcommand{\fracts}[2]{{\textstyle \frac{#1}{#2}}}
\newcommand{\ExtnF}[1]%
   {\overset{{\scriptscriptstyle \pmb{\smile}}}{#1}}

\newcommand{\DiffF}[1]%
   {\overset{{\scriptscriptstyle \pmb{\lor}}}{#1}}
\newcommand{\LocaF}[1]%
   {\overset{{\scriptscriptstyle \leftrightarrow}}{#1}}

\newcommand{\jsDist}[2][{}]%
   {\operatorname{\mathbf{d}_{#1}}\left(#2\right)}

\newcommand{\jsHaut}[1]{{\operatorname{\mathsf{h}}[#1]}}

\newcommand{\jsEnla}[1]{{\operatorname{\mathsf{lc}}(#1)}}
\newcommand{\iEB}[1]{{\operatorname{\mathsf{i}_{\EBzi}}(#1)}}

\renewcommand{\min}{{\operatornamewithlimits{\mathsf{min}}}}
\renewcommand{\max}{{\operatornamewithlimits{\mathsf{max}}}}

\renewcommand{\lim}{{\operatornamewithlimits{\mathsf{lim}}}}

\newcommand{\Pfrak}{\mathfrak{P}}
\newcommand{\jsPart}[1]{{\operatorname{\Pfrak}\left(#1\right)}}
\newcommand{\jspart}[1]{\jsPart{#1}} % compatibilit\'e

\newcommand{\PdAe}{\jspart{\Ae}}
\newcommand{\StruSA}[1]{\aut{#1}}
\newcommand{\ETAze}[1]{0_{#1}}
\newcommand{\ETAun}[1]{1_{#1}}
\newcommand{\zeK}{\ETAze{\K}}

\newcommand{\unK}{\ETAun{\K}}

\newcommand{\unT}{\ETAun{\T}}

\newcommand{\x}{\! \times \!}
\newcommand{\KQQ}{{\K}^{Q \x Q}}

\newcommand{\TermCst}[1]{{\operatorname{\mathsf{c}}(#1)}}

\newcommand{\SerSAnMon}[2]%
    {#1 \langle \! \langle  #2  \rangle \! \rangle }
\newcommand{\SerSAnMonD}[2]%
    {\left[#1\right] \langle \! \langle  #2  \rangle \! \rangle }
\newcommand{\SerMon}[1]%
    {\!\langle \! \langle  #1  \rangle \! \rangle }
\newcommand{\KAe}{\SerSAnMon{\K}{\Ae }}

\newcommand{\KA}{\KAe} % compatibilit\'e
\newcommand{\KM}{\SerSAnMon{\K}{M}}

\newcommand{\KMQQ}{\KM ^{Q \x Q}}

\newcommand{\KQQM}{\SerSAnMon{\KQQ }{M}}

\newcommand{\PolSAnMon}[2]%
    {{#1 \langle  #2 \rangle }}
\newcommand{\PolMon}[1]%
    {{\!\langle  #1 \rangle }}

\newcommand{\KPM}{\PolSAnMon{\K}{M}}

\newsavebox{\LeftBraket}
\savebox{\LeftBraket}{\scalebox{0.7 1.2}{$<$}}
\newsavebox{\RightBraket}
\savebox{\RightBraket}{\scalebox{0.7 1.2}{$>$}}
\newcommand{\bra}[1]{\hbox{}\usebox{\LeftBraket}%
                           #1\usebox{\RightBraket}\hbox{}}
\newcommand{\Rec}{\mathrm{Rec}\,}
\newcommand{\Rat}{\mathrm{Rat}\,}
\newcommand{\KRat}{\K \Rat}

\newcommand{\KRec}{\K \Rec}

\newcommand{\RecA}{\Rec \Ae}

\newcommand{\RatA}{{{\Rat \Ae}}}

\newcommand{\lmn}{{(\lambda , \mu , \nu )}}

\newcommand{\jsStar}[1]{{{#1}^{*}}}
\newcommand{\Ae}{\jsStar{A}}

\newcommand{\jsPlus}[1]{{{#1}^{+}}}
\newcommand{\Ap}{\jsPlus{A}}

\newcommand{\iotaK}{\iota_{\ShiftInd{K}}}

\newcommand{\unAe}{{1_{\Ae}}}

\newcommand{\unM}{{1_{\!M}}}
\newcommand{\unN}{{1_{N}}}

\newcommand{\RecM}{{{\Rec M}}}

\newcommand{\RatM}{{{\Rat M}}}

\newcommand{\compos}{\ccdot }
\newcommand{\matmul}{\mathbin{\cdot}}

\newcommand{\plusopr}{\mathbin{\mathjsu{+}}}
\newcommand{\prodopr}{\mathbin{\mathjsu{\cdot}}}

\newcommand{\dedjs}{\mathrel{%
   \, \rule{0.12ex}{1.5ex}\rule[0.69ex]{1em}{0.12ex} \,} \ee }
\newcommand{\dedtxt}{\mathrel{%
   \, \rule{0.12ex}{1.5ex}\rule[0.69ex]{1em}{0.12ex} \,}\msp\msp}
\newcommand{\phiikpsi}%
{{\varphi ^{-1}\! \compos        \iotaK \! \compos \! \psi }}
\newcommand{\phiiotpsi}[1]%
{{\varphi ^{-1}\! \compos        \iota _{\ShiftInd{#1}} \! \compos \! \psi }}
\newcommand{\phiintkpsi}[1]%
{{(#1\varphi ^{-1}\! \cap K) \psi }}

\newcommand{\jsleq}{\leqslant }
\newcommand{\jsgeq}{\geqslant }
\newcommand{\jsless}% 001102
   {\mathrel{\leqslant_{\!\!\!\!\scriptscriptstyle{/}}}}
\newcommand{\jsgrea}% 020120
   {\mathrel{\geqslant_{\!\!\!\!\scriptscriptstyle{\backslash}}}}

\newcommand{\lexiconeq}% 010125
   {\preccurlyeq_{\!\!\!\!\!\scalebox{1.8 1}{\scriptscriptstyle{\pmb{/}}}}}
\newcommand{\jsAutUn}[1]%  
   {\mbox{$\left\langle \thinspace #1 \thinspace \right\rangle $}}
\newcommand{\aut}[1]{\jsAutUn{#1}} % pour compatibilit\'e
\newcommand{\auta}{\jsAutUn{Q,A,E,I,T}}

\newcommand{\autiet}{\jsAutUn{I, E ,T}}

\newcommand{\Tran}[1]{\bigl ( #1\bigr )}
\newcommand{\jsNorm}[1]{{#1_{\mathsf{nor}}}}
\newcommand{\ShiftInd}[1]{\raisebox{-0.3ex}{$\scriptstyle{#1}$}}
\newcommand{\actb}{\mathbin{\raisebox{0.2ex}%
                        {${\scriptscriptstyle \circ} $}}}
\newcommand{\ccdot}{\actb} % pour compatibilit\'e
\newcommand{\dpT}{{\delta_{p,T}}} %000521
\newcommand{\dpI}{{\delta_{p,I}}} %000521
\newlength{\vbh}\newlength{\vbd}\newlength{\vbt}%
\newcommand{\CompAuto}[1]%
    {%
     \settodepth{\vbd}{\mbox{$\displaystyle{#1\strut}$}}%
     \settoheight{\vbh}{\mbox{$\displaystyle{#1\strut}$}}%
     \setlength{\vbt}{\vbh}\addtolength{\vbt}{\vbd}%
     {}%
     \psline[linewidth=0.8pt]{c-c}(0,-.65\vbd)(0,.9\vbh)%
     \hspace*{0.7pt}%    
     {#1}%
     \kern0.8pt%
     \psline[linewidth=0.8pt]{c-c}(0,-.65\vbd)(0,.9\vbh)%
     }%
\newcommand{\mBC}[2]{{\mathjsu{E}}^{(#1)}(#2)}% _{\BzCi}
\newcommand{\mMNY}[2]{{\mathjsu{M}}^{(#1)}_{#2}}
\newcommand{\bornedeuxlignes}[2]%
{\mbox{$
\begin{array}{c}{\scriptstyle #1}\\ {\scriptstyle #2} \end{array}
       $}}
\renewcommand{\path}[1]{\xrightarrow{\ #1 \ }} %100309 pb
\newcommand{\pathaut}[2]{\underset{#2}{\path{#1}}}
\newcommand{\RatE}{\mathjsu{RatE}\,}% mod. 020212
\newcommand{\RatEA}{\mathjsu{RatE}\,\Ae}% mod. 020430,030123
\newcommand{\RatExp}[2]{{#1\,\mathjsu{RatE}\,#2}}
\newcommand{\KRA}{\RatExp{\K}{\Ae}}
\newcommand{\KREM}{\RatExp{\K}{M}}
\newcommand{\Dpth}[1]{\operatorname{\mathsf{d}}(#1)}
\newcommand{\Depth}[1]{\Dpth{#1}} %compatibility
\newcommand{\ExpDer}[2][a]%
    {\operatorname{\frac{\partial}{\partial \mbox{$#1$}}}#2}
\newcommand{\ExpDerP}[2][a]%
    {\operatorname{\frac{\partial}{\partial\mbox{$#1$}}}\left(#2\right)}
\newcommand{\ExpDerr}[2][a]%
    {\operatorname{\frac{\partial_{\mathrm{R}}}{\partial \mbox{$#1$}}}#2}
\newcommand{\ExpDerB}[2][a]%
   {\operatorname{\frac{\partial_\mathsf{b}}{\partial \mbox{$#1$}}}#2}
\newcommand{\ExpDerBP}[2][a]%
   {\operatorname{\frac{\partial_\mathsf{b}}{\partial \mbox{$#1$}}}\left(#2\right)}
\newcommand{\TerScale}{0.9}
\newcommand{\DTer}{\scalebox{\TerScale}{\mathrm{D}}}
\newcommand{\TDTer}{\scalebox{\TerScale}{\mathrm{TD}}}
\newcommand{\DerTer}[1]{{\DTer}\left(#1\right)}

\newcommand{\TruDerTer}[1]{{\TDTer}\left(#1\right)}

\newcommand{\CompExpr}[1]{\CompAuto{#1}}
\newcommand{\Kadd}{\oplus}
\newcommand{\Ksum}{\bigoplus}
\newcommand{\LitLength}{\operatorname{\ell}}
\newcommand{\LitL}[1]{\LitLength (#1)}

\newtheorem{property}[theorem]{Property}
\newcommand{\corol}[2][{}]{Corollary~\ref{#1:cor:#2}}%
\newcommand{\defin}[2][{}]{Definition~\ref{#1:def:#2}}%
\newcommand{\equat}[2][{}]{Equation~(\ref{#1:equ:#2})}%
\newcommand{\equnm}[2][{}]{(\ref{#1:equ:#2})\xspace}%
\newcommand{\equnmnosp}[2][{}]{(\ref{#1:equ:#2})}%
\newcommand{\exemp}[2][{}]{Example~\ref{#1:exa:#2}}%
\newcommand{\figur}[2][{}]{Figure~\ref{#1:fig:#2}}%
\newcommand{\lemme}[2][{}]{Lemma~\ref{#1:lem:#2}}%
\newcommand{\propo}[2][{}]{Proposition~\ref{#1:pro:#2}}%
\newcommand{\prpri}[2][{}]{Property~\ref{#1:pty:#2}}%
\newcommand{\theor}[2][{}]{Theorem~\ref{#1:the:#2}}%
\newcommand{\secti}[2][{}]{Section~\ref{#1:sec:#2}}%
\newcommand{\sbsct}[2][{}]{Section~\ref{#1:ssc:#2}}%
\newcommand{\CTchp}[1]{Chapter~\ref{#1:chp:#1}}%
\iflong\else
\renewcommand{\secti}[2][{}]{Sec.~\ref{#1:sec:#2}}%
\renewcommand{\sbsct}[2][{}]{Sec.~\ref{#1:ssc:#2}}%
\renewcommand{\CTchp}[1]{Chap.~\ref{#1:chp:#1}}%
\fi 
\newcommand{\NOTA}%
   {\medskip\noindent\textnormal{\textbf{Notation}}
    \hspace*{0.5em}\ignorespaces}

\newcommand{\EOP}{{\strut}\hfill\qedsymbol}%
\newcommand{\PushLine}{~\hfill~}
\newlength{\retraita}
\newlength{\listespa}
\newcommand{\EnumLbl}[1]{{\rm (#1)}}
\renewcommand{\jsListe}[1]%
    {\noindent\makebox[\retraita][r]{\EnumLbl{#1}}%
     \hspace*{\listespa}\ignorespaces}

\newcommand{\dex}[1]{\xmd(#1)}
\newcommand{\NeM}[1]%
   {\ifdraft%
    \marginpar[\begin{flushright}%
               {\sl {\scriptsize #1}}%
               \end{flushright}]%
              {\begin{flushleft}%
               {\sl {\scriptsize #1}}%
               \end{flushleft}}%
	\fi}%
\renewcommand\phi\varphi
\renewcommand\epsilon\varepsilon
\newcommand{\B}{\mathbb{B}}
\newcommand{\K}{\Kmbb}
\newcommand{\Lb}{\Lmbb}
\newcommand{\T}{\Tmbb}
\newcommand{\Ed}{\Emsf}
\newcommand{\Fd}{\Fmsf}
\newcommand{\Gd}{\Gmsf}
\newcommand{\Hd}{\Hmsf}

\newcommand{\Kd}{\Kmsf}
\newcommand{\Ld}{\Lmsf}
\newcommand{\Md}{\Mmsf}

\newcommand{\Ud}{\Umsf}
\newcommand{\Vd}{\Vmsf}

\newcommand{\Xd}{\Xmsf}
\newcommand{\Yd}{\Ymsf}

\newcommand{\zed}{\mathsf{0}}
\newcommand{\und}{\mathsf{1}}
\newcommand{\IdR}[1]{\mathbf{#1}}
\newcommand{\IdRAs}{\IdR{A}}

\newcommand{\IdRC}{\IdR{C}}
\newcommand{\IdRD}{\IdR{D}}
\newcommand{\IdRI}{\IdR{I}}

\newcommand{\IdRJ}{\IdR{J}}

\newcommand{\IdRN}{\IdR{N}}
\newcommand{\IdRM}{\IdR{P}}% compatibilit\'e
\newcommand{\IdRP}{\IdR{P}}
\newcommand{\IdRS}{\IdR{S}}

\newcommand{\IdRT}{\IdR{T}}
\newcommand{\IdRA}{\IdR{U}}% compatibilit\'e
\newcommand{\IdRU}{\IdR{U}}

\newcommand{\autplus}{+}
\newcommand{\autprod}{\cdot}
\newcommand{\autstarsymb}{\strut^{*}}
\newcommand{\autstar}[1]{#1^{*}}
\newcommand{\nsp}{\hspace*{0.2em}} %space for making symbols stand out 
\newcommand{\cupsp}{\nsp \cup \nsp}
\newcommand{\cupssp}{\xmd \cup \xmd}
\newcommand{\setvrg}{\xmd, \nsp}
\newcommand{\RatEM}{\RatE M}% mod. 020430,030123
\newcommand{\KKREA}{\PolSAnMon{\K}{\KRA}}%
\newcommand{\LiteLgth}[1]{\operatorname{\ell}\left(#1\right)}
\newcommand{\autshrnk}[3]{\aut{\hspace{-#2em}#1\hspace{-#3em}}}
\newcommand{\autsk}[1]{\autshrnk{#1}{0.6}{0.4}}
\newcommand{\ExprLine}[1]{\OL{#1}}
\newcommand{\Edl}{\ExprLine{\Ed}}
\newcommand{\pyaexponent}[1]%
    {\raisebox{.3ex}{\hspace{0.05em}$\scalebox{.4}{#1}$}}

\newcommand{\squa}{\scriptscriptstyle\square}

\newcommand{\rn}[1]{{#1}^{\bullet}}

\newcommand{\cb}[1]{{#1}^{\pyaexponent{\square}}}

\newcommand{\Ern}{\rn{\Ed}}
\newcommand{\Frn}{\rn{\Fd}}
\newcommand{\Grn}{\rn{\Gd}}

\newcommand{\Ecb}{\cb{\Ed}}
\newcommand{\Fcb}{\cb{\Fd}}
\newcommand{\Gcb}{\cb{\Gd}}
\newcommand{\DTAut}[1]{\Ac_{#1}}
\renewcommand{\jsCard}[1]{\operatorname{\mathsf{Card}}\left(#1\right)}
\renewcommand{\EqVrgInt}{\, , \hspace{.8em}}
\newcommand{\MNYi}{\text{MN-Y}}
\newcommand{\MNY}{\MNYi\xspace}
\newcommand{\EBz}{\xspace}%  mod 020410
\newcommand{\EBzi}{}
\newcommand{\AtEs}{\Gamma}
\newcommand{\EtAs}{\Delta}
\newcommand{\atealgo}[3]{{\mathbf{#1}}_{#2}\!\left(#3\right)\xspace}
\newcommand{\sem}[2][\omega]{\atealgo{B}{#1}{#2}}
\newcommand{\srm}[2][\omega]{\atealgo{E}{#1}{#2}}
\newcommand{\mny}[2][\omega]{\atealgo{M}{#1}{#2}}
\newcommand{\rcm}[2][\tau]{\atealgo{C}{#1}{#2}}
\newcommand{\DerT}[1]{\Ac_{#1}}
\newcommand{\Stan}[1]{\Sc_{#1}}

\newcommand{\Thom}[1]{\Tc_{#1}}
\newcommand{\One}[1]{{#1}^{\mathsf{1}}}
\newcommand{\Zero}[1]{{#1}^{\mathsf{0}}}
\newcommand{\StanAuto}[4]%
   {\aut{\redmatu{\ligned{1}{0}},%
         \redmatu{\matricedd{0}{#1}{0}{#2}},%
         \redmatu{\vecteurd{#3}{#4}}}}
\newcommand{\StanAutoSP}[9]%
   {\aut{\redmatu{\lignet{1}{0}{0}},%
         \redmatu{\matricett{0}{#1}{#2}%
                            {0}{#3}{#4}%
                            {0}{#5}{#6}},%
         \redmatu{\vecteurt{#7}{#8}{#9}}}}
\newcommand{\ILC}[2][\omega]{\operatorname{\mathfrak{I}_{#1}}\!\left(#2\right)\xspace}
\makeatletter
 \@input{AE-refcro.aux}
\makeatother
\SmallPicture
\ifdraft\ShowFrame\fi 
\newlength{\lga}\newlength{\lgb}\newlength{\lgc}
\begin{document}
\JSLandscape
%========== Headers =============================
\markboth{J.~Sakarovitch}{Automata and rational expressions}
%=====================================================
\title{%
\longonly{\vspace*{-4ex}}
Automata and rational expressions%
\longonly{\vspace*{-2ex}}} 
\author%
  {Jacques Sakarovitch}
\ifhdbk
\address{Laboratoire Traitement et Communication de l'Information\\
  CNRS and T\'el\'ecom ParisTech\\
\ifcheat
email:\,\url{sakarovitch@telecom-paristech.fr}
\else
email:\,\url{sakarovitch@telecom-paristech.fr}\\[4mm]
\ifdraft{\upshape{\dateandtime}}\fi
\fi}
\else
\address{LTCI,\ CNRS and T\'el\'ecom ParisTech}
\fi

%=====================================================
% \longonly{\vspace*{-3ex}}
\maketitle\label{chapter2} 
%=====================================================

%=====================================================
\ifhdbk\setcounter{page}{35}\fi
%=====================================================

%========== Section call =============================
%%%%%%  HAMA Chapter 2 Front Matters   %%%%%%%%%%
%%%             150205               %%%
%%%%%%%%%%%%%%%%%%%%%%%%%%%%%%%%%%%%%%%%%%%%
\ArXivonly{\vspace*{-6ex}\medskipneg\medskipneg\medskipneg}
\hdbkonly{%
\ifcheat\medskipneg\fi%
\ifcheat\medskipneg\fi%
\ifcheat\medskipneg\fi%
}%
\begin{classification}
  68Q45
\end{classification}

\begin{keywords}
  Finite automata, regular expressions.
  Rational sets, recognisable sets.
%   Kleene's theorem.
\end{keywords}

%==========Table of contents=======================
\longshort{%
\enlargethispage*{7ex}
This text is an extended version of the chapter 
`Automata and rational expressions' in the 
\emph{AutoMathA Handbook}~\cite{Pin15Ed} that will appear soon,  
published by the European Science Foundation
and 
edited by Jean-{\'E}ric Pin.

It contains not only proofs, examples, and remarks that had been 
discarded due to the severe space constraints induced by the edition 
of a handbook of very large scope, but also developments that were 
not included as they did not seem to belong to the main stream of 
the subject. 
For that reason, the numbering of theorems, propositions, 
defintions, \etc may differ in the two versions, even if the general 
outline is the same.

% \clearpage
\vspace*{-2ex}
\localtableofcontents

% \clearpage
}{%
% too much tempting  !!!!
\ifcheat\medskipneg\fi 
\ifcheat\medskipneg\fi
\localtableofcontents
\ifcheat\clearpage\fi
}

%%%%%%%%%%%%%  AE Section 1  %%%%%%%%%%%%%%%
%%%%%%%%%%%%%%%%%%%%%%%%%%%%%%%%%%%%%%%%%%%%
\section{A new look at Kleene's theorem}%
\label{:sec:new-loo}%

Not very many results in computer science are recognised 
as being as basic and fundamental as \emph{Kleene's theorem}.
\index{Kleene's theorem}%
It was originally stated  as the equality of two 
sets of objects, and is still so, even if the names of the objects 
have changed --- see for instance \theor[JEP]{Kleene} in 
\CTchp{JEP} 
\shortlong{of this handbook}{of~\cite{Pin15Ed}}.
This chapter proposes a new look at this statement, in two ways.
First, we explain how Kleene's theorem can be seen as the 
conjunction of \emph{two results} with distinct hypotheses and scopes.
Second, we express the first of these two results as the  
\emph{description of algorithms} that relate the symbolic descriptions 
of the objects rather than as 
the \emph{equality} of two sets.

\ifcheat\smallskipneg\fi%
\paragraph{A two step Kleene's theorem}

In Kleene's theorem, we first distinguish a step that consists in proving 
that the \emph{set of regular} (or \emph{rational}) \emph{languages} 
is equal to the \emph{set of languages accepted by finite automata} 
--- a set which we denote by~$\Rat\Ae$.
This seems already to be Kleene's theorem itself and is indeed what 
S.~C.~Kleene established in~\cite{Klee56hb}.
But it is not, if one considers --- as we shall do here --- that this 
equality merely 
% the consequence of is
states the equality of the expressive power of 
rational expressions and that of finite labelled directed graphs.
This is universally true.
% It does not require~$A$ to be finite nor the monoid to be free.
It holds independently of the structure in which the labels of the 
automata or the atoms of the expressions are taken, in any monoids or 
even in the algebra of polynomials under certain hypotheses.

By the virtue of the numerous properties of finite automata over
finitely generated (f.g., for short) free monoids: being apt to
determinisation for instance, the family of languages accepted by such
automata is endowed with many properties as well: being closed under
complementation for instance.
These properties are extraneous to the definition of the 
languages by expressions, and then --- by the former result --- to 
the definition by automata.
It is then justified, especially in view of the generalisation of 
expressions and automata to other monoids and even to other 
structures, to set up a definition of a new family of languages by
new means, that  
will extend in the case of other structures, these properties of the
languages over f.g. free monoids.
It turns out that the adequate definition will be given in terms of 
\emph{representations by matrices of finite dimension};
we shall call the languages defined in that way the 
\emph{recognisable languages} and we shall denote their family 
by~$\Rec\Ae$.
The second step of Kleene's theorem consists then in establishing that
finite automata are equivalent to matrix representations of finite
dimension under the hypothesis that the labels of automata are taken
in f.g. free monoids.

These two steps correspond to two different concepts: \emph{rationality} 
for the first one, and \emph{recognisability} for the second one.
This chapter focusses on rationality and on the first step, namely the
equivalence of expressiveness of finite automata and rational
expressions.
For sake of completeness however, we sketch in \secti{rat-rec} how 
one gets from rational sets to recognisable sets in the case of free 
monoids and in \secti{cha-mon}, we see that the same construction fails in 
non-free monoids and explore what remains true.

\ifcheat\smallskipneg\fi%
\paragraph{The languages and their representation}

Formal languages or, in the weighted variant, formal power series, 
are potentially \emph{infinite objects}.
We are only able to compute \emph{finite} ones; here, expressions 
that denote, or automata that accept, languages or series.
Hopefully, these expressions and automata are faithful description of
the languages or  series they stand for, all the more effective
that one can take advantage of this double view.

In order to prove that the family of languages accepted by finite 
automata coincide with that of the languages denoted by rational 
expressions we proceed
% of course 
by establishing a double inclusion.
As sketched in \figur{phi-psi-map}\footnote{%
   $\jsPart{\Ae}$ denotes the power set of~$\Ae$, \ie the set of all 
   languages over~$\Ae$.}, 
given an automaton~$\Ac$ that
accepts a language~$K$, we describe algorithms which compute
from~$\Ac$ an expression~$\Fd$ that denotes the same language~$K$ ---
I call such algorithms a \emph{$\AtEs$-map}. 
% --- and this shows a first inclusion.
Conversely, given an expression~$\Ed$ that denotes a language~$L$, we 
describe algorithms that compute from~$\Ed$ an automaton~$\Bc$ that 
accepts the same language~$L$ --- I call such algorithms a \emph{$\EtAs$-map}.

Most of the works devoted to the conversion between automata and 
expressions address the problem 
of the \emph{complexity} of the computation of these $\AtEs$- and $\EtAs$-maps.
I have chosen to study here the maps for themselves, how the results of 
different maps applied to a given argument are related, rather than to 
describe the way they are actually computed.
The $\AtEs$-maps are considered in \secti{aut-exp}, the $\EtAs$-maps 
in \secti{exp-aut}. 

\begin{figure}[htbp]
\begin{center}
\SetNodeSep{4pt}% 
\setlength{\lga}{3.5cm}%
\setlength{\SideLabelDist}{0.22cm}%
\VCDraw[.5]{%
\begin{VCPicture}{(-14,-1)(14,12)}
    \psset{arrows=-,linestyle=solid}%[0]
    \ChgStateLabelScale{1}% 
% ellipses Rat
\VCPut{(0,9)}%
   {%
    \psframe[linewidth=1.2pt](-8.5,-1)(8.5,3)%linecolor=lightgray,
	\psellipse[linecolor=lightgray,linewidth=3pt,%
	           fillstyle=solid,fillcolor=lightgray](0,1)(4.8,1.5)%
    \DotState[160]{(3,1)}{A2}{K}%
    \DotState[140]{(-2,0.5)}{E2}{L}%
    \VCPutLabel[.9]{(-11,0)}{A}{\jsPart{\Ae}}%[1.2]
    \VCPutLabel[.9]{(6.8,0.5)}{C}{\Rat\Ae}
    }%
% ellipse RegE
\VCPut{(0,2)}%
{%
\VCPut{(-\lga,0)}%
   {\VCPut[30]{(0,0)}{\psellipse[linewidth=1.5pt](0,0)(6,3)}%
    \DotState{(-2,0)}{E1}{\Ed}%
    \DotState[135]{(4,2)}{E3}{\Fd}%
    \VCPutLabel[1]{(-2,-3)}{B}{\RatE\Ae}%[1.3]
	}%
% ellipse Aut    
\VCPut{(\lga,0)}%
   {\VCPut[-30]{(0,0)}{\psellipse[linewidth=1.5pt](0,0)(6,3)}%
    \DotState[-60]{(3,0)}{A1}{\Ac}%
    \DotState[-20]{(-1,-2)}{A3}{\Bc}%
    \VCPutLabel[1]{(3,-3)}{C}{\mathsf{Aut}\,\Ae}%
	}%
}%
\ArcR{A1}{A2}{}%
\ArcL{E3}{A2}{}%
\ArcL{E1}{E2}{}%
\ArcR{A3}{E2}{}%
\RstState\RstEdge
\EdgeLineDouble\ChgEdgeLineWidth{2}\ChgEdgeLabelScale{1.4}%
\ArcR[.4]{A1}{E3}{\AtEs}%
\ArcL[.6]{E1}{A3}{\EtAs}%
\RstEdge%
  \end{VCPicture}}%
  \end{center}
  \caption{The $\AtEs$- and $\EtAs$-maps}
  \medskipneg
  \label{:fig:phi-psi-map}%
\end{figure}

\ifcheat\smallskipneg\fi%
\ifcheat\smallskipneg\fi%
\paragraph{The path to generalisation}

The main benefit of splitting 
Kleene's theorem into two steps is to bring to light that the first 
one is a statement whose scope extends much beyond 
languages.
It is first generalised to \emph{subsets} of 
\emph{arbitrary monoids} and then, with some precaution, to 
\emph{subsets with multiplicity}, that is, to (formal power) 
\emph{series}.
This latter extension of the realm of Kleene's theorem is a matter 
for the same `splitting' and distinction between series 
on arbitrary monoids and series on f.g. free monoids.

It would thus be possible to first set up the convenient and most 
general structure and then state and prove Kleene's theorem in that 
framework.
My experience, however, is that many readers tend to be repelled and 
flee when confronted with statements outside the classical realm of 
\emph{words}, \emph{languages}, and \emph{free monoids}.
This is where I stay in the first three sections of this 
chapter.
The only difference with the classical exposition will be in the 
terminology and notation that will be carefully chosen or coined so 
that they will be ready for the generalisation to arbitrary monoids 
in \secti{cha-mon} and to series in \secti{int-mul}.

% \shortonly{%
Notation and definitions given in \CTchp{JEP} are used in this
chapter without comment when they are refered to under the same
form and with the exact same meaning.
% }
% 
% \longonly{%
% Basic definitions...
% }

%%%%%%  AE Section 2  %%%%%%%%%%
%%%%%%%%%%%%%%%%%%%%%%%%%%%%%%%%%%%%%%%%%%%%
\ifcheat\smallskipneg\fi%
\ifcheat\smallskipneg\fi%
\ifcheat\smallskipneg\fi%
\section{Rationality and recognisability}
\label{:sec:rat-rec}%

We first introduce here a precise notion of \emph{rational expression},
and revisit the definition of finite automata in order to fix
our notation and to state, under the form that is studied here
and eventually generalised later, what we have called above the `first
step of Kleene's theorem' and which we now refer to as the
\emph{Fundamental theorem of finite automata}.
\index{Fundamental theorem!of finite automata}%
Second, we state and prove `the second step' of Kleene's theorem
in order to make the scope and 
essence of the first step clearer by contrast and difference. 

\ifcheat\smallskipneg\fi%
\subsection{Rational expressions}
\label{:sec:rat-exp-fm}

\shortlong{%
The set of \emph{rational languages} of~$\Ae$, denoted
by~$\RatA$, is 
defined as in \CTchp{JEP}:
it is the smallest subset 
of~$\jsPart{\Ae}$ which contains the finite sets (including the empty 
set) and is closed under union, product, and star. 
}{%
The set of \emph{rational languages} of~$\Ae$, denoted
by~$\RatA$
is the smallest subset 
of~$\jsPart{\Ae}$ which contains the finite sets (including the empty 
set) and is closed under union, product, and star. 
}%
\index{rational!language}%
\index{language!rational --}%
% The notion of
% \emph{rational expression} 
% allows to describe precisely how every
% element of this family can be built.
A precise structure-revealing specification for building elements of 
this family can be given by \emph{rational expressions}.

\ifcheat\smallskipneg\fi%
\begin{definition}
    \label{:def:rat-exp-fm}%
A \emph{rational expression over~$\Ae$} 
\index{rational!expression|see{expression}}%
\index{expression}%
\index{expression!rational --}%
is a well-formed formula built inductively from the 
\emph{constants}~$\zed$ and~$\und$ 
and the letters~$a$ in~$A$ as \emph{atomic formulas}, 
using
% and with 
two binary operators~$\autplus$ and~$\autprod$ and one unary 
operator~$\autstarsymb$:
if~$\Ed$ and~$\Fd$ are rational expressions, so are
$(\Ed\autplus\Fd)$,
$(\Ed\autprod\Fd)$, and
$(\autstar{\Ed})$.
We denote by~$\RatEA$ the set of rational expressions over~$\Ae$ 
and often write \emph{expression} for \emph{rational expression}. 
(As in~\cite{Saka03}, `rational expression' is preferred to the more 
traditional \emph{regular expression} for several reasons and in particular as 
\index{regular expression|see{expression}}%
\index{expression!regular --}%
it will be used in the weighted case as well, see \secti{int-mul}.) 
\end{definition}

\ifcheat\smallskipneg\fi%
With every expression~$\Ed$ in~$\RatEA$ is associated a 
language of~$\Ae$,
which is called \emph{the language denoted by~$\Ed$} and we 
\index{language!denoted|see{expression}}%
\index{denoted|see{expression}}%
\index{expression!language denoted by --}%
write\footnote{%
   The notation~$L(\Ed)$ is more common,
   but~$\msp\CompExpr{\Ed}\msp$ is simpler
%    lighter in computations 
   and more appropriate when 
   dealing with expressions over an arbitrary monoid or with weighted 
   expressions.}
it as~$\CompExpr{\Ed}\msp$. 
The language~$\CompExpr{\Ed}$ is inductively defined by\footnote{%
   The empty word of~$\Ae$ is denoted by~$\unAe$.} 
$\msp\CompExpr{\zed}=\es\msp$,
$\msp\CompExpr{\und}=\{\unAe\!\}\msp$,
$\msp\CompExpr{a}=\{a\}\msp$ for every~$a$ in~$A$,
$\msp\CompExpr{(\Ed\autplus\Fd)}= {}\CompExpr{\Ed}\cup\xmd\CompExpr{\Fd}\msp$,
$\msp\CompExpr{(\Ed\autprod\Fd)}={}\CompExpr{\Ed}\msp\CompExpr{\Fd}\msp$, and
$\msp\CompExpr{(\autstar{\Ed})}=\{\CompExpr{\Ed}\}^{*}\msp$.
Two expressions are \emph{equivalent} if they denote the same 
\index{equivalent|see{expression}}%
\index{expression!equivalent --s}%
language. 

\ifcheat\smallskipneg\fi%
\begin{proposition}
    \label{:pro:rat-exp-fm}%
A language is rational if and only if it is denoted 
by an expression.
\end{proposition}

Like any formula, an expression~$\Ed$ is canonically
represented by a tree, which is called \emph{the syntactic tree} of~$\Ed$.
Let us denote by~$\LiteLgth{\Ed}$ the \emph{literal length} of the
expression~$\Ed$ (\ie the number of all occurences of letters from $A$
in~$\Ed$)
and by~$\Depth{\Ed}$ the \emph{depth} of~$\Ed$ which is
defined as the depth --- or height\footnote{%
   We rather not use \emph{height} because of the possible confusion 
   with the \emph{star height}, \cf \secti{exp-aut}.}%
--- of the syntactic tree of the expression.
\index{literal length|see{expression}}%
\index{expression!literal length of --}%
\index{depth|see{expression}}%
\index{expression!depth of --}%

The classical precedence relation between operators: 
`$\msp\autstarsymb\xmd>\xmd\autprod\xmd>\xmd\autplus\msp$'
allows to save parentheses in the writing of expressions:
for instance,
$\msp\Ed\autplus\Fd\autprod\autstar{\Gd}\msp$
is an unambiguous writing for the expression
$\msp(\Ed\autplus(\Fd\autprod(\autstar{\Gd})))\msp$.
% ---
But one should be aware that, for instance,
$\msp(\Ed\autprod(\Fd\autprod\Gd))\msp$ and
$\msp((\Ed\autprod\Fd)\autprod\Gd)\msp$ are two equivalent
\emph{but distinct} expressions.
In particular, the \emph{derivation} that we define 
at \secti{exp-aut} yields different results on these two 
expressions. 

In the sequel, any operator defined on expressions is implicitely 
extended additively to sets of expressions.
For instance, it holds:
\begin{equation}
\textstyle{\fa X \subseteq \RatEA \quantsp
\CompExpr{X}=\bigcup_{\Ed\in X}{}\CompExpr{\Ed}
\eqpnt}
\ee\eee
\notag
\end{equation}

\begin{definition}
\label{:def:con-ter-exp}%
The \emph{constant term} of an expression~$\Ed$ over~$\Ae$, 
\index{constant term!of an expression}%
\index{expression!constant term of --}%
written $\TermCst{\Ed}$, is the Boolean value,
% effectively computable by a bottom-up traversal of the syntactic 
% tree of~$\Ed$ 
inductively defined and computed
using the following equations:

\begin{gather}
\TermCst{\zed} = 0 \EqVrgInt
\TermCst{\und} = 1 \EqVrgInt  
\forall a \in A\quantsmsp \TermCst{a} =0 
\eqvrg 
\notag
\\
\TermCst{\Fd \autplus \Gd}= \TermCst{\Fd} + \TermCst{\Gd}\EqVrgInt
\TermCst{\Fd \autprod \Gd} = \TermCst{\Fd}  \TermCst{\Gd}\EqVrgInt
\TermCst{\autstar{\Fd}}=1
\eqpnt
\notag
\end{gather}
\end{definition}

The \emph{constant term} of a language~$L$ of~$\Ae$ is the 
Boolean value~$\TermCst{L}$ that is equal to~$\und$
if and only if~$\msp\unAe$ belongs to~$L$.
By induction on~$\Depth{\Ed}$, 
\index{constant term!of a language}%
\index{language!constant term of --}%
$\msp \TermCst{\Ed}=\TermCst{\CompExpr{\Ed}}\msp$ holds.

\ifcheat\smallskipneg\fi%
\ifcheat\miniskipneg\fi%
\subsection{Finite automata}
\label{:sec:fin-aut-fm}%

We denote an \emph{automaton over~$\Ae$} by 
$\msp \Ac= \aut{Q,A,E,I,T}\msp$
where\shortonly{, as in \CTchp{JEP},}
$Q$ is the \emph{set of states}, and is also called the \emph{dimension} 
of~$\Ac$, $I$ and $T$ are subsets of~$Q$, and 
$\msp E\subseteq Q\x A\x Q\msp$ is 
the set of transitions labelled by letters of~$A$.
The automaton~$\Ac$ is \emph{finite} if~$E$ is finite,
\index{automaton}%
\index{automaton!dimension of --}%
hence, if~$A$ is finite, if and only if (the useful part of)~$Q$ 
is finite.

A \emph{computation} in~$\Ac$ from state~$p$ to state~$q$ with 
label~$w$ is denoted by
$\msp p\pathaut{w}{\Ac}q\msp$.
The \emph{language accepted}\footnote{%
   I prefer not to speak of the language `recognised' by an 
   automaton, and\longshort{ I}{, in contrast with \CTchp{JEP}, I}
   would not say that a language is 
   `recognisable' when accepted by a finite automaton, in
   order to have a consistent terminology when generalising 
   automata to arbitrary monoids.}
by~$\Ac$, also called the 
\emph{behaviour} of~$\Ac$, denoted by~$\CompAuto{\Ac}$, is the set of 
\index{language!accepted|see{automaton}}%
\index{accepted|see{automaton}}%
\index{automaton!language accepted by --}%
\index{language!recognised|see{automaton}}%
\index{recognised|see{automaton}}%
\index{automaton!language recognised by --}%
\index{behaviour|see{automaton}}%
\index{automaton!behaviour of --}%
words accepted by~$\Ac$, that is, the set of labels of 
\emph{successful computations}: 
\shortlong{\PushLine
$\msp \CompAuto{\Ac}= 
\Defi{w\in\Ae}{\ext i\in I,\ext t\in T\quantsmsp i\pathaut{w}{\Ac}t}\msp$.
\PushLine}%
{\begin{equation}
\CompAuto{\Ac}= 
\Defi{w\in\Ae}{\ext i\in I,\ext t\in T\quantsmsp i\pathaut{w}{\Ac}t}
\eqpnt
\notag
\end{equation}}

\noindent 
The first step of Kleene's theorem, which we call \emph{Fundamental 
\index{Fundamental theorem!of finite automata}%
theorem of finite automata} then reads as follows.
   
\ifcheat\smallskipneg\fi%
\begin{theorem}
    \label{:the:fun-aut-fm}% 
A language of~$\Ae$ is rational if and only if it is 
the behaviour of a finite automaton over~$\Ae$.
\end{theorem}

\ifcheat\smallskipneg\fi%
\theor{fun-aut-fm} is 
proved by building connections between automata and expressions.

\ifcheat\smallskipneg\fi%
\begin{proposition}[$\AtEs$-maps]
    \label{:pro:phi-map-fm}%
For every finite automaton~$\Ac$ over~$\Ae$, there exist rational 
expressions over~$\Ae$ which denote~$\CompAuto{\Ac}$.
\end{proposition}

\ifcheat\smallskipneg\fi%
\ifcheat\smallskipneg\fi%
\ifcheat\smallskipneg\fi%
\begin{proposition}[$\EtAs$-maps]
    \label{:pro:psi-map-fm}%
For every rational expression~$\Ed$ over~$\Ae$, there exist 
finite automata over~$\Ae$ whose behaviour is equal to~$\CompExpr{\Ed}$.
\end{proposition}

\ifcheat\smallskipneg\fi%
\secti{aut-exp} describes how expressions are computed from automata,
\secti{exp-aut} how automata are associated with expressions.
Before going to this matter, which is the main subject of this chapter, 
let us establish the second step of Kleene's theorem.

\ifcheat\smallskipneg\fi%
\subsection{The `second step' of Kleene's theorem}
\label{:ssc:sec-ste}%

Let us first state the definition of recognisable languages, under the 
form that is given for recognisable subsets of arbitrary monoids (\cf 
\secti[JEP]{rec-sub}). 

\begin{definition}
    \label{:def:rec-lan}%
A language~$L$ of~$\Ae$ is % said to be 
\emph{recognised} by a 
morphism~$\alpha$ from~$\Ae$ into a monoid~$N$
% $\msp\alpha\colon\Ae\rightarrow N\msp$
if $\msp L= \alpha^{-1}(\alpha(L))\msp$.
A language % of~$\Ae$ 
is \emph{recognisable} if it is recognised by a morphism 
% from~$\Ae$ 
into a \emph{finite} monoid.
The set of \emph{recognisable languages} of~$\Ae$ is
\index{recognisable!language}%
\index{language!recognisable --}%
denoted by~$\RecA$.
\end{definition}

\ifcheat\smallskipneg\fi%
\begin{theorem}[Kleene]
\label{:the:kle-fm}%
If~$A$ is a finite alphabet, 
\index{Kleene's theorem}%
then $\msp\Rat\Ae=\Rec\Ae$. %\msp
\end{theorem}

\ifcheat\smallskipneg\fi%
The proof of this statement paves the way to 
further developments in this chapter.
Let
$\msp \Ac= \aut{Q,A,E,I,T}\msp$
be a finite automaton.
The set~$E$ of transitions
% $\msp E\subseteq Q\x A\x Q\msp$
may be written as a $Q\x Q$-matrix, called the 
\emph{transition matrix} of~$\Ac$, also 
\index{transition!matrix}%
\index{matrix!transition --}%
denoted by~$E$, and whose $(p,q)$-entry is the set (the Boolean sum) 
of letters that label  
the transitions from~$p$ to~$q$ in~$\Ac$.
A fundamental (and well-known) lemma relates matrix multiplication 
and graph walking.

\begin{lemma}
\label{:lem:mat-mul-gra}%
Let~$E$ be the transition matrix of the 
automaton~$\Ac$ of finite dimension~$Q$.
Then, for every~$n$ in~$\N$, 
% the $(p,q)$-entry of~$E^{n}$ is the 
% set of labels of paths of length~$n$ in~$\Ac$:
$E^{n}$ is the matrix of the labels of paths of length~$n$ in~$\Ac$:
\begin{equation}
\textstyle{E^{n}_{p,q} = \Defi{w\in A^{n}}{p\pathaut{w}{\Ac}q}}
\eqpnt
\notag
\end{equation}
\end{lemma}

The subsets~$I$ and~$T$ of~$Q$ may then be seen as Boolean vectors of 
dimension~$Q$ ($I$ as a row and~$T$ as a column-vector). 
From the notation
$\msp\longonly{\displaystyle}{E^{*} = \sum_{n\in\N}E^{n}}\msp$,
it follows: 
% then the equation:
\begin{equation}
\CompAuto{\Ac} = I\matmul E^{*}\matmul T
\eqpnt
% \notag
\label{:equ:com-aut-mat}
\end{equation}
The next step in the preparation of the proof of 
\theor{kle-fm} is to write the transition matrix~$E$ as
a formal sum
$\msp\longonly{\displaystyle}{E = \sum_{a\in A} \mu(a)\xmd a}\msp$,
where for every~$a$ in~$A$, $\mu(a)$ is a Boolean $Q\x Q$-matrix.
These matrices~$\mu(a)$ define a morphism
$\msp\mu\colon\Ae\rightarrow \B^{Q\x Q}\msp$\longshort{.}{
(the Boolean semiring~$\B$ has been defined at~\CTchp{JEP}).}
The second lemma involves the \emph{freeness} of~$\Ae$ and  
reads:

\begin{lemma}
\label{:lem:mat-mul-rep}%
Let
$\msp\mu\colon\Ae\rightarrow \B^{Q\x Q}\msp$
be a morphism and let
$\msp\longonly{\displaystyle}{E = \sum_{a\in A} \mu(a)\xmd a}\msp$.
Then, for every~$n$ in~$\N$,
$\msp E^{n} = \sum_{w\in A^{n}} \mu(w)\xmd w\msp$
and thus 
$\msp E^{*} = \sum_{w\in A^{*}} \mu(w)\xmd w\msp$.
\end{lemma}

\begin{proof}[Proof of \theor{kle-fm}]
By \theor{fun-aut-fm}, a rational language~$L$ of~$\Ae$
is the behaviour of a finite automaton
$\msp \Ac= \aut{Q,A,E,I,T}\msp$.
By~\equnm{com-aut-mat} and \lemme{mat-mul-gra}, we write
\begin{equation}
L= {}\CompAuto{\Ac}{} =  \Defi{w\in A^{*}}{ I\matmul\mu(w)\matmul T=1}
\eqpnt
\notag
%     \label{:equ:}
\end{equation}
and thus
$\msp\displaystyle{L = \mu^{-1}(S)}\msp$
where
$\msp\displaystyle{S = \Defi{m\in\B^{Q\x Q}}{ I\matmul m\matmul T=1}}\msp$
and~$L$ is recognisable.

Conversely, let~$L$ be a recognisable language of~$\Ae$, recognised 
by the morphism
$\msp\alpha\colon\Ae\rightarrow N\msp$ and let
$\msp S = \alpha(L)\msp$.
Consider the automaton
$\msp \Ac_{\alpha}= \aut{N,A,E,\{\unN\},S}\msp$
where
$\msp\displaystyle{E = \Defi{\Tran{n,a,n\xmd\alpha(a)}}{a\in A, n\in N}}\msp$.
It is immediate that
\begin{equation}
\CompAuto{\Ac_{\alpha}} = 
  \Defi{w\in A^{*}}{\ext p\in S\quantsmsp \unN\pathaut{w}{\Ac}p}
  = \Defi{w\in A^{*}}{\alpha(w)\in S}
  = \alpha^{-1}(S) = L
% \eqpnt
\notag
%     \label{:equ:}
\end{equation}
and~$L$ is rational by \theor{fun-aut-fm}.
\end{proof}

We postpone to \secti{cha-mon} 
the example that shows that recognisability and 
rationality are indeed two distinct concepts and the description of 
the relationships that can be found between them.
As mentioned in \CTchp{JEP}, the following holds.

\begin{theorem}
    \label{:the:dec-equ-aut}%
The equivalence of finite automata over~$\Ae$ is decidable.
\end{theorem}

\propo{psi-map-fm} then implies:

\begin{corollary}
    \label{:cor:dec-equ-exp}%
The equivalence of rational expressions over~$\Ae$ is decidable.
\end{corollary}

%%%%%%  AE Section 3  %%%%%%%%%%
%%%%%%%%%%%%%%%%%%%%%%%%%%%%%%%%%%%%%%%%%%%%
\section{From automata to expressions: the $\AtEs$-maps}
\label{:sec:aut-exp}%

For the rest of this section, 
$\msp \Ac= \aut{Q,A,E,I,T}\msp$
is a finite automaton over~$\Ae$, and~$E$ is viewed, 
depending on the context, as the \emph{set of transitions}
or as the \emph{transition matrix} of~$\Ac$. 
As in~\equnm{com-aut-mat}, the language accepted by~$\Ac$ is 
conveniently written as
\begin{equation}
\CompAuto{\Ac} = I \matmul E^{*} \matmul T
        = \textstyle{\bigcup_{i\in I, t\in T}} \left(E^{*}\right)_{i,t}
\eqpnt
\notag
\end{equation}
In order to prove that~$\CompAuto{\Ac}$ 
is rational, it is sufficient to establish the following.

\ifcheat\smallskipneg\fi%
\begin{proposition}
\label{:pro:sta-mat}%~$\Ed$
The entries of~$E^{*}$ belong to the rational closure of the 
\index{rational!closure}%
entries of~$E$.
\end{proposition}

\ifcheat\smallskipneg\fi%
But we want to be more precise and describe 
procedures that produce 
for every entry of~$E^{*}$ a rational expression whose atoms are the 
entries of~$E$ (and possibly~$\und$).
There are (at least) four classical methods to proving 
\propo{sta-mat},
which can easily be viewed as algorithms serving our purpose
and which we present here:

\begin{conditions}

\item  Direct computation of~$\CompAuto{\Ac}\xmd$:
the \textit{state-elimination method} 
\index{state-elimination|see{method}}%
\index{method!state-elimination}%
looks the most elementary  
and is indeed the easiest for both hand computation 
and computer implementation.

\item  Computation of $E^{*}\matmul T$ as a 
\index{system-solution|see{method}}%
\index{method!system-solution}%
    solution of a system of linear equations.
Based on Arden's lemma, it also allows to consider $E^{*}\matmul T$ 
as a fixed point. 
    
\item  Iterative computation of $E^{*}$:
known as \textit{McNaughton--Yamada algorithm} 
\index{McNaughton--Yamada algorithm}%
\index{algorithm|see{McNaughton--Yamada}}%
and probably the most 
popular among textbooks on automata theory.

\item  Recursive computation of $E^{*}$:
\index{recursive|see{method}}%
\index{method!recursive}%
based on Arden's lemma as well, this algorithm
combines mathematical elegance and 
computational inefficiency.
    
\end{conditions}

The first three are based on an 
ordering of the states of the automaton.
For \emph{comparing} the results of these different algorithms, 
and of a given one when the ordering of states varies,  
we first introduce the notion of 
\emph{rational identities}, together with 
the key Arden's lemma for 
establishing the correctness of the algorithms as well as the 
identities.
The section ends with a refinement of 
\theor{fun-aut-fm}
which, by means of the notions of \emph{star height} 
and \emph{loop complexity}, relates even more closely an automaton 
and the rational expressions that are computed from it.

%%%%%%%%%%%%%%%%%%%%%%
\ifcheat\smallskipneg\fi%
\ifcheat\smallskipneg\fi%
\subsection{Preparation: rational identities and Arden's lemma}
\label{:sec:Ard-rat-ide}

\ifcheat\smallskipneg\fi%
By definition, all expressions which
% are computed from 
denote the behaviour of
a given automaton~$\Ac$ are \emph{equivalent}.
We may then ask whether, and how, this equivalence may be established 
\emph{within the world of expressions itself}.
We consider `elementary equivalences' of more or less simple 
expressions, which we call \emph{rational identities},  
or \emph{identities} for short, and which correspond to properties of 
(the semiring of) the languages denoted by the 
expressions.
\index{rational!identities|see{identities}}%
\index{identities!rational --}%
And we try to determine which of these identities, considered as 
\emph{axioms}, are necessary, or sufficient, to obtain by substitution 
one expression from another equivalent one.
It is known --- and out of the scope of this chapter --- that no 
finite sets of identities exist that allow to establish the 
equivalence of expressions in general (see \CTchp{ZE}).
We shall see however that a \emph{basic set of identities} is sufficient to 
deduce the equivalence between the expressions computed by the 
different $\AtEs$-maps described here.

\ifcheat\smallskipneg\fi%
\paragraph{Trivial and natural identities}
A first set of  identities, that we call \emph{trivial 
identities}, expresses the fact that~$\zed$ and~$\und$ are 
\index{identities!trivial --}%
interpreted as the zero and unit of a semiring:
\begin{equation}
\Ed \autplus \zed \equiv \Ed \EqVrgInt \zed \autplus \Ed \equiv \Ed \EqVrgInt
\Ed \autprod \zed \equiv \zed \EqVrgInt \zed \autprod \Ed \equiv \zed 
\EqVrgInt
\Ed \autprod \und \equiv \Ed \EqVrgInt \und \autprod \Ed \equiv \Ed  \EqVrgInt
\autstar \zed \equiv \und \e
\eqpnt
\tag{$\Tmbf$}
\label{:equ:tri-ide}
\end{equation}
An expression is said to be \emph{reduced} if it contains no 
subexpressions which is a left-hand side of one of the above 
identities; in particular, $\zed$ does not appear in a non-zero 
reduced expression.
\index{reduced expression|see{expression}}%
\index{expression!reduced --}%
Any expression~$\Hd$ can be rewritten in an 
equivalent reduced expression~$\Hd'$; this~$\Hd'$ is unique 
and independent of the way the rewriting is conducted.
From now on, all
% All 
expressions are implicitely reduced, which means that 
\emph{all the computations on expressions that will  
be defined below are performed modulo the trivial identities}.

The next set of identities expresses the fact that the 
operators~$\plusopr$ and~$\prodopr$ are interpreted as the 
\emph{addition} and \emph{product} in a semiring\longonly{ with their 
associativity, distributivity and commutativity properties}:
\begin{align}
\ \ \ \ (\Ed+\Fd)+\Gd \equiv \Ed+(\Fd+\Gd) \e \e &\text{and}\e  \e
\ \ (\Ed \cdot \Fd) \cdot \Gd \equiv \Ed \cdot (\Fd \cdot \Gd) \eqvrg \ee \ee
\tag{$\IdRAs$} 
\label{:equ:ass-ide}
\\
\Ed \cdot (\Fd + \Gd)  \equiv \Ed \cdot \Fd + \Ed \cdot \Gd \e \e &\text{and}\e  \e
(\Ed + \Fd) \cdot \Gd  \equiv \Ed \cdot \Gd + \Fd \cdot \Gd \eqvrg \ee \e
\tag{$\IdRD$} 
\label{:equ:dis-ide}
\\
\Ed + \Fd \, &\equiv \, \Fd + \Ed \eqpnt \ee \ee
\tag{$\IdRC$}
\label{:equ:com-ide}
\end{align}
The conjunction $\IdRAs\land\IdRD\land\IdRC$
\index{identities!natural --}%
is abbreviated as~$\IdRN$ and called the set of \emph{natural identities}.

\ifcheat\smallskipneg\fi%
\paragraph{Aperiodic identities}
The product in~$\jsPart{\Ae}$ is  \emph{distributive} over {infinite 
sums}; then
\begin{equation}
    \fa K\in\jsPart{\Ae}\quantsp
    K^{*} = \unAe + K^{*} K = \unAe + K\xmd K^{*}
    \eqvrg
    \eee
\label{:equ:ide-U-lan}
\end{equation}
from which we deduce the identities:
\begin{equation}
\Ed^{*} \equiv  \und \autplus \Ed \autprod \Ed^{*}
\ee \text{and} \ee
\Ed^{*} \equiv \und \autplus \Ed^{*}\autprod \Ed 
\eqpnt
\tag{$\IdRU$}
\label{:equ:uni-ide}
\end{equation}
From~\equnm{uni-ide} and the gradation\footnote{%
   That is, the elements of~$\Ae$ have a \emph{length} which is a 
   morphism from~$\Ae$ onto~$\N$ (\cf \secti{int-mul}).}
of~$\Ae$ follows
\index{monoid!graded --}%
% the property known as 
Arden's lemma whose usage is ubiquitous. 

\ifcheat\smallskipneg\fi%
\begin{lemma}[Arden]
    \label{:lem:Ard}%
Let~$K$ and~$L$ be two subsets of~$\Ae $.
Then $\msp K^{*}L\msp $ is a solution,
$\msp K^{*}L\msp$~is \emph{the unique solution} if 
 $\msp\TermCst{K}=0\msp$,
of the equation
$\msp\mathrm{X} = K \xmd \mathrm{X} + L $. 
% %\msp
% If then  
\index{Arden's lemma}%
% \index{Lemma!Arden's --}%
% .
\end{lemma}

For computing \emph{expressions}, we prefer to
use Arden's lemma under the following form:

\ifcheat\smallskipneg\fi%
\begin{corollary}
\label{:cor:Ard}%
Let~$\Kd$ and~$\Ld$ be two rational expressions over~$\Ae$ 
with~$\TermCst{\Kd}=\zed$.
Then, 
% $\msp\CompExpr{\Kd^{*}\Ld}\msp$
$\msp\Kd^{*}\Ld\msp$ %\CompExpr{}
denotes the unique solution of 
$\msp\mathrm{X}=\CompExpr{\Kd}\xmd\xmd\mathrm{X}+\CompExpr{\Ld}\msp$.
\end{corollary}

The next two identities, called \emph{aperiodic identities},
\index{identities!aperiodic --}%
\index{aperiodic|see{identities}}%
are a consequence of \lemme{Ard}.%

\ifcheat\smallskipneg\fi%
\begin{proposition}
    \label{:pro:ide-ape}%
For all rational expressions~$\Ed$ and~$\Fd$ over~$\Ae$
\begin{align}
(\Ed + \Fd)^{*} \equiv   \Ed^{*}\cdot (\Fd\cdot \Ed^{*})^{*}
\ee &\text{and} \ee
(\Ed + \Fd)^{*} \equiv  (\Ed^{*}\cdot \Fd)^{*}\cdot \Ed^{*}
\eqvrg
\tag{$\IdRS$}
\label{:equ:sum-sta-ide}
\\
(\Ed \cdot \Fd)^{*} &
\equiv   1 + \Ed \cdot (\Fd \cdot \Ed)^{*} \cdot \Fd 
\eqpnt
\tag{$\IdRP$}
\label{:equ:pro-sta-ide}
\end{align}
\end{proposition}

There are many other (independent) identities (\cf Notes). 
The remarquable fact is that those listed above will 
be sufficient for our purpose.

\ifcheat\smallskipneg\fi%
\paragraph{Identities special to~$\PdAe$}

Finally, the \emph{idempotency} of the union in~$\jsPart{\Ae}$ yields two further 
identities:

\smallskipneg\smallskipneg
\TextText{%
\begin{equation}
\Ed \autplus \Ed  \equiv  \Ed
\eqvrg
\tag{$\IdRI$}
\label{:equ:ide-ide}
\end{equation}
}{%
\begin{equation}
\left(\Ed^{*}\right)^{*}  \equiv  \Ed^{*}
\eqpnt
\tag{$\IdRJ$}
\label{:equ:sta-sta-ide}
\end{equation}
}
\miniskip 

In contrast with the preceding ones, these two identities 
\equnm{ide-ide} and \equnm{sta-sta-ide} do not hold for expressions 
over arbitrary semirings of formal power series 
(\cf~\secti{int-mul}). 

% \ifcheat\smallskipneg\fi%
% \ifcheat\smallskipneg\fi%
%%%%%%%%%%%%%%%%%%%%%%
\subsection{The state-elimination method}
\label{:ssc:sta-eli-met}%

The algorithm known as 
\emph{state-elimination method},
\index{method!state-elimination}%
originally due to Brzozowski and  McCluskey \cite{BrzoMClu63},
works directly on the automaton~$\Ac=\auta$.
It consists in suppressing the states in~$\Ac$, 
one after the other, while transforming the labels 
of the transitions so that the language accepted by the resulting 
automaton is unchanged (\cf \cite{Wood87,Yu97hb}). 

A current step of the algorithm is represented at 
\figur{sta-eli-ste}.
The left diagram shows the state~$q$ to be 
suppressed, a state~$p_{i}$ which is the origin of a transition whose 
end is~$q$ and a state~$r_{j}$ which is the end of a transition whose 
origin is~$q$ (it may  be the case that~$p_{i}=r_{j}$). 
By induction, the labels 
are \emph{rational expressions}.
The right diagram shows the automaton after the suppression of~$q$, 
and the new label of the transition from~$p_{i}$ to~$r_{j}$.
The languages accepted by the automaton before and after the suppression 
of~$q$ are equal --- a formal proof will follow in the next subsection.

\begin{figure}[htbp]
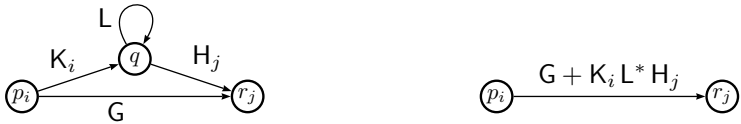

% 	\medskipneg 
\centering
\VCDraw{%
\begin{VCPicture}{(-3.5,-0.5)(3.5,2.5)}
\MediumState 
% etats
\State[q]{(0,1)}{A}
\State[p_{i}]{(-3,0)}{B}\State[r_{j}]{(3,0)}{C}
% transitions  [.5]
\ChgEdgeLabelScale{1.2}
\EdgeL[.4]{B}{A}{\Kd_{i}}\EdgeL[.6]{A}{C}{\Hd_{j}}
\LoopN[.2]{A}{\Ld}
\EdgeR[.4]{B}{C}{\Gd}
\end{VCPicture}}%
\ee\ee\ee\ee
\VCDraw{%
\begin{VCPicture}{(-3.5,-0.5)(3.5,2.5)}
\MediumState 
% etats
% \State[q]{(0,1)}{A}
\State[p_{i}]{(-3,0)}{B}\State[r_{j}]{(3,0)}{C}
% transitions  
\ChgEdgeLabelScale{1.2}
% \EdgeL[.5]{B}{C}{\Gd + \Kd_{i}\xmd (\Ld^{*}\xmd \Hd_{j})}
\EdgeL[.5]{B}{C}{\Gd + \Kd_{i}\xmd \Ld^{*}\xmd \Hd_{j}}
\end{VCPicture}}%
\caption{One step in the state-elimination method}
\medskipneg 
\label{:fig:sta-eli-ste}
\end{figure}

More precisely, the state-elimination method consists 
first in augmenting the set~$Q$ with two new states~$i$ 
and~$t$, and adding transitions labelled with~$\und$ from~$i$ to 
every initial state of~$\Ac$ and from every final state of~$\Ac$ 
to~$t$.
Then all states in~$Q$ are suppressed according to the procedure 
described above and in a certain order~$\omega$\longonly{%
(that can be decided beforehand or determined step by step)}.
At the end, only remain states~$i$ and~$t$, together with a 
transition from~$i$ to~$t$ labelled with an expression which we 
denote by~$\sem{\Ac}$ and which is the \emph{result} of the 
algorithm. 
Thus it holds:
% of states of~$\Ac$we have
\begin{equation}
    \CompAuto{\Ac} = \CompExpr{\sem{\Ac}}
    \eqpnt
\notag
\end{equation}
\figur{sta-eli-exa} shows every step of the state-elimination method 
on the automaton~$\Dc_{3}$ drawn in the upper left corner and following 
the order
$\msp\omega_{1}= r<p<q \msp$.
It shows the result 
$\msp
\sem[\omega_{1}]{\Dc_{3}}=a^{*}b\xmd (b\xmd a^{*}b 
	     + a\xmd b^{*}a)^{*}b\xmd a^{*}
	     + a^{*}
\msp$.
% One should be aware that the natural computation of~$\sem{\Ac}$ may 
% silently involve the identities~$\IdRAs$, $\IdRC$, and~$\IdRI$ as 
% well. 
The computation of~$\sem{\Ac}$ may silently involve identities 
in~$\IdRN$. 
A common and natural way of performing the computation is to use 
identities~$\IdRI$ and~$\IdRJ$ as well: it yields simpler results. 
It is then to be stressed that
the use of~$\IdRI$ and~$\IdRJ$
is not needed to establish these equivalence results.
% the  equivalence results that follow hold without using~$\IdRI$ 
% and~$\IdRJ$. 

\ifcheat\smallskipneg\fi%
%%%%%%%%%%%%%%%%%%%%%%
\paragraph{The effect of the order}

The result of the state-elimination method obviously depends on the 
order~$\omega$ in which the states are suppressed.
For instance, on the automaton~$\Dc_{3}$ of \figur{sta-eli-exa}, 
the other order
% $\msp\omega_{1} = r<p<q\msp$ yields
% $\msp\sem[\omega_{1}]{\Dc_{3}}  = a^{*}b\xmd (b\xmd a^{*}b 
% 	     + a\xmd b^{*}a)^{*}b\xmd a^{*}
% 	     + a^{*}$,
$\msp\omega_{2} = r<q<p\msp$ yields
$\msp\sem[\omega_{2}]{\Dc_{3}}  = (a +b\xmd(a\xmd b^*a)^*b)^*$,
and $\msp\omega_{3} = p<q<r\msp$ yields\\
\e$\msp
\sem[\omega_{3}]{\Dc_{3}}  = a^{*} + a^*b\xmd (b\xmd a^*b)^*b\xmd a^*
               + a^*b\xmd(b\xmd a^*b)^*a\xmd
    (b + a\xmd(b\xmd a^*b)^*a)^*a\xmd(b\xmd a^*b)^*b\xmd a^*$.

% All these expressions are equivalent modulo the aperiodic identities:

\setlength{\lga}{2.6cm}\setlength{\lgb}{2.25cm}\setlength{\lgc}{1.75cm}%
\begin{figure}[htbp]
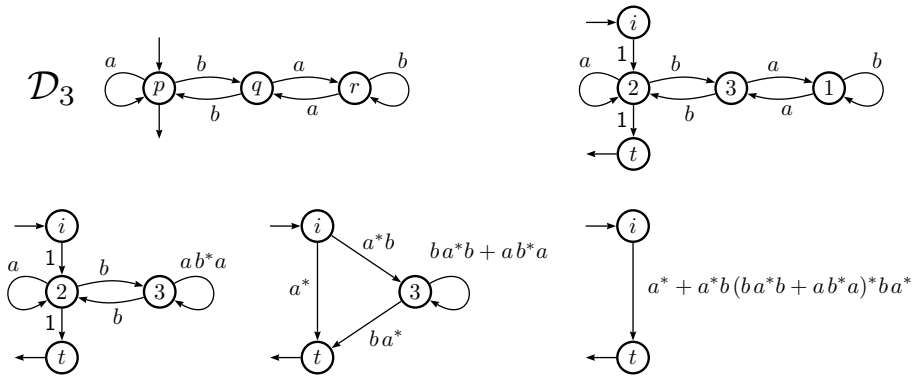

\centering
\VCDraw{%
\begin{VCPicture}{(-3.5,-\lgb)(6.7,\lgb)}
\MediumState
\State[p]{(0,0)}{A}
\State[q]{(\lga,0)}{B}
\State[r]{(2\lga,0)}{C}
\Initial[n]{A}\Final[s]{A}
\ArcL{A}{B}{b}\ArcL{B}{A}{b}
\ArcL{B}{C}{a}\ArcL{C}{B}{a}
\LoopW{A}{a}\LoopE{C}{b}
\VCPutText[1.8]{(-2.8,0)}{\Dc_{3}}
\end{VCPicture}
}%
\eee\ee
\VCDraw{%
\begin{VCPicture}{(-1.6,-\lgb)(6.7,\lgb)}
\MediumState
\SetAOSLengthCoef{1.4}%
\State[i]{(0,\lgc)}{I}\State[t]{(0,-\lgc)}{T}
\State[2]{(0,0)}{A}
\State[3]{(\lga,0)}{B}
\State[1]{(2\lga,0)}{C}
\Initial{I}\Final[w]{T}
\ArcL{A}{B}{b}\ArcL{B}{A}{b}
\ArcL{B}{C}{a}\ArcL{C}{B}{a}
\LoopW{A}{a}\LoopE{C}{b}
% \ChgEdgeLabelScale{0.8}%[.7][.3]
\EdgeR{I}{A}{\und}\EdgeR{A}{T}{\und}
\end{VCPicture}
}%

\shortlong{\medskip}{\bigskip}

\VCDraw{%
\begin{VCPicture}{(-1.6,-\lgb)(4.7,\lgb)}
% etats
\MediumState
\SetAOSLengthCoef{1.4}%
\State[i]{(0,\lgc)}{I}\State[t]{(0,-\lgc)}{T}
\State[2]{(0,0)}{A}
\State[3]{(\lga,0)}{B}
\Initial{I}\Final[w]{T}
\ArcL{A}{B}{b}\ArcL{B}{A}{b}
\LoopW{A}{a}%[.3]
% \ChgEdgeLabelScale{0.9}
\LoopE[.2]{B}{a\xmd b^{*}a}
% \ChgEdgeLabelScale{0.8}
\EdgeR{I}{A}{\und}\EdgeR{A}{T}{\und}
\end{VCPicture}
}%
\e
\VCDraw{%
\begin{VCPicture}{(-1.3,-\lgb)(6.3,\lgb)}
% etats
\MediumState
\SetAOSLengthCoef{1.4}%
\State[i]{(0,\lgc)}{I}\State[t]{(0,-\lgc)}{T}
\State[3]{(\lga,0)}{B}
\Initial{I}\Final[w]{T}
% transitions 
\EdgeR{I}{T}{a^{*}}
\EdgeL{I}{B}{a^{*}b}\EdgeL{B}{T}{b\xmd a^{*}}
%
% \ChgEdgeLabelScale{0.9}
\LoopE[.25]{B}{b\xmd a^{*}b+a\xmd b^{*}a}
\end{VCPicture}
}%
\e
\VCDraw{%
\begin{VCPicture}{(-1.3,-\lgb)(7.4,\lgb)}
% etats
\MediumState
\SetAOSLengthCoef{1.4}%
\State[i]{(0,\lgc)}{I}\State[t]{(0,-\lgc)}{T}
\Initial{I}\Final[w]{T}
% transitions 
% \ChgEdgeLabelScale{0.9}
\EdgeL{I}{T}{\msp\msp a^{*} + a^{*}b\xmd (b\xmd a^{*}b 
	               + a\xmd b^{*}a)^{*}b\xmd a^{*}}
\end{VCPicture}
}%
\caption{The state-elimination method exemplified on the 
automaton~$\Dc_{3}$}
\label{:fig:sta-eli-exa}
\end{figure}

\ifcheat\smallskipneg\fi%
\begin{theorem}[Conway \cite{Conw71}, Krob \cite{Krob91}]
    \label{:the:ord-eli-met}%\land \IdRI
Let~$\omega $ and~$\omega' $ be two orders 
on the set of states of an automaton~$\Ac$.
Then,
$\msp\msp
 \IdRN \land \IdRS \land \IdRP \msp\msp \dedtxt 
 \sem{\Ac} \msp \equiv \msp \sem[\omega']{\Ac}
\msp\msp$ holds.
\end{theorem}

\longonly{%
\begin{proof}
We can go from any order~$\omega$ to any other order~$\omega'$ on~$Q$ 
by a sequence of transpositions. 
We therefore arrive at the situation illustrated in 
\figur{MNY-ord} (left)
and need to show that the expressions obtained 
when we first remove the state~$r$ and then~$r'$
are equivalent to those obtained from removing first~$r'$ and then~$r$,
modulo~$\IdRS \land \IdRM$ (without mentioning the \emph{natural 
identities}). 

\setlength{\lgb}{2cm}
\begin{figure}[htbp]
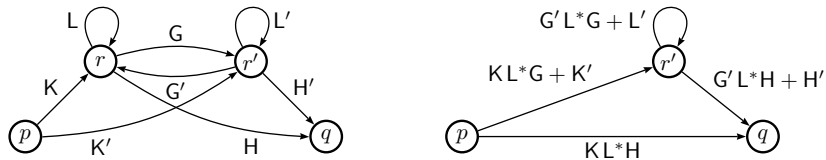

    \centering
    \VCDraw{%
    \begin{VCPicture}{(-4.4,-0.4)(4.4,3)}
        % param etats
        \MediumState 
        % etats
        \State[p]{(-4,0)}{A}
        \State[q]{(4,0)}{B}
        \State[r]{(-2,\lgb)}{C}
        \State[r']{(2,\lgb)}{D}
        %
        % transitions  
        \ArcR[.25]{A}{D}{\Kd'}
        \ArcR[.75]{C}{B}{\Hd}
        %[.05]
        \ArcL[.5]{C}{D}{\Gd}
        \ArcL[.5]{D}{C}{\Gd'}
        % \EdgeL[.5]{D}{C}{\Gd'}
        \EdgeL[.5]{A}{C}{\Kd}
        \EdgeL[.5]{D}{B}{\Hd'}
        \LoopN[.2]{C}{\Ld}
        \LoopN[.8]{D}{\Ld'}
     \end{VCPicture}}
    \eee
    \VCDraw{%
    \begin{VCPicture}{(-4.4,-0.4)(4.4,3)}
        % param etats
        \MediumState 
        % etats
        \State[p]{(-4,0)}{A}
        \State[q]{(4,0)}{B}
        % \State[r]{(-2,\lgb)}{C}
        % \State[r']{(2,\lgb)}{D}
        \State[r']{(1.5,\lgb)}{D}
        %
        % transitions  
        \EdgeL[.5]{A}{D}{\Kd\xmd \Ld^{*}\Gd + \Kd'}
        \EdgeR[.5]{A}{B}{\Kd\xmd \Ld^{*}\Hd}
        % \ArcR[.75]{C}{B}{H}
        %[.05]
        % \ArcL[.5]{C}{D}{G}
        % \ArcL[.5]{D}{C}{G'}
        % \EdgeL[.5]{A}{C}{K}
        \EdgeL[.5]{D}{B}{\Gd'\xmd \Ld^{*}\Hd + \Hd'}
        %
        % \LoopN[.5]{C}{L}
        \LoopN[.2]{D}{\Gd'\xmd \Ld^{*}\Gd + \Ld'}
    \end{VCPicture}}
\caption{First step of two in the state-elimination method}
    \label{:fig:MNY-ord}
\end{figure}

The removal of state~$r$ gives the expressions in
\figur{MNY-ord} (right).
The removal of state~$r'$ gives the expression:
\begin{displaymath}
    \Ed =
    \Kd\xmd \Ld^{*}\Hd + (\Kd\xmd \Ld^{*}\Gd + \Kd') \xmd
     \left[\Gd'\xmd \Ld^{*}\Gd + \Ld'\right]^{*}
     (\Gd'\xmd \Ld^{*}\Hd + \Hd') \eqvrg
\end{displaymath}
which using~$\IdRS$ (and the natural identities) becomes:
\begin{multline*}
\Ed \equiv
\Kd\xmd \Ld^{*}\Hd
  + \Kd\xmd \Ld^{*}\Gd \xmd \left[{\Ld'}^{*}\Gd'\xmd \Ld^{*}\Gd \right]^{*}{\Ld'}^{*}
    \Gd'\xmd \Ld^{*}\Hd \\
  + \Kd'\xmd \left[{\Ld'}^{*}\Gd'\xmd \Ld^{*}\Gd \right]^{*}{\Ld'}^{*}
    \Gd'\xmd \Ld^{*}\Hd
  + \Kd\xmd \Ld^{*}\Gd \xmd \left[{\Ld'}^{*}\Gd'\xmd \Ld^{*}\Gd \right]^{*}{\Ld'}^{*}
    \Hd' \\
  + \Kd'\xmd \left[{\Ld'}^{*}\Gd'\xmd \Ld^{*}\Gd \right]^{*}{\Ld'}^{*} \Hd'\EqPnt
\end{multline*}
We write:
\begin{displaymath}
    \Kd'\xmd \left[{\Ld'}^{*}\Gd'\xmd \Ld^{*}\Gd \right]^{*}{\Ld'}^{*} \Hd'
    \equiv
    \Kd'\xmd {\Ld'}^{*} \Hd' +
    \Kd'\xmd {\Ld'}^{*}\Gd'\xmd \Ld^{*}
      \left[\Gd \xmd {\Ld'}^{*}\Gd'\xmd \Ld^{*}\right]^{*}\Gd \xmd{\Ld'}^{*} \Hd'
\end{displaymath}
by using~$\IdRP$, and  then,
by `switching the brackets'
(using the identity $\msp (\Xd\xmd \Yd)^{*}\Xd\msp \equiv \msp
\Xd\xmd (\Yd\xmd \Xd)^{*}\msp $ which is also
a consequence of~$\IdRP$),
we obtain:
\begin{multline*}
\Ed \equiv
\Kd\xmd \Ld^{*}\Hd \\
  + \Kd\xmd \Ld^{*}\Gd \xmd \left[{\Ld'}^{*}\Gd'\xmd \Ld^{*}\Gd \right]^{*}{\Ld'}^{*}
    \Gd'\xmd \Ld^{*}\Hd
  + \Kd'\xmd {\Ld'}^{*}\Gd'\xmd \left[\Ld^{*}\Gd \xmd {\Ld'}^{*}
    \Gd'\right]^{*} \Ld^{*}\Hd \\
  + \Kd\xmd \Ld^{*}\Gd \xmd \left[{\Ld'}^{*}\Gd'\xmd \Ld^{*}\Gd \right]^{*}{\Ld'}^{*}
    \Hd'
  + \Kd'\xmd {\Ld'}^{*}\Gd'\xmd \left[\Ld^{*}
    \Gd \xmd {\Ld'}^{*}\Gd'\right]^{*} \Ld^{*}\Gd \xmd{\Ld'}^{*} \Hd' \\
  + \Kd'\xmd {\Ld'}^{*} \Hd'
\end{multline*}
an expression that is perfectly symmetric in the letters with and
without `primes', which shows that we would have obtained the same result
if we had started by removing~$r'$ then~$r$.
% The identity~$\IdRI$ may be used in the computation of~$\sem{\Ac}$ 
% or~$\sem[\omega']{\Ac}$ and must then be an axiom in the conversion 
% of one into the other. 
\end{proof}

}%

\ifcheat\smallskipneg\fi%
% The question of \emph{the length} of these 
% expressions is also of interest, both from a theoretical as well 
% as practical point of view (see Notes).
\shortlong{%
The question of \emph{the length} of these 
expressions is also of interest, both from a theoretical as well 
as practical point of view.
The above example~$\Dc$ is easily generalised so as to find an 
exponential gap between the length of expressions for two distinct 
orders.  
The search for short expressions is 
performed by heuristics
(see Notes). 
}{%
Aside from the formal proximity between expressions obtained from a 
given automaton, the question of \emph{the length} of these 
expressions is of course of interest, both from a theoretical as well 
as practical point of view.
The above example~$\Dc$ is easily generalised so as to find an 
exponential gap between the length of expressions for two distinct 
orders.  
The search for short expressions is 
performed by heuristics, with more or less degree of sophistication 
(see Notes). 
}%
% \ifcheat\clearpage\fi%

\ifcheat\smallskipneg\fi%
\ifcheat\smallskipneg\fi%
\ifcheat\smallskipneg\fi%
%%%%%%%%%%%%%%%%%%%%%%
\subsection{The system-solution method}
\label{:ssc:sys-res-met}%

The computation of an expression that denotes the language accepted by 
a finite automaton 
\index{method!system-solution}%
as the solution of a system of linear equations is nothing else 
than the state-elimination method turned into a more 
mathematical setting\longonly{%
, which allows then easier formal proofs}.

\ifcheat\smallskipneg\fi%
\paragraph{Description of the algorithm}
Given $\msp\Ac=\auta\msp$,
 for every~$p$ in~$Q$, we write~$L_{p}$ for the 
\emph{set of words} which are the label of computations from~$p$ to a
final state of~$\Ac$:
$\msp\displaystyle{%
L_{p} = \Defi{w\in\Ae}{\ext t \in T \quantsmsp p \pathaut{w}{\Ac } t}
}\msp$.
For a subset~$R$ of~$Q$, we write the symbol $\msp\delta_{p,R}\msp$ 
% for~$\unAe $ if~$p$ is in~$R$ and~$\es$ if not.
for~$\und $ if~$p$ is in~$R$ and~$\zed$ if not.
The system of equations associated with~$\Ac$ is written:
\shortlong{%

\smallskipneg \smallskipneg 
\TextFigu[0.45]{%
\begin{equation}
 \CompAuto{\Ac}  = \sum_{p \in I} L_{p} 
                  = \sum_{p \in Q} {\dpI} \xmd L_{p} 
\label{:equ:BMC-1}	
\end{equation}
}{%
\begin{equation}
\forall p \in Q \quantsmsp
L_{p} = \sum_{q \in Q} \CompExpr{\Ed_{p,q}} \xmd L_{q} + \CompExpr{\dpT}
% \eee 
\label{:equ:BMC-2}
\end{equation}
}%
\miniskip\miniskip

\noindent
}{%
\begin{align}
 \CompAuto{\Ac} & = \sum_{p \in I} L_{p} 
                  = \sum_{p \in Q} {\dpI} \xmd L_{p} 
\label{:equ:BMC-1} \\[.5ex]
\forall p \in Q \quantsp
L_{p} &= \sum_{q \in Q} \CompExpr{\Ed_{p,q}} \xmd L_{q} + \CompExpr{\dpT}
\eee 
\label{:equ:BMC-2}
\end{align}
}%
where the~$L_{p}$ are the `unknowns' and the entries~$E_{p,q}$,
% which are sums of letters of~$A$, 
which represent subsets of~$A$, 
% are considered as expressions and denoted as such.
as expressions~$\Ed_{p,q}$ are sums of letters labelling paths of 
length~$1$.  
The system~\equnm{BMC-2} may be solved by successive 
\emph{elimination} of the unknowns%
\shortlong{%
, by means of Arden's lemma,
}{%
.
The pivoting operations, which involve subtraction and division that 
are not available in the semiring~$\jsPart{\Ae}$, are replaced by the 
application of Arden's lemma,
}%
since~$\TermCst{\Ed_{p,q}}=\zed$ for all~$p$, $q$ in~$Q$.
\index{Arden's lemma!}%

\longonly{%
After the elimination of a certain number of unknowns~$L_{p}$ --- we
write~$Q'$ for the set of indices of those which have not been
eliminated --- we obtain a system of the form:
\begin{align}
\CompAuto{\Ac} & = \sum_{p \in Q'} \CompExpr{\Gd_{p}} \xmd L_{p} + 
\CompExpr{\Hd} 
\label{:equ:BMC-3} \\[.5ex]
\forall p \in Q' \quantsp
L_{p} &= \sum_{q \in Q'} \CompExpr{\Fd_{p,q}} \xmd L_{q} + 
\CompExpr{\Kd_{p}}
\eee \label{:equ:BMC-4}
\end{align}
If we choose (arbitrarily) one element~$q$ in~$Q'$, \corol{Ard} 
applied to the corresponding equation from the system~\equnm{BMC-4}, 
yields:
\begin{align}
% \forall p \in Q' \quantsp
L_{q} &= \CompExpr{\Fd_{q,q}^{*}}\xmd
        \left(\sum_{p \in Q'\bk q} \CompExpr{\Fd_{q,p}} \xmd L_{p} 
                             + \CompExpr{\Kd_{q}}\right)
% \eee 
\label{:equ:BMC-5}
\end{align}
which allows the elimination of~$L_{q}$ 
in~\equnm{BMC-3}--\equnm{BMC-4} and gives:
% The elimination in the system~\equnm{BMC-3}--\equnm{BMC-4} of the
% unknown~$L_{q}$ by application of Arden's lemma and substitutions 
% give the system:
\begin{align}
\CompAuto{\Ac} & = \sum_{r \in Q'\bk q}
   \left(\CompAuto{\Gd_{r} + \Gd_{q}\xmd\Fd_{q,q}^{*}\Fd_{q,r}}\right) 
       \xmd L_{r} 
 + \CompAuto{\Hd + \Gd_{q}\xmd\Fd_{q,q}^{*}\Kd_{q}}
\label{:equ:BMC-6} 
\\
\forall r \in Q'\bk p \quantsmsp
L_{r} &= \sum_{p \in Q'\bk q}
   \left(\CompAuto{\Fd_{r,p}+\Fd_{r,q}\xmd\Fd_{q,q}^{*}\Fd_{r,p}}\right)
   \xmd L_{p}
 + \CompAuto{\Kd_{r}+\Fd_{r,q}\xmd\Fd_{q,q}^{*}\Kd_{q}}
\eqpnt 
\e 
\label{:equ:BMC-7}
\end{align}
}%

When all unknowns~$L_{q}$ have been eliminated in the 
ordering~$\omega$ on~$Q$, the computation yields an 
expression that we denote by~$\srm{\Ac}$ and 
\shortlong{%
$\msp\CompAuto{\Ac} = \CompExpr{\srm{\Ac}}\msp$ holds.
}{%
\equnm{BMC-6} becomes:
\begin{equation}
    \CompAuto{\Ac} = \CompExpr{\srm{\Ac}}
    \eqpnt
\notag
\end{equation}
}%
As for the state-elimination method, the identities~$\IdRN$ 
(and~$\IdRI$ and~$\IdRJ$) are likely to have been involved at any 
step of the computation of~$\srm{\Ac}$.

\ifcheat\smallskipneg\fi%
\paragraph{Comparison with the state-elimination method}

The state-elimination method and the system-solution are indeed one 
and the same algorithm for computing the language accepted by a 
finite automaton, as stated by the following.

\begin{proposition}[\cite{Saka03}]
    \label{:pro:eli-equ}%
For any order~$\omega$ on the states of~$\Ac$,
\shortlong{%
$\msp\sem{\Ac} = \srm{\Ac}\msp$ holds.
}{%
it holds:\relax 
\begin{equation}
    \sem{\Ac} = \srm{\Ac}
    \eqpnt
\notag
\end{equation}
}%
\end{proposition}

\longonly{%
\begin{proof}
    We can build a generalised automaton~$\Bc '$ corresponding to the
system~\equnm{BMC-3}--\equnm{BMC-4},
with set of states is~$Q' \cup \{i,t\}$, where~$i$ and~$t$ do not
belong to~$Q'$, and such that, for all~$p$ and~$q$ in~$Q'$:

\thi the transition from~$p$ to~$q$ is labelled~$F_{p,q}$;

\thii the transition from~$p$ to~$t$ is labelled~$K_{p}$;

\thiii the transition from~$i$ to~$p$ is labelled~$G_{p}$;
 and 
 
\thiv the transition from~$i$ to~$t$ is labelled~$H$.

Note that this definition applied to the
system~\equnm{BMC-1}--\equnm{BMC-2} characterises the automaton
constructed in the first phase of the state-elimination method applied
to~$\Ac$.

The elimination in the system~\equnm{BMC-3}--\equnm{BMC-4} of the
unknown~$L_{q}$ by substitutions and the application of Arden's lemma
\index{Lemma!Arden's}%
give the system~\equnm{BMC-6}--\equnm{BMC-7}
whose coefficients are exactly the transition labels of the
generalised automaton obtained by removing the state~$q$ from~$\Bc'$.

Thus, since the starting points correspond and since each step
maintains the correspondence, the expression obtained 
for~$\CompAuto{\Ac}$ 
by the state-elimination method is the same as that obtained by the
solution of the system~\equnm{BMC-1}--\equnm{BMC-2}.
\end{proof}
}%

\ifcheat\smallskipneg\fi%
The state-elimination method reproduces, in the automaton~$\Ac$, the
computations corresponding to the solution of the system: the latter
is a \emph{formal proof} of the former.
As another consequence of \propo{eli-equ}, the following 
corollary of \theor{ord-eli-met} holds:

\ifcheat\smallskipneg\fi%
\begin{corollary}
    \label{:cor:ord-sem}%\land \IdRI
Let~$\omega $ and~$\omega' $ be two orders on the set of states of an 
automaton~$\Ac$.
Then,
\begin{equation}
       \IdRN  \land \IdRS \land \IdRP \msp\msp \dedjs 
        \srm{\Ac} \msp \equiv \msp \srm[\omega']{\Ac}
\eqpnt
\eee\eee
\notag
\end{equation}
\end{corollary}

%%%%%%%%%%%%%%%%%%%%%%
\ifcheat\smallskipneg\fi%
\ifcheat\smallskipneg\fi%
\subsection{The McNaughton--Yamada algorithm}
\label{:sec:MNa-Yam-alg}

Given $\msp\Ac=\auta\msp$,
the McNaughton--Yamada algorithm (\cite{MNauYama60}) 
--- called here \MNY algorithm for short ---
\index{McNaughton--Yamada algorithm}%
\longshort{%
truly addresses the problem of computing the matrix~$E^{*}$, whereas 
the two preceding methods rather compute the sum of some of the 
entries of~$E^{*}$.
}{%
computes~$E^{*}$, whereas 
the two preceding methods rather compute~$\CompAuto{\Ac}$ directly.
}%
Like the former methods, it relies on an ordering of~$Q$ but it is based 
on a different grouping of computations\footnote{%
   In order to avoid confusion between the \emph{computations} of 
   expressions that denote the language accepted by~$\Ac$ and whose 
   variations are the subject of the chapter, and the 
   \emph{computations} within~$\Ac$, which is the way we call the 
   paths in the labelled directed graph~$\Ac$, we use the latter 
   terminology in this section.}
within~$\Ac$.

\ifcheat\smallskipneg\fi%
\ifcheat\miniskipneg \miniskipneg \fi%
\paragraph{Description of the algorithm}

\longonly{%
We write $\msp M_{p,q}\msp $ for $(E^{*})_{p,q}$:
\begin{equation}
    M_{p,q} = \Defi{w \in \Ae}{p \pathaut{w}{\Ac} q}
    \eqpnt
    \notag
\end{equation}
}%
The set~$Q$ ordered by~$\omega$ is identified with the set of 
integers from~$1$ to~$n = \jsCard{Q}$.
The key idea of the algorithm is to group the set of paths 
between any states~$p$ and~$q$ in~$Q$ according to the \emph{highest 
rank} of the intermediate states.
We denote by $\msp M^{(k)}_{p,q}\msp$ the set of labels of 
paths from~$p$ to~$q$
which \emph{do not pass through intermediate states of rank
greater than~$k$}.
And we shall compute expressions $\msp\Md^{(k)}_{p,q}\msp$
such that $\msp \CompExpr{\Md^{(k)}_{p,q}} =M^{(k)}_{p,q}\msp$.

\longshort{%
A path that does not pass through any intermediate state of
rank greater than~$0$ passes through no intermediate states, and
therefore reduces to a single transition. 
Thus, $M^{(0)}_{p,q}$ is,
for all~$p$ and~$q$ in~$Q$, the set of labels of transitions which go
from~$p$ to~$q$; \ie  $\msp M^{(0)}_{p,q} = E_{p,q}\msp $ and 
$\msp\Md^{(0)}_{p,q}=\Ed_{p,q}\msp$.
}{%
A path that does not pass through any intermediate state of
rank greater than~$0$ reduces to a direct transition. 
Thus $\msp M^{(0)}_{p,q} = E_{p,q}\msp $ and 
$\msp\Md^{(0)}_{p,q}=\Ed_{p,q}\msp$.
}%
A path which goes from~$p$ to~$q$ without
visiting intermediate states of rank greater than~$k$ is:

\smallskip
\tha either a path (from~$p$ to~$q$) which does not visit
intermediate states of rank greater than~$k-1\xmd$;

\smallskip
\thb or the concatenation:

\miniskipneg %\miniskipneg 
\miniskipneg \miniskipneg 
\begin{itemize}
\item of a path from~$p$ to $k$ without passing through
  states of rank greater than~\mbox{$k\!-\!1$};
\item followed by an arbitrary number of paths which
  go from $k$ to $k$ without passing through intermediate states
  of rank greater than~\mbox{$k\!-\!1\xmd$};
\item followed finally by a path from $k$ to~$q$ without
  passing through intermediate states of rank greater 
  than~\mbox{$k\!-\!1\xmd$}. 
\end{itemize}

\noindent
This decomposition
\longonly{is sketched in \figur{MNY} and} 
implies that for
all~$p$ and~$q$ in~$Q$, for all~$k \leq n$, it holds:
\begin{equation}
\Md^{(k)}_{p,q} = \Md^{(k-1)}_{p,q} +
\Md^{(k-1)}_{p,k}\xmd \left(\Md^{(k-1)}_{k,k}\right)^{*}\xmd \Md^{(k-1)}_{k,q}
\eqpnt
\notag
\end{equation}
\shortonly{%
We write $\msp M_{p,q}\msp $ for $(E^{*})_{p,q}$, $\Md_{p,q}$ for an 
expression that denotes it.
}%
The algorithm ends with the last equation:
\begin{equation}
\Md_{p,q} = \Md^{(n)}_{p,q} 
\e\text{if $\msp p \neq q\msp$}\EqVrgInt
\ee
\Md_{p,q} = \Md^{(n)}_{p,q} + \und
\e\text{if $\msp p = q\msp$}
\eqpnt
\notag
\end{equation}
For consistency with the previous sections, we write
$\msp \mny{\Ac} = \sum_{p \in I, q \in T}\Md_{p,q}\msp $
\shortlong{%
and $\msp\CompAuto{\Ac} = \CompExpr{\mny{\Ac}}\msp$ holds.
}{%
and it holds:
\begin{equation}
    \CompAuto{\Ac} = \CompExpr{\mny{\Ac}}
    \eqpnt
\notag
\end{equation}
}%

\longonly{%
\begin{figure}[htbp]
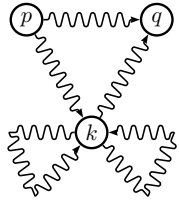

\centering
\VCDraw{%
\begin{VCPicture}{(-3,-3.7)(3,1.6)}
% etats
\State[p]{(-1.732,1)}{A}
\State[q]{(1.732,1)}{B}
% \LargeState
% \ChgStateLabelScale{0.65}
\State[k]{(0,-2)}{C}
\Point{(-2.2,-2)}{C1}
\Point{(-1.5,-3.7)}{C2}
\Point{(2.2,-2)}{C3}
\Point{(1.5,-3.7)}{C4}
%
% transitions  
\ChgCoilSize{\SmallStateDiameter}
\ChgCoilLineWidth{1.2}
\EdgeZZ{A}{B}
\EdgeZZ{A}{C}
\EdgeZZ{C}{B}
\EdgeZZ{C2}{C}
\EdgeZZ{C3}{C}
\ChgEdgeArrowStyle{-}
\EdgeZZ{C}{C1}
\EdgeZZ{C1}{C2}
\EdgeZZ{C}{C4}
\EdgeZZ{C4}{C3}
\RstCoilSize
\end{VCPicture}
}%
\caption{Step~$k$ of \MNY algorithm}
    \label{:fig:MNY}
\end{figure}
}%

\begin{example}
    \label{:exa:MNY-ex1}%
The \MNY algorithm applied to the automaton~$\Rc_{1}$ of 
\figur{MNY-ex1} yields the following matrices 
(we group together, for each~$k$, the
four~$\Md^{(k)}_{p,q}$ into a matrix~$\Md^{(k)}$):
\begin{align}
\Md^{(0)} & = \matricedd{a}{b}{a}{b} 
\EqVrg \ee 
\Md^{(1)} = \matricedd{a + a(a)^{*}a}{b + a(a)^{*}b}%
                    {a + a(a)^{*}a}{b + a(a)^{*}b}
\EqVrg
\notag 
\\[1.5ex]
\Md^{(2)} & = 
\left (\begin{matrix}
{a + a(a)^{*}a + (b + a(a)^{*}b)(b + a(a)^{*}b)^{*}(a + a(a)^{*}a)}\\
{a + a(a)^{*}a + (b + a(a)^{*}b)(b + a(a)^{*}b)^{*}(a + a(a)^{*}a)}%
\end{matrix} \right .
\notag
\\
&\eee\eee
\left . \begin{matrix}
{(b + a(a)^{*}b) + (b + a(a)^{*}b)(b + a(a)^{*}b)^{*}(b + a(a)^{*}b)}\\
{(b + a(a)^{*}b) + (b + a(a)^{*}b)(b + a(a)^{*}b)^{*}(b + a(a)^{*}b)}
\end{matrix} \right )
\EqPnt\notag 
\end{align}
\end{example}

As in the first two methods, identities in~$\IdRN$ (as well 
as~$\IdRI$ and~$\IdRJ$) are likely to be used at any step of the 
\MNY algorithm.
What is new is that
identities~$\IdRD$ and~$\IdRU$ are particularly fit for the 
computations involved in the \MNY algorithm.
For instance, after using
these identities, the above matrices  become: 
\begin{equation}
\Md^{(1)}  =  \matricedd{a^{*}a}{a^{*}b}{a^{*}a}{a^{*}b}
\ee \text{and} \ee 
\Md^{(2)} = \matricedd{(a^{*}b)^{*}a^{*}a}%
                    {(a^{*}b)^{*}a^{*}b}%
                    {(a^{*}b)^{*}a^{*}a}%
                    {(a^{*}b)^{*}a^{*}b} 
\eqpnt
\notag 
\end{equation}

\begin{figure}[htbp]
\centering
\VCDraw{%
\begin{VCPicture}{(-1.4,-1cm)(4.4,0.8cm)}
% etats.4
\MediumState
\State[1]{(0,0)}{A}\State[2]{(3,0)}{B}
\Initial[n]{A}\Final[s]{A}
% transitions  
\ArcL{A}{B}{b}\ArcL{B}{A}{a}
\LoopW{A}{a}\LoopE{B}{b}
\end{VCPicture}
}%
\caption{The automaton~$\Rc_{1}$}
\medskipneg 
    \label{:fig:MNY-ex1}%
\end{figure}

\ifcheat\smallskipneg\fi%
\paragraph{Comparison with the state-elimination method}

Comparing the \MNY algorithm with the state-elimination method 
amounts to relating
two objects whose form and mode of construction are rather different:
on the one hand, a $Q\x Q$-matrix obtained by successive
transformations and on the other hand, an
expression obtained by repeated modifications of an automaton, hence of
a matrix, but one whose size decreases at each step.
\longonly{This leads us to a more detailed statement.}

\ifcheat\smallskipneg\fi%
\begin{proposition}[\cite{Saka03}]
    \label{:pro:com-mya-sem}%\land\IdRI
Let~$\msp\Ac = \auta\msp$ be an automaton and
for every~$p$ and~$q$ in~$Q$, 
let~$\Ac_{p,q}$ be the automaton defined by 
$\msp\Ac_{p,q} = \aut{Q,A,E,\{p\},\{q\}}\msp$.
\shortlong{%
For every order~$\omega$ on~$Q$,
$\msp
\IdRN\land\IdRA \msp\msp\dedtxt
\mny{\Ac_{p,q}} \msp \equiv \msp
\sem{\Ac_{p,q}} 
\msp$
holds.
}{%
For every (total) order~$\omega$ on~$Q$ and every~$p$ and~$q$ in~$Q$, 
it holds: 
\begin{equation}
\IdRN\land\IdRA \msp\msp\dedjs
\mny{\Ac_{p,q}} \msp \equiv \msp
\sem{\Ac_{p,q}} 
\eqpnt
\notag
\end{equation}
}%
\end{proposition}

\longonly{%
\begin{proof}
% To prove this result we will show a
% correspondence between the operations performed by the two algorithms.
% The difficulty, if it can be called that, is that we have to compare
% two objects whose form and mode of construction are rather different:
% on one hand a~$Q \x Q$ matrix obtained by successive
% transformations, from which we choose one entry, and on the other an
% expression obtained by repeated modification of an automaton, hence of
% a matrix, but one whose size decreases at each step.

In the following, $\Ac$~and~$\omega $ are fixed and remain implicit. The
automaton~$\Ac$ has~$n$ states, identified with the integers from~$1$
to~$n$; the two algorithms perform~$n$ steps starting in a situation
called `step~$0$', the $k$th step of the state-elimination method
consisting of the removal of state~$k$, and that of algorithm\MNY
consisting of calculating the labels of paths that do not include
nodes (strictly) greater than~$k$. We write:
\begin{displaymath}
    \mBC{k}{r,s}
\end{displaymath}
for the label of the transition from~$r$ to~$s$ in the automaton
obtained from~$\Ac$ (and~$\omega $) at the $k$th step of the state
elimination method; necessarily, in this notation, $\msp k+1\jsleq r\msp$
and $\msp k+1\jsleq s\msp $ (abbreviated to $\msp k+1\jsleq r,s\msp $).
As above, we write:
\begin{displaymath}
    \mMNY{k}{r,s}
\end{displaymath}
for the entry~$r,s$ of the~$n\x n$ matrix computed by the $k$th
step of \MNY algorithm. At step~$0$, the automaton~$\Ac$ has not been
modified and we have:
\begin{equation}
\fa r,s \quantvrg 1\jsleq r,s\jsleq n\quantsp
\mMNY{0}{r,s} = \mBC{0}{r,s} \EqVrgInt
\eee \eee
\label{:equ:BM-1}
\end{equation}
which will be the base case of the inductions to come.
The \MNY algorithm is written:
\begin{multline}
\fa k \quantvrg 0< k\jsleq n \quantvrg
\fa r,s \quantvrg 1\jsleq r,s\jsleq n \quantsp \\
\mMNY{k}{r,s} = \mMNY{k-1}{r,s} + \mMNY{k-1}{r,k}\cdot
\left(\mMNY{k-1}{k,k}\right)^{*}\cdot \mMNY{k-1}{k,s} \eqpnt\ee
\label{:equ:BM-2}
\end{multline}
The state-elimination algorithm is written:
\begin{multline}
\fa k \quantvrg 0< k\jsleq n \quantvrg
\fa r,s \quantvrg k< r,s\jsleq n \quantsp \\
\mBC{k}{r,s} = \mBC{k-1}{r,s} + \mBC{k-1}{r,k}\cdot
\left(\mBC{k-1}{k,k}\right)^{*}\cdot \mBC{k-1}{k,s}
\label{:equ:BM-3}
\end{multline}
Hence we conclude, for given~$r$ and~$s$ and by induction on~$k$:
\begin{equation}
\fa r,s \quantvrg 1\jsleq  r,s\jsleq n \quantvrg
\fa k \quantvrg 0\jsleq  k <\min (r,s)  \quantsp
\mMNY{k}{r,s} = \mBC{k}{r,s}
\label{:equ:BM-4}
\end{equation}
We see in fact (as there is even so something to see) that
if~$k < \min (r,s)$ then all integer triples~$(l,u,v)$ such
that~$\mMNY{l}{u,v}$ occurs in the computation of~$\mMNY{k}{r,s} $ by
the (recursive) use of~\equnm{BM-2}, are such that~$l < \min (u,v)$.

Suppose now that we have~$p$ and~$q$, also fixed, such that~$1\jsleq
p<q\jsleq n$ (the other cases are dealt with similarly). We call the
initial and final states added to~$\Ac$ in the first phase of the
state-elimination method~$i$ and~$t$ respectively;
% \NoteEnMarge{\cf the proof of \PropoP{rec-s-rat-gen}}%
$i$~and~$t$ are not integers between~$1$ and~$n$.
The transition from~$i$ to~$p$ and that from~$q$ to~$t$ are
labelled~$\unAe$. Now let us consider step~$p$ of each algorithm.
For every state~$s$, $p<s$, $\mMNY{p}{p,s}$ is given by~\equnm{BM-2}:
\begin{equation}
\mMNY{p}{p,s} = \mMNY{p-1}{p,s} + \mMNY{p-1}{p,p}\cdot
\left(\mMNY{p-1}{p,p}\right)^{*}\cdot \mMNY{p-1}{p,s}
\notag
\end{equation}
and~$\mBC{p}{i,s}$ by:
\begin{equation}
\mBC{p}{i,s} =  \left(\mBC{p-1}{p,p}\right)^{*}\cdot \mBC{p-1}{p,s}
\notag
\end{equation}
and hence, by~\equnm{BM-4}:
\begin{equation}
\fa s \quantvrg p<s\jsleq n \quantsp
\IdRA \dedjs \mMNY{p}{p,s} \equiv \mBC{p}{i,s} \eqpnt \ee
\label{:equ:BM-5}
\end{equation}
Next we consider the steps following~$p$ (and row~$p$ of the
matrices~$\Md^{(k)}$).
For all~$k$, $p<k$, and all~$s$, $k<s\jsleq n$,
$\mMNY{k}{p,s}$ is always computed by~\equnm{BM-2}
and~$\mBC{k}{i,s}$ by:
\begin{equation}
\mBC{k}{i,s} =  \mBC{k-1}{i,s} + \mBC{k-1}{i,k}\cdot
\left(\mBC{k-1}{k,k}\right)^{*}\cdot \mBC{k-1}{k,s} \eqpnt
\label{:equ:BM-6}
\end{equation}
From~\equnm{BM-5}, and based on an observation analogous to the
previous one, we conclude from the term-by-term correspondence
of~\equnm{BM-2} and~\equnm{BM-6} that:
\begin{equation}
\fa k \quantvrg p<k \quantvrg
\fa s \quantvrg p<s\jsleq n \quantsp
\IdRA \dedjs \mMNY{k}{p,s} \equiv \mBC{k}{i,s} \eqpnt \ee \ee
\label{:equ:BM-7}
\end{equation}
The analysis of step~$q$ gives a similar, and symmetric, result to
that which we have just obtained from the analysis of step~$p$:
for all~$r$, $q<r$, we have:
\begin{gather*}
\mMNY{q}{r,q} = \mMNY{q-1}{r,q} + \mMNY{q-1}{r,q}\cdot
\left(\mMNY{q-1}{q,q}\right)^{*}\cdot \mMNY{q-1}{q,q}\\
\text{and}\eee \eee
\mBC{q}{r,t} =  \mBC{q-1}{r,q}\cdot \left(\mBC{q-1}{q,q}\right)^{*}\eee \eee
\end{gather*}
and hence
\begin{equation}
\fa r \quantvrg q<r\jsleq n \quantsp
\IdRA \dedjs \mMNY{q}{r,q} \equiv \mBC{q}{r,t} \eqpnt \ee
\label{:equ:BM-8}
\end{equation}
The steps following~$q$ give rise to an equation symmetric
to~\equnm{BM-7} (for column~$q$ of the matrices~$\Md^{(k)}$):
\begin{equation}
\fa k \quantvrg q<k \quantvrg
\fa r \quantvrg q<r\jsleq n \quantsp
\IdRA \dedjs \mMNY{k}{r,q} \equiv \mBC{k}{r,t} \eqpnt \ee \ee
\label{:equ:BM-9}
\end{equation}
Finally, from:
\begin{gather*}
\mMNY{k}{p,q} = \mMNY{k-1}{p,q} + \mMNY{k-1}{p,k}\cdot
\left(\mMNY{k-1}{k,k}\right)^{*}\cdot \mMNY{k-1}{k,q}\\
\text{and}\ee
\mBC{k}{i,t} =  \mBC{k-1}{i,t}+ \mBC{k-1}{i,k}\cdot
\left(\mBC{k-1}{k,k}\right)^{*}\cdot \mBC{k-1}{k,t}\ee
\end{gather*}
Equations~\equnm{BM-4}, \equnm{BM-7} and~\equnm{BM-9}
together allow us to conclude, by induction on~$k$, that:
\begin{equation}
\fa k \quantvrg q\jsleq k \jsleq n \quantsp
\IdRA \dedjs \mMNY{k}{p,q} \equiv \mBC{k}{i,t} \eqpnt \ee
\label{:equ:BM-10}
\end{equation}
When we reach~$k=n$ in this equation we obtain the identity we want.
\end{proof}
}%

\ifcheat\smallskipneg\fi%
As a consequence of \propo{com-mya-sem}, we have the following 
corollary of \theor{ord-eli-met}:\relax

\ifcheat\smallskipneg\fi%
\begin{corollary}
    \label{:cor:ord-mny}%\land \IdRI
Let~$\omega $ and~$\omega'$ be two orders on the states of an 
automaton~$\Ac$.
Then,
\begin{equation}
\textstyle{\IdRN  \land \IdRS \land \IdRP \msp\msp \dedjs 
        \mny{\Ac} \msp \equiv \msp \mny[\omega']{\Ac}}
\eqpnt
\eee\eee
\notag
\end{equation}
\end{corollary}

\ifcheat\smallskipneg\fi%
\ifcheat\smallskipneg\fi%
%%%%%%%%%%%%%%%%%%%%%%
\subsection{The recursive method}
\label{:sec:rec-met}%

\shortlong{%
This last method is due to Conway \cite{Conw71}.
It is based on computation on matrices \via \emph{block 
decomposition}.
}{%
The last method we want to quote appeared first in Conway's book 
\textit{Regular Algebra and Finite Machines} \cite{Conw71} which gave 
a new start to the formal study of rational expressions (\cf 
\CTchp{ZE}).
}%
\index{method!recursive}%
\index{block decomposition|see{matrix}}%
\index{matrix!block decomposition of --}%
Originally, it yields a proof of \propo{sta-mat}%
\longonly{%
(the entries 
of~$E^*$ belong to the rational closure of the entries of~$E$)}.
As we did above, we modify it so as  
to make it compute from~$\Ed$, a matrix of rational expressions which  
denotes~$E$, a matrix~$\Ed'$ of rational expressions which 
denotes the matrix~$E^*$. 

\ifcheat\smallskipneg\fi%
\paragraph{Description of the algorithm}

\longonly{%
The recursive method (for computing~$\Ed'$, our goal) is based on 
computation on matrices \via \emph{bloc decomposition}.
\index{bloc decomposition (of matrices)}%
\index{decomposition!bloc -- (of matrices)}%

Let~$M$ and~$M'$ be two $Q\x Q$-matrices (over any semiring indeed) 
and let~$Q$ be the disjoint union of~$R$ and~$S$.
Let us write their bloc decomposition according to~$Q=R\cup S$ as:
\begin{equation}
    M =  \matricedd{F}{G}{H}{K}
    \eee
    M' =  \matricedd{F'}{G'}{H'}{K'}   
\notag
\end{equation}
where~$F$ and~$F'$ are $R\x R$-matrices,
$K$ and~$K'$ are $S\x S$-matrices,
$G$ and~$G'$ $R\x S$-matrices,
and $H$ and~$H'$ $S\x R$-matrices.
The bloc decomposition is consistent with the matirx operations in 
the sense that we have:
\begin{multline}
    M + M' =  \matricedd{F+F'}{G+G'}{H+H'}{K+K'}
    \ee \text{and}\\ \ee
    M\matmul M' =  \matricedd{F\matmul F' + G\matmul H'}%
                             {F\matmul G' + G\matmul K'}%
                             {H\matmul F' + K\matmul H'}%
                             {H\matmul G' + K\matmul K'}   
\notag	
\end{multline}
}%
Let us write a block decomposition of~$E$ and the corresponding ones  
for~$\Ed$ and~$E^*$:
\begin{equation}
E= \matricedd{F}{G}{H}{K} \EqVrgInt
\ee
\Ed= \matricedd{\Fd}{\Gd}{\Hd}{\Kd} \EqVrgInt
\ee
E^*= \matricedd{U}{V}{W}{Z} 
\eqvrg
\notag
\end{equation}
where~$F$ and~$K$ (and thus~$\Fd$, $\Kd$, $U$ and~$Z$) are \emph{square matrices}.
By~\equnm{ide-U-lan}\footnote{%
   applied to matrices with entries in~$\PdAe$ rather than to 
   elements of~$\PdAe$.},
it follows that
\begin{equation}
E^* = \matricedd{U}{V}{W}{Z} =
\matricedd{1}{0}{0}{1} + \matricedd{F}{G}{H}{K}\matricedd{U}{V}{W}{Z} \eqvrg
\notag
\end{equation}
an equation which can be decomposed into a system of four other 
equations: 
\begin{alignat}{2}
U & =  1 + \CompExpr{\Fd}\xmd\xmd U + \CompExpr{\Gd}\xmd\xmd W  \eqvrg 
&
Z & =  1 + \CompExpr{\Hd}\xmd\xmd V + \CompExpr{\Kd}\xmd\xmd Z   \eqvrg
\ee
\label{:equ:con1}\\
V & =  \CompExpr{\Fd}\xmd\xmd V + \CompExpr{\Gd}\xmd\xmd Z \eqvrg  
& \ee\text{and}\eee
W & =  \CompExpr{\Hd}\xmd\xmd U + \CompExpr{\Kd}\xmd\xmd W   \eqpnt 
\label{:equ:con2}
\end{alignat}
\corol{Ard} applies to \equnm{con2} and then, after substitution, to 
\shortlong{%
\equnm{con1}.
}{%
\equnm{con1} gives:
\begin{alignat}{2}
V &= \CompExpr{\Fd}^*\xmd\CompExpr{\Gd}\xmd\xmd Z &\e \et \e 
W & = \CompExpr{\Kd}^*\xmd\CompExpr{\Hd}\xmd\xmd U \eqpnt
		\ee
		\notag
		\\
U &= (\CompExpr{\Fd}+\CompExpr{\Gd}\xmd\xmd\CompExpr{\Kd}^*\xmd\CompExpr{\Hd})^* 
&\e \et \e 
Z &= (\CompExpr{\Kd}+\CompExpr{\Hd}\xmd\xmd\CompExpr{\Fd}^*\xmd\CompExpr{\Gd})^* 
\eqpnt
\notag
\end{alignat}
}%
This procedure leads to the computation of~$E^{*}$ by induction on 
its dimension.
By the induction hypothesis, obviously fulfilled for matrices of 
dimension~$1$, $\msp\CompExpr{\Fd}^{*}$ and~$\CompExpr{\Kd}^{*}$ are 
denoted by matrices of rational expressions~$\Fd'$ and~$\Kd'$. 
Let us write
\begin{equation}
 \Ed'= \matricedd{({\Fd} + {\Gd}\xmd {\Kd'}\xmd {\Hd})^*}%
                {{\Fd'}\xmd{\Gd}\xmd ({\Kd} + {\Hd}\xmd{\Fd'}\xmd{\Gd})^*}%
                {{\Kd'}\xmd{\Hd}\xmd ({\Fd} + {\Gd}\xmd{\Kd'}\xmd{\Hd})^*}%
                {({\Kd} + {\Hd}\xmd {\Fd'}\xmd{\Gd})^*} \eqpnt
\notag
\end{equation}
and $\msp\CompExpr{\Ed'}=E^*\msp$ holds.
Another application of the induction hypothesis to 
$\msp\CompExpr{({\Fd} + {\Gd}\xmd {\Kd'}\xmd {\Hd})}^*\msp$
and
$\msp\CompExpr{({\Kd} + {\Hd}\xmd {\Fd'}\xmd{\Gd})}^*\msp$
shows that the entries of~$\Ed'$,
which we denote by $\msp \rcm{\Ac}\xmd$,
where~$\tau$ is the recursive division of~$Q$ used in the computation,
are all in~$\RatEA$.

\ifcheat\smallskipneg\fi%
\begin{example}
    \label{:exa:rcm-ex1}%
The recursive method applied to the automaton~$\Rc_{1}$ of 
\exemp{MNY-ex1} (\cf~\figur{MNY-ex1})
directly gives 
(there is no choice for the recursive division):
\begin{equation}
\rcm{\Rc_{1}}  =  \matricedd{(a + b(b)^{*}a)^{*}}%
                            {a^{*}b(b + a(a)^{*}b)^{*}}%
                            {b^{*}a(a + b(b)^{*}a)^{*}}%
                            {(b + a(a)^{*}b)^{*}}
\eqpnt 
\notag 
\end{equation}
\end{example}

\ifcheat\smallskipneg\fi%
\ifcheat\smallskipneg\fi%
\ifcheat\smallskipneg\fi%
\ifcheat\smallskipneg\fi%
\paragraph{Comparison with the state-elimination method}

Both the recursive method and the \MNY algorithm  yield a matrix of 
expressions.
\exemp{rcm-ex1} shows that
there is no hope for an easy inference of teh equivalence of the two 
matrices. 
We state however the following conjecture.

\ifcheat\smallskipneg\fi%
\begin{conjecture}
\label{:con:com-rec-sem}%
Let~$\Ac = \auta$ be an automaton.
For every recursive division $\tau$ of $Q$ and for every pair $(p,q)$ 
of states,
there exists an ordering $\omega$ of $Q$ such that:
\ifcheat\smallskipneg\fi%\land\IdRI~$\IdRI$,
\ifcheat\smallskipneg\fi%
\begin{equation}
{\IdRN\land\IdRU} \msp\msp\dedjs 
     \left(\rcm{\Ac}\right)_{p,q} \msp\equiv\msp \sem{\Ac_{p,q}}
\eqpnt 
\notag
\end{equation}
\end{conjecture}

\ifcheat\smallskipneg\fi%
More generally, and as a conclusion of the description of these four 
methods, one would conjecture that \emph{the rational expressions 
computed from a same finite automaton are all equivalent modulo the 
natural identities and the aperiodic ones~$\IdRS$ and~$\IdRP$.}
Even if \emph{computed from} is not formal enough, the above 
developments should make the general idea rather clear.

% \ifcheat\smallskipneg\fi%
% \ifcheat\smallskipneg\fi%
%%%%%%%%%%%%%%%%%%%%%%
\subsection{Star height and loop complexity}
\label{:sec:sth-lcp}%

\longonly{%
The purpose of this last subsection is to present a refinement of Kleene's 
Theorem --- or, rather, of the Fundamental Theorem of Finite Automata 
--- 
which relates even more closely than above an automaton and the 
rational expressions that are computed from it.
}%

Among the three rational operators~$+$, $\cdot$ and~$*$, 
the operator~$*$ is the one that `gives access to the infinite', hence
the idea of measuring the complexity of an expression 
% by counting the number of nested uses of 
by finding the degree of nestedness of 
this operator, 
a number called \emph{star height}. 
On the other hand, it is the 
\emph{circuits} in a finite automaton that produce an infinite number 
of computations, `all the more' that the circuits are more 
`entangled'.
The intuitive idea of entanglement of circuits will be captured by 
the notion of \emph{loop complexity}.
\shortlong{%
A refinement of \theor{fun-aut-fm} relates 
}{%
We show how
}%
the loop complexity of an automaton to 
the star height of an expression that is computed from this 
automaton, a result which is due originally to Eggan (\cite{Egga63}).

\ifcheat\smallskipneg\fi%
\paragraph{Star height of an expression}

Let~$\Ed$ be an expression over~$\Ae$.
The \emph{star height} of~$\Ed$, denoted by~$\jsHaut{\Ed}$, is 
inductively defined by 
\shortlong{%
$\msp \jsHaut{\Ed} = \max (\jsHaut{\Ed'},\jsHaut{\Ed''})\msp$ if
$\msp\Ed= \Ed' + \Ed''\msp$ or $\msp\Ed= \Ed' \cdot \Ed''\msp$,
$\msp \jsHaut{\Ed} = 1 + \jsHaut{\Fd} \msp$
if $\msp\Ed= \Fd^{*}\msp$,
and starting from 
$\msp \jsHaut{\zed} = \jsHaut{\und} =0\msp$ 
and $\msp\jsHaut{a} =0\msp$, 
for every~$a$ in~$A$.
}{%
\begin{align}
& \text{if \e $\msp\Ed= \zed\msp$, $\msp\Ed= \und\msp$
      or $\msp\Ed= a \in A\msp$,} &\eee
\jsHaut{\Ed}& = 0 \eqvrg
% \label{:equ:hau-eto-1}
\notag
\\
& \text{if \e $\msp\Ed= \Ed' + \Ed''\msp$
      or $\msp\Ed= \Ed' \cdot \Ed''\msp$,} & \eee
\jsHaut{\Ed}& =
\max (\jsHaut{\Ed'},\jsHaut{\Ed''}) \eqvrg \ee
% \label{:equ:hau-eto-2} 
\notag
\\
& \text{if \e $\msp\Ed= \Fd^{*}\msp$,} & \eee
\jsHaut{\Ed}& = 1 + \jsHaut{\Fd} 
\eqpnt
% \label{:equ:hau-eto-3}
\notag
% \end{alignat}
\end{align}
}%
\index{star height!of expression}%
\index{expression!star height of --}%
\index{height|see{star height}}%

\begin{example}
    \label{:exa:sta-hgt-1}%
\thi
$\msp\jsHaut{(a+b)^{*}} = 1\msp$;
\ee 
$\msp\jsHaut{a^{*}\xmd (b\xmd a^{*})^{*}} = 2\msp$.

\thii
The heights of the three expressions computed for the 
automaton~$\Dc_{3}$ at \sbsct{sta-eli-met} are:
$\msp\jsHaut{\sem[\omega_{1}]{\Dc_{3}}}=2\msp$,
$\msp\jsHaut{\sem[\omega_{2}]{\Dc_{3}}}=3\msp$, and
$\msp\jsHaut{\sem[\omega_{3}]{\Dc_{3}}}=3\msp$.
\end{example}

\shortlong{%
Two equivalent
expressions may then have different star heights, and thus give
rise to the \emph{star-height problem} (see Notes and \CTchp{LM}).
}{%
These examples draw attention to the fact that two equivalent
expressions may have different star heights and that star height is 
unrelated to the length.
They also naturally give
rise to the so-called \emph{star-height problem}.
As it does not directly pertain to the matter developed in this 
chapter, we postpone the few indications we give on this problem
to the Notes section ((see also \CTchp{LM}).
}% 
\index{star height!problem}%

\paragraph{Loop complexity of an automaton}

\shortlong{%
We call \emph{loop complexity} of an automaton~$\Ac$  
\index{loop!complexity}%
% \index{graph!ball in --}%
% \index{graph!strongly connected component}%
the integer $\msp\jsEnla{\Ac}\msp$ defined inductively by the
following equations (where a \emph{ball}
is a  non-trivial strongly connected component):

\smallskip  
$\jsEnla{\Ac}  = 0$  
\PushLine
if~$\Ac $ contains no balls (in particular if~$\Ac$ is empty);\e
          
\smallskip  
$\jsEnla{\Ac}  =
 \max\Defi{\jsEnla{\Pc}}{\text{$\Pc$ a ball in~$\Ac$}}$  
\PushLine
if~$\Ac$ is not strongly connected;\e

\smallskip  
$\jsEnla{\Ac}  = 1 + 
 \min \Defi{\jsEnla{\Ac \bk \{s\}}}{\text{$s$ state of~$\Ac$}}$
\PushLine
if~$\Ac$ is strongly connected.\e
}{%
Let $\msp\Ac=\auta\msp$ be an automaton;
we call \emph{balls}\footnote{%
   Translation of the French: \emph{pelote}.}
the strongly connected components of~$\Ac$  
that contain at least one transition.
\index{graph!ball in --}%
\index{graph!strongly connected component}%
In other words, a strongly connected component that contains at least
two states is a ball, and a strongly connected component reduced to a
single state~$s$ is a ball if and only if~$s$ is the source (and the
destination) of at least one loop. 
Balls are pairwise disjoint but do
not form a covering (hence a partition) of~$Q$ since a state may belong
to no ball (\cf \figur{ree-con}).
\begin{figure}[htbp]
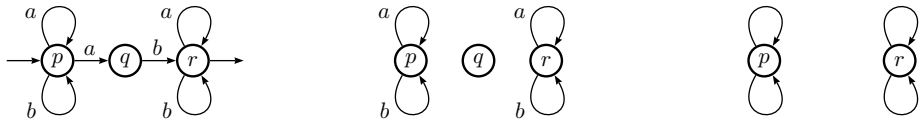

\centering
% \subfigure[]{%
\VCDraw[.9]{%
\begin{VCPicture}{(-1.4,-1.4cm)(5.4,1.4cm)}
\State[p]{(0,0)}{A}\State[q]{(2,0)}{B}\State[r]{(4,0)}{C}
\Initial{A}\Final{C}
\EdgeL{A}{B}{a}\EdgeL{B}{C}{b}
\LoopN{A}{a}\LoopS{A}{b}\LoopN{C}{a}\LoopS{C}{b}
\end{VCPicture}
}%
% }%
\PushLine
% \subfigure[]{%
\VCDraw[.9]{%
\begin{VCPicture}{(-1.4,-1.4cm)(5.4,1.4cm)}
\State[p]{(0,0)}{A}\State[q]{(2,0)}{B}\State[r]{(4,0)}{C}
% \Initial{A}\Final{C}
% \EdgeL{A}{B}{a}\EdgeL{B}{C}{b}
\LoopN{A}{a}\LoopS{A}{b}\LoopN{C}{a}\LoopS{C}{b}
\end{VCPicture}
}%
% }%
\PushLine
% \subfigure[]{%
\VCDraw[.9]{%
\begin{VCPicture}{(-1.4,-1.4cm)(5.4,1.4cm)}
\State[p]{(0,0)}{A}\State[r]{(4,0)}{C}
% \Initial{A}\Final{C}\State[q]{(3,0)}{B}
% \EdgeL{A}{B}{}\EdgeL{B}{C}{}
\LoopN{A}{}\LoopS{A}{}\LoopN{C}{}\LoopS{C}{}
\end{VCPicture}
}%
% }%
\caption{An automaton, its strongly connected components and its two balls}
\label{:fig:ree-con}
\end{figure}

\begin{definition}
    \label{:def:loo-cpl}
The \emph{loop complexity} of 
an automaton~$\Ac$  
\index{star height!loop complexity}%
is the integer $\msp\jsEnla{\Ac}\msp$ defined inductively by the
following equations:

\smallskip  
$\jsEnla{\Ac}  = 0$  
\PushLine
if~$\Ac $ contains no balls (in particular if~$\Ac$ is empty);\e
          
\smallskip  
$\jsEnla{\Ac}  =
 \max\Defi{\jsEnla{\Pc}}{\text{$\Pc$ a ball in~$\Ac$}}$  
\PushLine
if~$\Ac$ is not strongly connected;\e

\smallskip  
$\jsEnla{\Ac}  = 1 + 
 \min \Defi{\jsEnla{\Ac \bk \{s\}}}{\text{$s$ state of~$\Ac$}}$
\PushLine
if~$\Ac$ is strongly connected.\e
\end{definition}

\figur{loo-cpl} shows automata with loop complexity~$1$, $2$ 
and~$3$ respectively.

\begin{figure}[htbp]
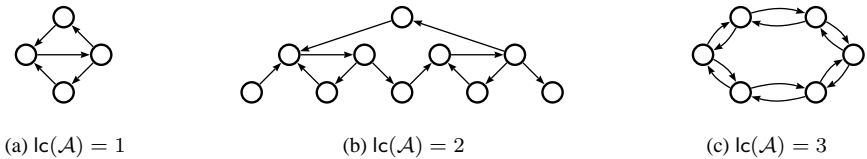

\centering
\subfigure[$\jsEnla{\Ac}  = 1$]{%
\ee
\VCDraw{%[.9]
\begin{VCPicture}{(-1.4,-1.4)(1.4,1.4)}
% etats
\SmallState
\State{(0,1)}{A}
\State{(-1,0)}{B}\State{(1,0)}{C}
\State{(0,-1)}{D}
%
% transitions 
\Edge{A}{B}\Edge{B}{C}\Edge{C}{A}
\Edge{C}{D}\Edge{D}{B}
\end{VCPicture}
}%
\ee}%
\PushLine
\subfigure[$\jsEnla{\Ac}  = 2$]{%
\VCDraw{%[.9]
\begin{VCPicture}{(-4.4,-1.4)(4.4,1.4)}
% etats
\SmallState
\State{(0,1)}{A}
\State{(0,-1)}{EF}
\VCPut{(-2,0)}%
{%
\State{(-1,0)}{B1}\State{(1,0)}{C1}
\State{(0,-1)}{D1}\State{(-2,-1)}{E1}
}
\VCPut{(2,0)}%
{%
\State{(-1,0)}{B2}\State{(1,0)}{C2}
\State{(0,-1)}{D2}\State{(2,-1)}{F2}
}
%
% transitions 
\Edge{A}{B1}
\Edge{E1}{B1}\Edge{B1}{C1}\Edge{C1}{EF}
\Edge{C1}{D1}\Edge{D1}{B1}
\Edge{C2}{A}
\Edge{EF}{B2}\Edge{B2}{C2}\Edge{C2}{F2}
\Edge{C2}{D2}\Edge{D2}{B2}
\end{VCPicture}
}%
}%
\PushLine
\subfigure[$\jsEnla{\Ac}  = 3$]{%
\ee
\VCDraw{%[.9]
\begin{VCPicture}{(-0.4,-1.4)(4.4,1.4)}
% etats
\SmallState
\State{(0,0)}{E}
\State{(4,0)}{F}
\VCPut{(1,-1)}%
{\State{(0,0)}{A1}\State{(2,0)}{B1}}
\VCPut{(1,1)}%
{\State{(0,0)}{A2}\State{(2,0)}{B2}}
%
% transitions 
\ArcL{A1}{E}{}\ArcL{E}{A1}{}
\ArcL{A2}{E}{}\ArcL{E}{A2}{}
\ArcL{A1}{B1}{}\ArcL{B1}{A1}{}
\ArcL{A2}{B2}{}\ArcL{B2}{A2}{}
\ArcL{B1}{F}{}\ArcL{F}{B1}{}
\ArcL{B2}{F}{}\ArcL{F}{B2}{}
\end{VCPicture}
\ee}%
}%
\medskipneg%\medskipneg 
\caption{Automata with differents loop complexities}
\label{:fig:loo-cpl}
\end{figure}

\paragraph{Eggan's Theorem}      
We are now ready to state the announced refinement of the Fundamental 
theorem of finite automata.
}%

\begin{theorem}[Eggan~\cite{Egga63}]
    \label{:the:the-Egg}%
The loop complexity of a trim automaton~$\Ac $ is the minimum
of the star height of the expressions computed on~$\Ac$ by
the state-elimination method.
\end{theorem}

This theorem may be proved by establishing a more 
precise statement which involves a refinement of the loop complexity 
\index{loop!index}%
and which we call the \emph{loop index}.

\shortlong{%
If~$\omega$ is an order on the state set of~$\Ac$, 
we write~$\OL{\omega}$ for the \emph{greatest state} 
according to~$\omega$. 
If~$\Rc$ is a subautomaton of~$\Ac$, we also
write~$\omega $ for the trace of~$\omega $ over~$\Rc$ and, in such a 
context, $\OL{\omega}$~for the greatest state 
of~$\Rc$ according to~$\omega $.
}{%
Let~$\msp\Ac = \auta\msp$ be an automaton.
If~$\omega $ is an order on~$Q$, 
we write~$\OL{\omega}$ for the \emph{greatest element} of~$Q$
according to~$\omega $. 
If~$\Sc$ is a sub-automaton of~$\Ac$, we also
write~$\omega $ for the trace of the order~$\omega$ over the set~$R$ 
of states of~$\Sc$ and, in such a context, $\OL{\omega}$~for the 
greatest element of~$R$ according to~$\omega$.
}%
Then, the \emph{loop index of~$\Ac $ relative to~$\omega $},
written~$\ILC{\Ac}$, is the integer inductively defined by the 
following: 
\longonly{\par}%
\begin{itemize}
\item  
if~$\Ac $ contains no ball, or is empty, then
\longshort{%
\begin{equation}
\ILC{\Ac}  = 0 
\eqpntvrg
\label{:equ:egg-5}
\end{equation}
}{%
$\msp\ILC{\Ac}  = 0\msp$;
}% 

\item if~$\Ac $ is not itself a ball, then
\longshort{%
\begin{equation}
\ILC{\Ac}  =
\max\left(\Defi{\ILC{\Pc}}{\text{$\Pc$ ball in~$\Ac $}}\right)
\eqpntvrg
\label{:equ:egg-3}
\end{equation}
}{%
$\msp\ILC{\Ac}  =
\max\left(\Defi{\ILC{\Pc}}{\text{$\Pc$ ball in~$\Ac $}}\right)
\msp$;
}%
 
\item if~$\Ac $ is a ball, then
\longshort{%
\begin{equation}
\ILC{\Ac}  = 1 + \ILC{\Ac \bk \OL{\omega}} 
\eqpnt
\label{:equ:egg-4}
\end{equation}
}{%
$\msp\ILC{\Ac}  = 1 + \ILC{\Ac \bk \OL{\omega}}\msp$.
}% 
\end{itemize}
The difference with respect to loop complexity is that the state that
we remove from a strongly connected automaton (in the inductive 
process) is fixed by the order~$\omega $ rather than being the result of
a minimisation.
This definition immediately implies that

\shortlong{%
\PushLine
$\msp
\jsEnla{\Ac}=
\min\Defi{\ILC{\Ac}}{\text{ $\omega $ is an order on~$Q$ }}
\msp$.
\PushLine

\noindent 
holds and \theor{the-Egg} is then a consequence
of the following.
}{%
\begin{property}
    \label{:pty:deg-enl-1}%
\ee $\msp
\jsEnla{\Ac}=
\min\Defi{\ILC{\Ac}}{\text{ $\omega $ is an order on~$Q$ }}
\msp$.
\EOP
\end{property}

\theor{the-Egg} is then a consequence of the following.
}%

\begin{proposition}[\cite{LombSaka01}]
    \label{:pro:enl-hau-eto}
For any order~$\omega$	on the states of~$\Ac$,
$\msp \ILC{\Ac} = \jsHaut{\sem{\Ac}}\msp$.
\end{proposition}

\longonly{%
At this point, let us note that in the inductive definition 
of~$\ILC{\Ac}$ it is the \emph{greatest} element~$\OL{\omega}$ of~$Q$ that 
is considered whereas in the construction of the 
expression~$\sem{\Ac}$ it is the \emph{smallest} element of~$Q$ 
suppressed first.

In the course of the computation of~$\sem{\Ac}$ by the state 
elimination method we consider automata whose transitions are 
labelled not by letters only but by rational expressions in general.
In order to define the index of such generalised automata, we first 
define the \emph{index}, written~$\iEB{e}$, of a transition~$e$ as 
the star height of the label of~$e$:
\begin{equation}
    \iEB{e} = \jsHaut{|e|} 
    \eqpnt
    \notag 
\end{equation}
The index of a generalised automaton is then defined inductively, by 
formulas that take into account the index of every transition:

\begin{itemize}
\item  
if~$\Ac $ is empty, then
\begin{equation}
\ILC{\Ac}  = 0 
\eqpntvrg
\label{:equ:egg-0}
\end{equation}

\item if~$\Ac $ is not itself a ball, then
\begin{equation}
    \begin{split}
\ILC{\Ac}  =
\max\left(\Defi{\iEB{e}}{\text{$e$ does not belong to a ball 
in~$\Ac$}}\right.\eee\\
         \left.\cup \Defi{\ILC{\Pc}}{\text{$\Pc$ ball in~$\Ac $}}\right)
     \eqpntvrg
     \end{split}
%  \ILC{\Ac}  =
%  \max\left(\Defi{\iEB{e}}{\text{$e$ does not belong to a ball 
%  in~$\Ac$}}
%           \cup \Defi{\ILC{\Pc}}{\text{$\Pc$ ball in~$\Ac $}}\right)
\label{:equ:egg-1}
\end{equation}
 
\item if~$\Ac $ is a ball, then
\begin{equation}
\ILC{\Ac}  = 1 +
\max\left(\Defi{\iEB{e}}{\text{$e$ is adjacent to~$\OL{\omega}$}},      
           \ILC{\Ac\bk\OL{\omega}}\right) 
\eqpnt
\label{:equ:egg-2}
\end{equation}
\end{itemize}

It is obvious that equations~\equnm{egg-0}--\equnm{egg-2} reduce 
to~\equnm{egg-3}--\equnm{egg-5} in the case of a `classic' automaton 
whose transitions are all labelled with letters, that is, have 
index~$0$.
\figur{ieb-1} shows two generalised automata and their index.

\begin{figure}[htbp]
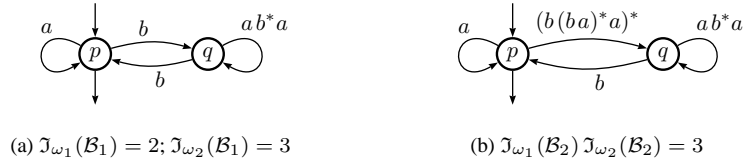

\centering
\subfigure[%
$\operatorname{\mathfrak{I}_{\omega_{1}}}\!\left(\Bc_{1}\right)= 2$;
$\operatorname{\mathfrak{I}_{\omega_{2}}}\!\left(\Bc_{1}\right) = 3$
           ]%
{\e%  
\VCDraw{%
\begin{VCPicture}{(-2.5,-1.4)(5.5,1.4)}
\MediumState
\State[p]{(0,0)}{A}\State[q]{(3,0)}{B}
\Initial[n]{A}\Final[s]{A}
% transitions 
\ArcL{A}{B}{b}\ArcL{B}{A}{b}
\LoopW{A}{a}\LoopE{B}{a\xmd b^{*}a}
\end{VCPicture}
}%
\e}%
\eee
\subfigure[%
$\operatorname{\mathfrak{I}_{\omega_{1}}}\!\left(\Bc_{2}\right)
\operatorname{\mathfrak{I}_{\omega_{2}}}\!\left(\Bc_{2}\right) = 3$
           ]%
{%  
\VCDraw{%
\begin{VCPicture}{(-2,-1.4)(6,1.4)}
\MediumState
\State[p]{(0,0)}{A}\State[q]{(4,0)}{B}
\Initial[n]{A}\Final[s]{A}
% transitions 
\ArcL[.5]{A}{B}{(b\xmd (b\xmd a)^{*}a)^{*}}
\ArcL{B}{A}{b}
\LoopW{A}{a}\LoopE{B}{a\xmd b^{*}a}
\end{VCPicture}
}% 
}%
\medskipneg
\caption{Computation of the index of two automata for the two
  possible orders on the states: \e $\omega_{1}= p<q$, \e $\omega_{2}=q<p$.}
\label{:fig:ieb-1}
\end{figure}

\begin{proof}[Proof of \propo{enl-hau-eto}]
We proceed by induction on the number of states of~$\Ac$. 
The state-elimination method consists in the first place of
transforming~$\Ac$ into an automaton~$\Bc$ by adding two states
to~$\Ac$, an initial state and a final state, and transitions which
are all labelled by the empty word, the index of~$\Bc$ being
equal to that of~$\Ac$. 
By convention, the added states are greater
than all the other states of~$\Ac$ in the order~$\omega$ and are
never removed by the state-elimination method. 
On the other hand, the
transition labels, including those of the new transitions, may be
modified in the course of the state-elimination method.

The base case of the induction is therefore a generalised automaton
with~$3$ states of the form of \figur{egg-pre-1}\dex{a} or~(b).

\begin{figure}[ht]
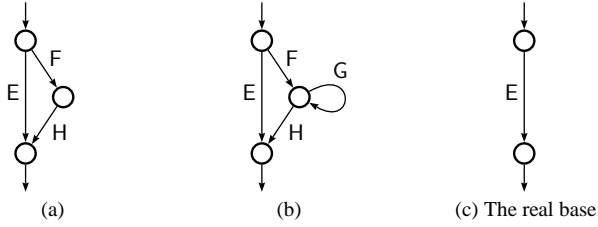

\centering
\e\e
\subfigure[]%
{\e  
\VCDraw{%
\begin{VCPicture}{(-.5,-2)(2,2)}
\SmallState
\State{(0,1.5)}{A}\State{(0,-1.5)}{B}\State{(1,0)}{C}
\Initial[n]{A}\Final[s]{B}
% transitions  
\EdgeR{A}{B}{\Ed}
\EdgeL{A}{C}{\Fd}
\EdgeL{C}{B}{\Hd}
\end{VCPicture}
}%
\e} 
\ee\e
\subfigure[]%
{\e  
\VCDraw{%
\begin{VCPicture}{(-.5,-2)(2,2)}
\SmallState
\State{(0,1.5)}{A}\State{(0,-1.5)}{B}\State{(1,0)}{C}
\Initial[n]{A}\Final[s]{B}
% transitions  
\EdgeR{A}{B}{\Ed}
\EdgeL{A}{C}{\Fd}
\EdgeL{C}{B}{\Hd}
\LoopE{C}{\Gd}
\end{VCPicture}
}%
\e} 
\e
\subfigure[The real base]%
{\ee\ee   
\VCDraw{%
\begin{VCPicture}{(-.5,-2)(.5,2)}
\SmallState
\State{(0,1.5)}{A}\State{(0,-1.5)}{B}
\Initial[n]{A}\Final[s]{B}
% transitions  
\EdgeR{A}{B}{\Ed}
\end{VCPicture}
}%
\ee\ee}
\medskipneg 
\caption{Base case of the induction}
\label{:fig:egg-pre-1}
\end{figure}

In case~(a), $\Bc$ contains no balls and we have
\begin{equation}\ILC{\Bc} =
   \max \bigl ( \jsHaut{\Ed}, \jsHaut{\Fd},
                \jsHaut{\Hd} \bigr )
        = \jsHaut{\Ed + \Fd \cdot \Hd}
        = \jsHaut{\sem{\Bc}} \eqpnt
\end{equation}
In case~(b), the unique state of~$\Bc$ that is neither initial nor
final is a ball whose index\EBz is~$1 + \jsHaut{\Gd}$, and we
have
\begin{equation}\ILC{\Bc} =
   \max \bigl ( \jsHaut{\Ed}, \jsHaut{\Fd},
                \jsHaut{\Hd}, (1 + \jsHaut{\Gd})\bigr )
        = \jsHaut{\Ed +
                  \Fd \cdot \Gd^{*} \cdot \Hd}
        = \jsHaut{\sem{\Bc}} \eqpnt
\end{equation}
Note that this reasoning is the essential part of the induction step
and that for rigour, if not for clarity, we could have taken the
automaton of \figur{egg-pre-1}(c) as our base case, for which the
statement is even more easily verified (and which corresponds to an
automaton~$\Ac $ with no state).

Now let~$\Bc $ be an automaton of the prescribed form with~$n+2$
states, $q$~the smallest state in the order~$\omega $ and~$\Bc'$ the
automaton which results from the first step of the state-elimination
method applied to~$\Bc $ (which consists of the elimination of~$q$).
Because the same states (other than~$q$) are adjacent
in~$\Bc $ as in~$\Bc '$, and since~$q$ is the \emph{smallest element}
in the order~$\omega $, the algorithm for computing the index runs the
same way in~$\Bc $ and~$\Bc'$, \ie  the succession of balls
constructed in each automaton is identical, excluding the examination
of~$q$ in~$\Bc $. It remains to show that the values calculated are
also identical.

Let~$\Pc $ be the smallest ball in~$\Bc $ that strictly
contains~$q$ -- if no such ball exists, take~$\Pc = \Bc$; and
let~$\Pc'$ the `image' of~$\Pc $ in~$\Bc' $. There are two
possible cases: either (a)~$q$ is not the source (and the destination)
of a loop, or (b)~it is, \ie  $q$~is a ball in~$\Bc $ all by itself,
and the label of this loop is an expression~$\Gd$.

The transitions of~$\Pc' $ are either identical to those of~$\Pc $,
or, in case~(a), labelled by products~$\Fd \cdot \Hd$,
where~$\Fd$ and~$\Hd$ are labels of transitions
of~$\Pc$, or, in case~(b), labelled by products~$\Fd \cdot
\Gd^{*} \cdot \Hd$. It therefore follows that, in
case~(a)
\begin{align}
\ILC{\Pc'} & =
\max \bigl( \max \Defi{\iEB{e}}%
       {\text{$e$ does not belong to a ball in~$\Pc'$}} ,
\notag 
\\
      &  \eee \eee \eee \e   
\max \Defi{\ILC{\Qc}}{\text{$\Qc$ ball in~$\Pc'$}}\bigr)
\ee
\notag 
\\[.5ex]
      & =
\max \bigl( \max \Defi{\iEB{e}}%
       {\text{$e$ does not belong to a ball in~$\Pc $}} ,
\notag 
\\[.5ex]
      &  \eee \eee \eee \e   
\max \Defi{\ILC{\Qc}}{\text{$\Qc$ ball in~$\Pc$}}\bigr)
\notag 
\\
      & = \ILC{\Pc} 
\eqpnt
\label{:equ:egg-6}
\end{align}
In case~(b), since 
$\msp\ILC{\{q\}} = 1 + \jsHaut{\Gd}\msp$,
we have:
\begin{align}
\ILC{\Pc'} & =
\max \bigl(\msp \max \Defi{\iEB{e}}%
           {\text{$e$ does not belong to a ball in~$\Pc'$}} ,
%       \max \bigl(\msp \max \{\iEB{e}\jsmid
%          \text{$e$ does not belong to a ball in~$\Pc' $} \},
\notag \\
      & \eee \eee \eee \eee \ee
      \max \Defi{\ILC{\Qc}}{\text{$\Qc$ ball in~$\Pc'$}}\msp\bigr)
      \ee
%              \max\{\ILC{\Qc}\jsmid
%                   \text{$\Qc$ ball in~$\Pc' $} \}\msp \bigr)
\notag \\[.3ex]
      & = 
      \max \bigl(\msp \max \Defi{\iEB{e}}%
                 {\text{$e$ does not belong to a ball in~$\Pc$}},\msp
%       \max \bigl(\msp \max \{ \iEB{e}\jsmid
%            \text{$e$ does not belong to a ball in~$\Pc $}\},\msp
               (1 + \jsHaut{\Gd}),
\notag \\
      & \eee \eee \eee \eee \ee
      \max \Defi{\ILC{\Qc}}{\text{$\Qc$ ball in~$\Pc'$}}\msp\bigr)
%               \max \{\ILC{\Qc}\jsmid 
%                        \text{$\Qc$ ball in~$\Pc' $}\}\msp\bigr)
\notag \\[.3ex]
      & = 
      \max \bigl(\msp \max \Defi{\iEB{e}}%
                 {\text{$e$ does not belong to a ball in~$\Pc$}},\msp
%       \max \bigl(\msp \max \{\iEB{e}\jsmid
%             \text{$e$ does not belong to a ball in~$\Pc $}\},\msp
                     \ILC{\{q\}},
\notag 
\\
      & \eee \eee\ee  
      \max \Defi{\ILC{\Qc}}{\text{$\Qc$ ball in~$\Pc$,
                                  different from~$\{q\}$}}\msp\bigr)
%                \max \{\ILC{\Qc}\jsmid
%                   \text{$\Qc$ ball in~$\Pc $, 
%                   different from~$\{q\}$} \}\msp \bigr)\e
\notag \\[.3ex]
      & = \ILC{\Pc} \eqpnt
\tag*{\equnmnosp{egg-6}$'$}
\end{align}

\noindent
If $\msp\Pc = \Bc \msp$ (and $\msp\Pc' = \Bc'\msp$), 
the equalities~\equnm{egg-6} and~\equnmnosp{egg-6}$'$ become
\begin{equation}
\ILC{\Bc'}  = \ILC{\Bc}
\label{:equ:egg-7}
\end{equation}
% \noindent
which proves the induction and hence the proposition. If not, and
without starting an induction on the number of overlapping balls that
contain~$\{q\}$, we can get from~\equnm{egg-6} to~\equnm{egg-7} by
noting that the transitions of~$\Bc' $ are either identical to
those of~$\Bc $, or correspond to transitions that
are adjacent to~$q$.

In case~(a), the labels of these transitions
(those corresponding to transitions adjacent to~$q$)
are products of
transitions of~$\Bc $: their index is obtained by taking a maximum,
and~\equnm{egg-7} is the result of the identity
$\msp\max (a,b,c) = \max(a, \max (b,c))\msp$.

In case~(b), the labels of the same transitions are, as before, of the
form
$\msp\Fd \cdot \Gd^{*} \cdot \Hd\msp$,
with index
$\msp\max(\jsHaut{\Fd}, \jsHaut{\Hd},
                                   1 + \jsHaut{\Gd})\msp$.
The corresponding transition in~$\Bc $ has label~$\Fd$
(or~$\Hd$); it is inspected in the algorithm for computing the
index when the indices of the transition with label~$\Hd$
(or~$\Fd$) and that of the \emph{ball~$\{q\}$}, with index
$\msp 1 +\jsHaut{\Gd}\msp$, have already been taken into account.
The result, which is~\equnm{egg-7}, follows for the same reason as
above.
\end{proof}
}%

\shortlong{%
\theor{the-Egg} admits a kind of converse stated in the following 
proposition. 
}{%
\theor{the-Egg} admits a kind of converse stated in the following 
proposition whose proof is postponed
to the next section where we build automata from 
expressions. 
}%

\begin{proposition}
\label{:pro:Egg-con}%
\shortlong{%
For every rational expression~$\Ed$ 
there exists an automaton~$\Ac$ which accepts~$\CompExpr{\Ed}$ and the
loop complexity of~$\Ac$ is equal to the star height of~$\Ed$.
}{%
With every rational expression~$\Ed$ 
is associated an automaton which accepts~$\CompExpr{\Ed}$ and whose 
loop complexity is equal to the star height of~$\Ed$.
}%
\end{proposition}

%%%%%%%%%   AE Section 4   %%%%%%%%%%%%%%%%%
%%%             150126                   %%%
%%%%%%%%%%%%%%%%%%%%%%%%%%%%%%%%%%%%%%%%%%%%
\ifcheat\smallskipneg\fi%
\ifcheat\smallskipneg\fi%
\ifcheat\smallskipneg\fi%
\ifcheat\smallskipneg\fi%
\ifcheat\smallskipneg\fi%
\section{From expressions to automata: the $\EtAs$-maps}
\label{:sec:exp-aut}

The transformation of rational expressions into finite automata
establishes \propo{psi-map-fm}. 
It is even more interesting than the 
transformation in the other way, 
both from a theoretical point of view and for practical purposes,
as there are many questions
that cannot be 
answered directly on expressions but require first their 
transformation into automata.

Every expression might be mapped to several automata, 
each of them being computed in different ways.
We distinguish the objects themselves, \ie the computed automata, 
% which we try to give definitions as intrinsic as possible, 
which we try to characterise as intrinsically as possible, 
from the 
algorithms that allow to compute them.
We present two such automata:
the \emph{Glushkov}, or \emph{position}, automaton
and that we rather call \emph{the standard automaton} of the expression,
\index{automaton (of an expression)!position --}%
\index{automaton (of an expression)!Glushkov --}%
\index{Glushkov|see{automaton (of...)}}%
and the \emph{derived-term automaton}, that was first defined by 
{Antimirov}.

The standard automaton may be defined for 
expressions over \emph{any monoid} whereas the derived-term 
automaton will be defined for expressions over a \emph{free monoid} 
only.
In this section however, we restrict ourselves to expressions over a 
free monoid.
We begin with the presentation of two techniques for transforming an 
automaton into another one, that will help us in comparing the various 
automata associated with a given expression.

\ifcheat\smallskipneg\fi%
\ifcheat\smallskipneg\fi%
%%%%
\subsection{Preparation: closure and quotient}
\label{:sec:clo-quo}

\paragraph{Closure}

Automata have been defined (\secti{fin-aut-fm}) as graphs labelled by 
\emph{letters} of an alphabet.
It is known that the family of languages accepted by finite automata 
is not enlarged if transitions labelled by {the empty 
word} --- called \emph{spontaneous transitions} --- are allowed as 
well. 
\index{spontaneous|see{transition}}%
\index{transition!spontaneous}%
The \emph{backward closure} of such an automaton
\index{backward|see{closure}}%
\index{closure (backward)}%
$\msp\Ac=\auta\msp$ is the equivalent automaton 
$\msp\Bc=\aut{Q,A,F,I,U}\msp$ with no spontaneous transitions defined by
{\small
\begin{equation}
F =  \Defi{(p,a,r)}{\ext q \in Q \quantsmsp  
p\pathaut{\unAe}{\Ac}q\msp, \quantsmsp (q,a,r) \in E }
\hspace{.7em}
\text{and}
\hspace{.7em}
U = \Defi{p}{\ext q \in T \quantsmsp p\pathaut{\unAe}{\Ac}q} 
\EqPnt
\notag
\end{equation}}
It is effectively computable, as the determination of~$F$ and~$U$ 
amounts to computing the transitive closure of a finite directed 
graph.

\ifcheat\smallskipneg\fi%
\paragraph{Morphisms and quotient}

Automata are structures; a morphism is a map from an automaton into 
another one which is compatible with this structure.

\ifcheat\smallskipneg\fi%
\begin{definition}
\label{:def:mor-aut}%
Let$\msp\Ac= \aut{Q,A,E,I,T}\msp$ and 
$\msp\Ac'= \aut{Q',A,E',I',T'}\msp$
be two automata.
A map
$\msp\varphi\colon Q\rightarrow Q'\msp$
is a \emph{morphism (of automata)} if:

\thi $\msp\varphi(I)\subseteq I'\msp$,\e\longonly{\\}
\thii $\msp\varphi(T)\subseteq T'\msp$,\e\longonly{\\} 
\thiii $\msp\fa(p,a,q)\in E\quantsmsp 
        \big(\varphi(p),a,\varphi(q)\big)\in E'\msp$.

\smallskip
\noindent
The automaton~$\Ac'$ is \emph{a quotient} of~$\Ac$ if, moreover, 

\thiv $\msp\varphi(Q)= Q'\msp$, that is, 
      $\varphi$ is \emph{surjective},\e\longonly{\\}
\thv $\msp\varphi(I)= I'\msp$,\e\longonly{\\}
\thvi $\msp\varphi^{-1}(T')= T\msp$,

\thvii $\msp\fa(r,a,s)\in E'\quantvrg \fa p\in\varphi^{-1}(r)\quantvrg 
\ext q\in\varphi^{-1}(s)\quantsmsp (p,a,q)\in E\msp$. 
\end{definition}
\index{quotient!of an automaton}%
    
If~$\varphi$ is a morphism, we write
$\msp\varphi\colon\Ac\rightarrow\Ac'\msp$, and the inclusion
$\msp\CompAuto{\Ac}\subseteq\CompAuto{\Ac'}\msp$ holds.
If~$\Ac'$ is a quotient of~$\Ac$, then
$\msp\CompAuto{\Ac}=\CompAuto{\Ac'}\msp$ holds.

\defin{mor-aut} generalises the classical notion of quotient of 
complete deterministic automata to arbitrary automata. 
Every automaton~$\Ac$ admits a \emph{minimal 
quotient}, which is a quotient of every quotient of~$\Ac$.
In contrast with the case of deterministic automata, the minimal 
quotient of~$\Ac$ is  
canonically associated with~$\Ac$, not  with the language 
accepted by~$\Ac$.

\ifcheat\smallskipneg\fi%former 
\ifcheat\smallskipneg\fi%
%%%%
\subsection{The standard automaton of an expression}
\label{:sec:sta-aut-exp}

The first automaton we associate with an expression~$\Ed$, 
which we write~$\Stan{\Ed}$ and which plays a 
central role in our presentation,
% \shortlong{%will
% has been first\footnote{%
%    Independently and even earlier, McNaughton and Yamada computed 
%    directly the \emph{determinisation} of~$\Sc_{\,\,\Ed}$ in  
%    the paper \cite{MNauYama60} that we already quoted (see Notes).} 
% defined by Glushkov (in~\cite{Glus61}).
% }{%
has first been defined by Glushkov (in~\cite{Glus61}).
For the same purpose, McNaughton and Yamada computed the 
determinisation of~$\Stan{\Ed}$ 
in their paper \cite{MNauYama60} that we already quoted.
% }%
In order to give an intrinsic description of~$\Stan{\Ed}$,
we define a restricted class of 
automata, and then show that rational operations on sets can be 
lifted on the automata of that class.

\ifcheat\smallskipneg\fi%
\subsubsection{Operations on standard automata}
\label{:sec:sta-aut}

An automaton is \emph{standard} if it has only one 
initial state, which is the end of no 
transition.
\index{standard|see{automaton}}%
\index{automaton!standard --}%
\figur{sta-aut} shows a standard automaton, both as a sketch, and 
under the matrix form.
The definition does not forbid the initial state~$i$
from also being final and the scalar $\msp c$, equal to~$\zed$ 
or~$\und$, is the \emph{constant term}
of $\msp\CompAuto{\Ac}\msp$.

\TinyPicture%
\begin{figure}[ht]
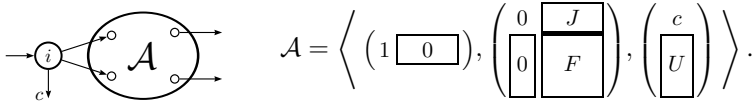

\centering 
\VCDraw{%
\begin{VCPicture}{(-4.5,-1.3)(2.5,1.3)}
\VCPut{(0,0)}%
   {%
    \BigAuto{\Ac}%
    \State[i]{(-3,0)}{A}%
    \VCPut{(-1,0)}{\VSState{(0,0.646)}{B1}\VSState{(0,-0.646)}{B2}}%
    \VCPut{(1,0)}{\VSState{(0,.736)}{C2}\VSState{(0,-0.736)}{C3}}%
   }%
\Initial[w]{A}\FinalR{s}{A}{c}
\Final[e]{C2}\Final[e]{C3}
\Edge{A}{B1}\Edge{A}{B2}
\end{VCPicture}%
}%
\ee
$\msp
\begin{displaystyle}
\Ac =\aut{\redmatu{\lignedblvs{1}{0}},
          \redmatu{\matriceddblvs{0}{J}{0}{F}},
          \redmatu{\vecteurdblvs{c}{U}}} 
\end{displaystyle}
\msp$.
\caption{A standard automaton}
\label{:fig:sta-aut}
\end{figure}
\SmallPicture
\ifcheat\shortonly{\miniskipneg\miniskipneg\miniskipneg}\fi%

Every automaton is
equivalent to a {standard} one.
\shortlong{Their special form}%
{More important for our purpose, their special form}
allows to define \emph{operations} on standard 
automata that are parallel to the 
\emph{rational operations}.
Let~$\Ac$ (as in~\figur{sta-aut}) and~$\Bc$ (with obvious 
notation) be two standard automata;
the following standard automata are defined:

\ifcheat\shortonly{\smallskipneg}\fi%
\begin{equation}
    {\Ac}+{\Bc} =
\aut{\redmatu{\lignetblblvs{1}{0}{0}},
     \redmatu{\matricettblblvs{0}{J}{K}%
                              {0}{F}{0}%
                              {0}{0}{G}},
     \redmatu{\vecteurtblblvs{c+d}{U}{V}}}\eqvrg
\label{:equ:sta-aut-sum}
\end{equation}

\begin{equation}
    {\Ac}\matmul{\Bc} = 
\aut{\redmatu{\lignetblblvs{1}{0}{0}},
     \redmatu{\matricettblblvs{0}{J}{c\xmd K}%
                              {0}{F}{U\matmul K}%
                              {0}{0}{G}},
     \redmatu{\vecteurtblblvs{c\xmd d}{Ud}{V}}}\eqvrg
\label{:equ:sta-aut-pro}
\end{equation}

\begin{equation}
    {\Ac}^{*} =
\aut{\redmatu{\lignedblvs{1}{0}},
     \redmatu{\matriceddblvs{0}{J}{0}{H}},
     \redmatu{\vecteurdblvs{1}{U}}}\eqvrg
\label{:equ:sta-aut-sta}
\end{equation}
where $\msp H= U\matmul J + F\msp$.
The use of the constants $\msp c\msp $ and $\msp d\msp $ allows 
a uniform treatment of
the cases whether the initial states of~$\Ac$ and~$\Bc$ are final or not.
\longonly{%

\smallskip
These constructions are shown at \figur{sta-aut-ope}.

\begin{figure}[ht]
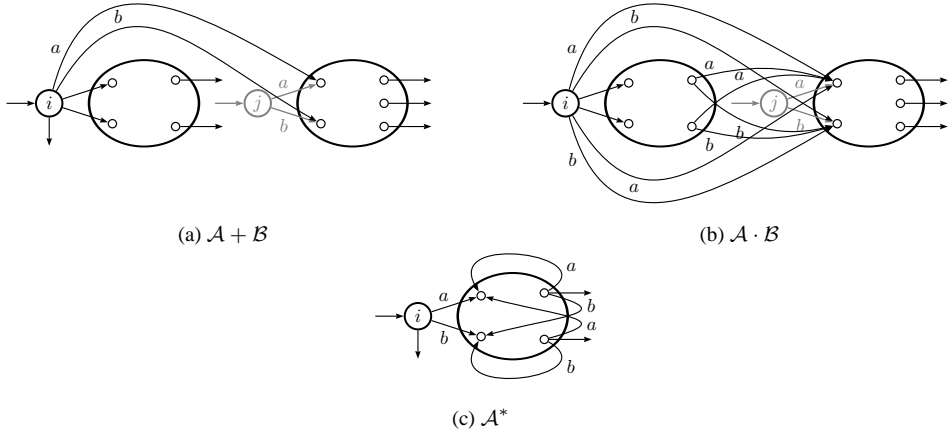

    \TinyPicture
\centering
\subfigure[$\displaystyle{\Ac + \Bc}$]{%
\VCDraw{%
\begin{VCPicture}{(-4,-3)(9,3)}
% First automaton
\VCPut{(0,0)}%
   {%
    \BigAuto{}%\Ac
    \State[i]{(-3,0)}{A}%
    \VCPut{(-1,0)}{\VSState{(0,0.646)}{B1}\VSState{(0,-0.646)}{B2}}%
    \VCPut{(1,0)}{\VSState{(0,.736)}{C2}\VSState{(0,-0.736)}{C3}}%
   }%
% \ChgStateLabelScale{1.8}%
% \VCPutStateLabel{(-5,1.5)}{\Ac + \Ac'}%{1}
\Initial[w]{A}\Final[s]{A}
\Final[e]{C2}\Final[e]{C3}
\Edge{A}{B1}\Edge{A}{B2}
% \EdgeL[.35]{A}{B1}{a}\EdgeR[.35]{A}{B2}{b}
% Second automaton
\VCPut{(6.6,0)}%
   {%
    \BigAuto{}%\Ac
    \DimState\State[j]{(-3,0)}{BB}\RstState%
    \VCPut{(-1,0)}{\VSState{(0,0.646)}{BB1}\VSState{(0,-0.646)}{BB2}}%
    \VCPut{(1,0)}{\VSState{(0,0)}{CC1}%
                  \VSState{(0,.736)}{CC2}\VSState{(0,-0.736)}{CC3}}%
   }%
\DimEdge
\Initial[w]{BB}
\EdgeL[.35]{BB}{BB1}{a}\EdgeR[.35]{BB}{BB2}{b}
\RstEdge 
\Final{CC1}\Final{CC2}\Final{CC3}
%
% \ArcL[.1]{C2}{BB1}{a}\ArcR[.1]{C3}{BB2}{b}
% \LArcL{C3}{BB1}{a}\LArcR{C2}{BB2}{b}
% \VCurveR[.1]{angleA=-75,angleB=-150,ncurv=1.1}{A}{BB2}{b}
\VCurveL[.1]{angleA=75,angleB=150,ncurv=1.1}{A}{BB1}{a}
% \VCurveR[.3]{angleA=-60,angleB=-150,ncurv=1.1}{A}{BB1}{a}
\VCurveL[.3]{angleA=60,angleB=150,ncurv=1.1}{A}{BB2}{b}
\end{VCPicture}%
}%
}%
\eee
\subfigure[$\displaystyle{\Ac\matmul\Bc}$]{%
\VCDraw{%
\begin{VCPicture}{(-4,-3)(9,3)}
% First automaton
\VCPut{(0,0)}%
   {%
    \BigAuto{}%\Ac
    \State[i]{(-3,0)}{A}%
    \VCPut{(-1,0)}{\VSState{(0,0.646)}{B1}\VSState{(0,-0.646)}{B2}}%
    \VCPut{(1,0)}{\VSState{(0,.736)}{C2}\VSState{(0,-0.736)}{C3}}%
   }%
% \ChgStateLabelScale{1.8}%
% \VCPutStateLabel{(-5,1.5)}{\Ac\cdot\Ac'}%
\Initial[w]{A}%\FinalR{s}{A}{1}
% \Final[e]{C2}\Final[e]{C3}
\Edge{A}{B1}\Edge{A}{B2}
% \EdgeL[.35]{A}{B1}{a}\EdgeR[.35]{A}{B2}{b}
% Second automaton
\VCPut{(6.6,0)}%
   {%
    \BigAuto{}%\Ac
    \DimState\State[j]{(-3,0)}{BB}\RstState%
    \VCPut{(-1,0)}{\VSState{(0,0.646)}{BB1}\VSState{(0,-0.646)}{BB2}}%
    \VCPut{(1,0)}{\VSState{(0,0)}{CC1}%
                  \VSState{(0,.736)}{CC2}\VSState{(0,-0.736)}{CC3}}%
   }%
\DimEdge
\Initial[w]{BB}
\EdgeL[.35]{BB}{BB1}{a}\EdgeR[.35]{BB}{BB2}{b}
\RstEdge 
\Final{CC1}\Final{CC2}\Final{CC3}
\ArcL[.1]{C2}{BB1}{a}\ArcR[.1]{C3}{BB2}{b}
\LArcL{C3}{BB1}{a}\LArcR{C2}{BB2}{b}
\VCurveR[.1]{angleA=-75,angleB=-150,ncurv=1.1}{A}{BB2}{b}
\VCurveL[.1]{angleA=75,angleB=150,ncurv=1.1}{A}{BB1}{a}
\VCurveR[.3]{angleA=-60,angleB=-150,ncurv=1.1}{A}{BB1}{a}
\VCurveL[.3]{angleA=60,angleB=150,ncurv=1.1}{A}{BB2}{b}
\end{VCPicture}%
}%
}%

\medskipneg
\medskipneg
\subfigure[$\displaystyle{\Ac^{*}}$]{%
\VCDraw{%
\begin{VCPicture}{(-4.5,-2)(2.5,2)}
\VCPut{(0,0)}%
   {%
    \BigAuto{}%\Ac
    \State[i]{(-3,0)}{A}%
    \VCPut{(-1,0)}{\VSState{(0,0.646)}{B1}\VSState{(0,-0.646)}{B2}}%
    \VCPut{(1,0)}{\VSState{(0,.736)}{C2}\VSState{(0,-0.736)}{C3}}%
%     \ChgStateLabelScale{1.8}%
%     \VCPutStateLabel{(-3.5,1.5)}{\Ac^{*}}%
   }%
\Initial[w]{A}\Final[s]{A}
\Final[e]{C2}\Final[e]{C3}
% \Edge{A}{B1}\Edge{A}{B2}{1}
\EdgeL[.35]{A}{B1}{a}\EdgeR[.35]{A}{B2}{b}
\VCurveR[.2]{angleA=30,angleB=120,ncurv=2}{C2}{B1}{a}
\VCurveL[.3]{angleA=-15,angleB=15,ncurv=2}{C2}{B2}{b}
\VCurveL[.2]{angleA=-30,angleB=-120,ncurv=2}{C3}{B2}{b}
\VCurveR[.3]{angleA=15,angleB=-15,ncurv=2}{C3}{B1}{a}
\end{VCPicture}%
}%
}%
\medskipneg
\caption{Operations on  standard automata}
\label{:fig:sta-aut-ope}
\end{figure}
\SmallPicture

}%
Straightforward computations show that
$\msp\CompAuto{(\Ac+\Bc)}= \CompAuto{\Ac}+\CompAuto{\Bc}\msp$,
$\msp\CompAuto{(\Ac\matmul\Bc)}=\CompAuto{\Ac}\matmul\CompAuto{\Bc}\msp$
and
$\msp\CompAuto{(\Ac^{*})}=\CompAuto{\Ac}^{*}\msp$.

With every rational expression~$\Ed$ and by induction on its depth, 
we thus canonically associate a  standard automaton, which we 
write~$\Stan{\Ed}$ and which we call \emph{the} standard automaton 
of~$\Ed$.
\index{automaton (of an expression)!standard --}%
% The same induction also shows that
% the \emph{dimension} of~$\Stan{\Ed}$ is
% $\msp\LitL{\Ed} + 1\msp $
% and that~$\Stan{\Ed}$ answers the question,
% that is, the map
% $\msp\Ed\mapsto\Stan{\Ed}\msp$ is a $\EtAs$-map:
The induction and the computations show 
% that the \emph{dimension} of~$\Stan{\Ed}$ is
% $\msp\LitL{\Ed} + 1\msp $ and 
% that~$\Stan{\Ed}$ answers the question,
% that is, 
that the map
$\msp\Ed\mapsto\Stan{\Ed}\msp$ is a $\EtAs$-map:

\begin{proposition}
\label{:pro:sta-aut-exp}%
If~\/$\Ed$ is a  rational expression over~$\Ae$, then
$\msp\CompAuto{\Stan{\Ed}}=\CompExpr{\Ed}\msp$.
\end{proposition}

The inductive construction of~$\Stan{\Ed}$ also implies:

\begin{property}
\label{:pty:sta-aut-dim}
If~$\Ed$ is a  rational expression, 
the dimension of~$\Stan{\Ed}$ is
$\msp\LitL{\Ed} + 1\msp$.
\end{property}

\begin{example}
    \label{:exa:sta-aut-1}%
\figur{sta-aut-1} shows~$\Stan{\Ed_{1}}$, 
where
$\msp
\Ed_{1}= (a^{*}b+b\xmd b^{*}a^{*})^{*}
\msp$.
\end{example}

\smallskipneg\smallskipneg%
\begin{figure}[ht]
\centering
\setlength{\lga}{4cm}%
\VCDraw{%
\begin{VCPicture}{(-8,-1)(8,3.2)}
\SmallState
% etats
\VCPut{(-1.5\lga,0)}{%
\VCPut{(0,0)}{%-\SQRTwo\lga
\State{(0,\SQRTwo\lga)}{A}\State{(0,0)}{B}\State{(\lga,0)}{C}}%
\VCPut{(3\lga,0)}{%\SQRTwo\lga
\State{(0,\SQRTwo\lga)}{D}\State{(0,0)}{E}\State{(-\lga,0)}{F}}%
}%
\Initial[w]{A}\Final[sw]{A}
\Final[s]{C}\Final[s]{F}
% transitions 
\EdgeL{A}{B}{a}
\ArcL{A}{C}{b}
\ArcL{B}{C}{b}\ArcL{C}{B}{a}
\LoopW{B}{a}\LoopN[.5]{C}{b}
%[.75]
\EdgeL{D}{E}{b}\EdgeL{E}{F}{a}
\ArcL{D}{F}{a}\ArcL[.25]{F}{D}{b}
\LoopE{E}{b}
% [.75]
\EdgeL{A}{D}{b}
\ArcL{C}{D}{b}
\LArcL[.2]{F}{B}{a}
\EdgeR{F}{C}{b}
\end{VCPicture}
}%
\caption{The automaton~$\Stan{\Ed_{1}}$. }
\medskipneg
\label{:fig:sta-aut-1}
\end{figure}

The example of
$\msp
\Ed =
\left(\left(\left(a^{*}+b^{*}\right)^{*}+c^{*}\right)^{*}+d^{*}\right)^{*}\ldots
\msp$
shows that the direct computation of~$\Stan{\Ed}$ 
by~\equnm{sta-aut-sum}--\equnm{sta-aut-sta} leads to an  
algorithm whose complexity is \emph{cubic} in~$\LitL{\Ed}$.
The quest for a better algorithm leads to a construction that is 
interesting \textit{per se}.

\ifcheat\smallskipneg\fi%
\subsubsection{The star-normal form of an expression}
\label{:sec:sta-nor-for}

The \emph{star-normal form} of an expression has been defined by
Br\"uggemann-Klein (in~\cite{Brug93}) in order to design a quadratic 
algorithm for the computation of the standard automaton of an 
expression. 
The interest of this notion certainly goes beyond that complexity 
improvement.
 
\ifcheat\smallskipneg\fi%
\begin{definition}[\cite{Brug93}]
    \label{:def:sta-nor-for}%
A rational expression~$\Ed$ is \emph{in star-normal form} (SNF) 
if and only if for 
any~$\Fd$ such that~$\Fd^*$ is a subexpression of~$\Ed$, 
\index{expression!star-normal form|see{star-normal form}}%
\index{star-normal form!expression in --}%
$\msp\TermCst{\Fd}=0\msp$.\footnote{%
   The definition, as well as the construction, have been slightly 
   modified from the original, for simplification.}
\end{definition}

\ifcheat\smallskipneg\fi%
Two operators on expressions, written 
$\bullet$ and $\squa$,
are defined by a mutual recursion
% an intertwined induction 
on the depth of the expression 
% in order to compute (and define by this computation) 
that defines and allows to compute
\emph{the star-normal form} of the 
\index{star-normal form!of an expression}%
expression.

\smallskipneg\smallskipneg%
\shortlong{%
\begin{gather}
\cb{\zed} = \rn{\zed} = \zed \EqVrgInt 
\cb{\und} = \rn{\und} =  \zed\EqVrgInt
\e\fa a \in A \quantsmsp \cb{a} = \rn{a} = a
\eqvrg 
\e        
\notag
\\
\cb{(\Fd + \Gd)} = \Fcb + \Gcb
\EqVrgInt
\cb{(\Fd^*)} = \Fcb
\EqVrgInt
\cb{(\Fd \cdot \Gd)} = 
\left\lbrace 
\begin{array}{ll}
  \Fcb + \Gcb & \text{if}\e \TermCst{\Fd}=\TermCst{\Gd}=1 \eqvrg\\
  \Frn \cdot \Grn & $otherwise $ \eqvrg
\end{array} \right. 
\notag
\\[0.5ex]
\rn{(\Fd + \Gd)} = \Frn + \Grn
\EqVrgInt 
\rn{(\Fd \cdot \Gd)} = \Frn \cdot \Grn
\EqVrgInt 
\rn{(\Fd^*)} = (\Fcb)^*
\eqpnt
\notag
\end{gather}
}{%
\begin{align}
\cb{\zed} = \zed \EqVrgInt \cb{\und} &= \zed\EqVrgInt
% \text{for all } 
\fa a \in A \quantsmsp \cb{a} = a
\eqvrg 
\eee         
\label{:equ:cb-1}
\\
\cb{(\Fd + \Gd)} &= \Fcb + \Gcb
\eqvrg 
\label{:equ:cb-2}
\\
\cb{(\Fd \cdot \Gd)} 
   &= \left\lbrace 
         \begin{array}{ll}
           \Fcb + \Gcb & \e\text{if}\e \TermCst{\Fd}=\TermCst{\Gd}=1 \eqvrg\\
           \Frn \cdot \Grn & \e$otherwise $ \eqvrg
         \end{array} \right. 
\label{:equ:cb-3}
\\
\cb{(\Fd^*)} &= \Fcb
\eqpnt 
\label{:equ:cb-4}
\\[1ex]
% \end{align}
% 
% 
% 
% \begin{align}
\rn{\zed} = \zed\EqVrgInt \rn{\und} &= \und \EqVrgInt 
% \text{for all } 
\fa a \in A\quantsmsp \rn{a} = a
\eqvrg 
\eee         
\label{:equ:rn-1}
\\
\rn{(\Fd + \Gd)} &= \Frn + \Grn
\eqvrg 
\label{:equ:rn-2}
\\
\rn{(\Fd \cdot \Gd)} &= \Frn \cdot \Grn
\eqvrg 
\label{:equ:rn-3}
\\
\rn{(\Fd^*)} &= (\Fcb)^*
\eqpnt 
\label{:equ:rn-4}
\end{align}
}%

\ifcheat\smallskipneg\fi%
\begin{example}
    \label{:exa:sta-nor-for-1}%   \PushLine{}{}\PushLine
Let
$\msp\Ed_{2}= (a^{*}b^{*})^{*}\msp$.
Then

$\msp\rn{\Ed_{2}}= \left(\cb{\left(a^{*}b^{*}\right)}\right)^{*}
  = \left(\cb{\left(a^{*}\right)}+\cb{\left(b^{*}\right)}\right)^{*}
                 = \left(\cb{(a)}+\cb{(b)}\right)^{*}
                 = (a+b)^{*}\msp$.
\end{example}

\ifcheat\smallskipneg\fi%
\begin{theorem}[\cite{Brug93}]
    \label{:the:BK-snf}%
For any expression~$\Ed$, $\Ern$ is in star-normal form and 
$\msp\Stan{\Ern}=\Stan{\Ed}\msp$.
\end{theorem}

\longonly{%
\theor{BK-snf} implies in particular that~$\Ern$ is 
equivalent to~$\Ed$.
It relies on three computations on Boolean standard automata which are the 
direct consequence of the 
formulas~\equnm{sta-aut-sum}--\equnm{sta-aut-sta}. 
and of the idempotency identity~$\IdRI$.
 
\begin{property}
    \label{:pty:ope-sta-aut}
Let~$\Ac$, $\Ac'$, $\Bc$ and~$\Bc'$ be four (Boolean) standard automata.
\begin{enumerate}
    \item  
If 
$\msp\Ac^{*}={\Ac'}^{*}\msp$
and
$\msp\Bc^{*}={\Bc'}^{*}\msp$
then
$\msp\left(\Ac+\Bc\right)^{*}=\left(\Ac'+\Bc'\right)^{*}\msp$.
% \begin{equation}
%     \left(\Ac+\Bc\right)^{*} =
%     \left(\Ac'+\Bc'\right)^{*}
%     \eqpnt
%     \notag
% \end{equation}

    \item  
If
$\msp\TermCst{\CompAuto{\Ac}}=\TermCst{\CompAuto{\Bc}}=1\msp$, then
$\msp\left(\Ac+\Bc\right)^{*}=\left(\Ac\matmul\Bc\right)^{*}\msp$.

    \item  
$\msp\Ac^{*}=\left(\Ac^{*}\right)^{*}\msp$.
\end{enumerate}
\end{property}
}%

\longonly{%
\begin{proof}
\thi 
Let
\begin{equation}
\Ac =\StanAuto{J}{F}{c}{U}
    \e\text{and}\e
\Ac' =\StanAuto{J'}{F'}{c'}{U'}
\eqpnt
\notag
%     \label{}
\end{equation}
The hypothesis implies
\begin{equation}
\Ac^{*} = \StanAuto{J}{U\cdot J +F}{1}{U} =    
{\Ac'}^{*} = \StanAuto{J'}{U'\cdot J' +F'}{1}{U'}
\eqvrg
\notag
%     \label{}
\end{equation}
hence 
$\msp J=J'\msp$,
$\msp U=U'\msp$
and $\msp U\cdot J +F=U'\cdot J' +F'\msp$.
Accordingly, if we have
\begin{equation}
    \Bc =\StanAuto{K}{G}{d}{V}
    \e\text{and}\e
    \Bc' =\StanAuto{K'}{G'}{d'}{V}
\eqvrg
\notag
%     \label{}
\end{equation}
then
$\msp K=K'\msp$,
$\msp V=V'\msp$
and $\msp V\cdot K +G=V'\cdot K' +G'\msp$
hold.
From~\equnm{sta-aut-sum} and~\equnm{sta-aut-sta} follow
\begin{align}
(\Ac+\Bc)^{*} 
   & = \StanAutoSP{J}{K}{U\cdot J +F}{U\cdot K}%
              {V\cdot K}{V\cdot K +G}{1}{U}{V}
\notag \\[.5ex]
   & = \StanAutoSP{J'}{K'}{U'\cdot J' +F'}{U'\cdot K'}%
              {V'\cdot K'}{V'\cdot K' +G'}{1}{U'}{V'}
     = (\Ac'+\Bc')^{*}
\eqpnt
\notag
%     \label{}
\end{align}
% \end{proof}
% 
% \begin{property}
%     \label{:pty:sum-pro-sta}
% Let~$\Ac$ and~$\Bc$ be two (Boolean) standard automata.
% 
% If
% $\msp\TermCst{\CompAuto{\Ac}}=\TermCst{\CompAuto{\Bc}}=1\msp$, then
% \begin{equation}
%     \left(\Ac+\Bc\right)^{*} =
%     \left(\Ac\matmul\Bc\right)^{*}
%     \eqpnt
%     \notag
% \end{equation}
% \end{property}
%  
% \begin{proof}

\thii
With the same notation as before
% in \prpri{sum-sta}, 
we have on one hand-side
\begin{align}
\Ac+\Bc 
   & = \StanAutoSP{J}{K}{F}{0}%
                  {0}{G}{1}{U}{V}
   \ee\text{and then}
\notag \\[.5ex]
(\Ac+\Bc)^{*} 
   & = \StanAutoSP{J}{K}{U\cdot J +F}{U\cdot K}%
              {V\cdot K}{V\cdot K +G}{1}{U}{V}
\eqvrg
\notag
%     \label{}
\end{align}
and on the other
\begin{align}
\Ac\matmul\Bc 
   & = \StanAutoSP{J}{K}{F}{U\cdot K}%
                  {0}{G}{1}{U}{V}
   \ee\text{and then}
\notag \\[.5ex]
(\Ac\matmul\Bc)^{*} 
   & = \StanAutoSP{J}{K}{U\cdot J +F}{U\cdot K + U\cdot K}%
              {V\cdot K}{V\cdot K +G}{1}{U}{V}
\eqvrg
\notag
%     \label{}
\end{align}
and the equality follows from the idempotency identity~$\IdRI$.
% \end{proof}
%     
% 
% \begin{property}
%     \label{:pty:sta-sta}
% Let~$\Ac$ be a (Boolean) standard automaton.
% Then:
% \begin{equation}
%     \left(\Ac^{*}\right)^{*} =
%     \Ac^{*}
%     \eqpnt
%     \notag
% \end{equation}
% \end{property}
%  
% \begin{proof}

\thiii
Again with the same notation,
% as in \prpri{sum-sta}, 
we have: 
\begin{equation}
\Ac^{*} = \StanAuto{J}{U\cdot J +F}{1}{U}
\e\text{and}\e
\left(\Ac^{*}\right)^{*} = \StanAuto{J}{U\cdot J +U\cdot J +F}{1}{U}   
\notag
%     \label{}
\end{equation}
and the equality follows from the idempotency identity~$\IdRI$.
\end{proof}
}%

\longonly{%
\begin{proof}[Proof of \theor{BK-snf}]
We establish by a simultaneous induction that the following 
\emph{two} statements hold:
\begin{align}
\Ern \e\text{is in SNF} &\ee\text{and}\ee 
\Stan{\Ern} = \Stan{\Ed} 
\label{:equ:BK-snf-1}
\\
\Ecb \e\text{is in SNF} &\ee\text{and}\ee 
(\Stan{\Ecb})^* = \Stan{\Ed^{*}} 
\eqpnt
\label{:equ:BK-snf-2}
\end{align}
Both~\equnm{BK-snf-1} and~\equnm{BK-snf-2} clearly hold for the base 
clauses.

\begin{itemize}
\item
$\msp\Ed = \Fd + \Gd\msp$  
\begin{itemize}
    \item 
    $\msp\Ern = \Frn + \Grn\msp$ is in star-normal form by induction.
Moreover,
\begin{align}
\Stan{\Ern} =  \Stan{\Frn}+\Stan{\Grn} & = \Stan{\Fd}+\Stan{\Gd}
&\ee&\text{by induction} 
\notag
\\
& =   \Stan{\Ed}
&\ee&\text{since $\msp\Ed = \Fd + \Gd\msp$} 
\eqpnt 
\notag
\end{align}

\item  
$\msp\Ecb = \Fcb + \Gcb\msp$ is in star-normal form by induction.
Moreover,
\begin{align}
\e\left(\Stan{\Ecb}\right)^{*} =  \left(\Stan{\Fcb}+\Stan{\Gcb}\right)^{*} 
   & = \left(\Stan{\Fd}+\Stan{\Gd}\right)^{*}
&\ee&\text{by induction and by \prpri{ope-sta-aut}~\dex{i}} 
\notag
\\
& =   \Stan{\Ed^{*}}
&\ee&\text{since $\msp\Ed = \Fd + \Gd\msp$} 
\eqpnt 
\notag
\end{align}
\end{itemize}

\item
$\msp\Ed = \Fd \matmul \Gd\msp$ 

\begin{itemize}
\item  
$\msp\Ern = \Frn \matmul \Grn\msp$ is in star-normal form by induction.
Moreover,
\begin{align}
    \Stan{\Ern} = \Stan{\Frn}\matmul\Stan{\Grn} & = \Stan{\Fd}\matmul\Stan{\Gd} 
&\ee&\text{by induction} 
\notag
\\
& =   \Stan{\Ed}
&\ee&\text{since $\msp\Ed = \Fd \matmul \Gd\msp$} 
\eqpnt 
\notag
\end{align}

\item  
    $\msp\Ecb = \Fcb + \Gcb\msp$ or $\msp\Ecb = \Frn\matmul\Grn\msp$ 
    is in star-normal form by induction.
    Moreover:
    
\thi if $\msp\TermCst{\Fd}=\TermCst{\Gd}=1\msp$, then
\begin{align}
\e\left(\Stan{\Ecb}\right)^{*} =  \left(\Stan{\Fcb}+\Stan{\Gcb}\right)^{*} 
   & = \left(\Stan{\Fd}+\Stan{\Gd}\right)^{*}
&\ee&\text{by induction and by \prpri{ope-sta-aut}~\dex{i}} 
\notag
\\
& =   \left(\Stan{\Fd}\matmul\Stan{\Gd}\right)^{*}
&\ee&\text{by \prpri{ope-sta-aut}~\dex{ii}}  
\notag
\\
& =   \Stan{\Ed^{*}}
&\ee&\text{since $\msp\Ed = \Fd \matmul \Gd\msp$} 
\eqpnt 
\notag
\end{align}

\thii if $\msp\TermCst{\Fd}\TermCst{\Gd}=0\msp$, then
\begin{align}
\e\left(\Stan{\Ecb}\right)^{*} =  \left(\Stan{\Frn}\matmul\Stan{\Grn}\right)^{*} 
   & = \left(\Stan{\Fd}\matmul\Stan{\Gd}\right)^{*}
&\ee&\text{by induction} 
\notag
\\
& =   \Stan{\Ed^{*}}
&\ee&\text{since $\msp\Ed = \Fd \matmul \Gd\msp$} 
\eqpnt 
\notag
\end{align}
\end{itemize}

\item
$\msp\Ed = \Fd^{*} \msp$  
 \begin{itemize}
    \item  
    $\msp\Ern = \left(\Fcb\right)^{*}\msp$ is in star-normal form by 
    induction and since $\msp\TermCst{\Fcb}=0\msp$.
    Moreover,
\begin{align}
\eee\Stan{\Ern} =  \left(\Stan{\Fcb}\right)^{*}& = \Stan{\Fd^{*}} 
&\ee&\text{by induction and~\equnm{BK-snf-2}} \ee
\notag
\\
& =   \Stan{\Ed}
&\ee&\text{since $\msp\Ed = \Fd^{*}\msp$} 
\eqpnt 
\notag
\end{align}

\item 
 $\msp\Ecb = \cb{\left(\Fd^{*}\right)} = \Fcb\msp$ is in star-normal form by induction.
Moreover
\begin{align}
\e\left(\Stan{\Ecb}\right)^{*} =  \left(\Stan{\Fcb}\right)^{*} 
   & = \Stan{\Fd^{*}}
&\ee&\text{by induction} 
\notag
\\
& =   \Stan{\Ed^{*}}
&\ee&\text{since $\msp\Ed = \Fd\msp$ and by 
\prpri{ope-sta-aut}~\dex{iii}} 
\eqpnt 
\notag
\end{align}
\end{itemize}
\end{itemize}

\end{proof}
}

\ifcheat\smallskipneg\fi%
As the computation of~$\rn{\Ed}$ is linear in~$\LitL{\Ed}$, the goal 
is achieved by the following:

\ifcheat\smallskipneg\fi%
\begin{theorem}[\cite{Brug93}]
    \label{:the:BK-snf-cpl}%
\shortlong{%
The computation of~$\Stan{\Ern}$ 
has a quadratic complexity in~$\LitL{\Ed}$.
}{%
Let~$\Ed$ be a rational expression in star-normal form.
Then, the inductive computation of~$\Stan{\Ed}$ 
by~\equnm{sta-aut-sum}--\equnm{sta-aut-sta} has a quadratic 
complexity in~$\LitL{\Ed}$.
}%
\end{theorem}

\longonly{%
\begin{proof}
	{\large \fbox{to be done}}
\end{proof}
}%

\ifcheat\smallskipneg\fi%
\ifcheat\smallskipneg\fi%
\subsubsection{The Thompson automaton}
\label{:sec:oth-rel-aut}

A survey on~$\EtAs$-maps cannot miss out the method due to Thompson 
\cite{ThompK68}. 
It was designed to be directly implementable as a program, primarily 
for searching with rational expressions in text.
It is based on the use of spontaneous transitions. 
\figur{Tho-1} shows the basic steps of the construction, which, by 
induction, associates with an expression~$\Ed$ a
unique (and well-defined) automaton~$\Thom{\Ed}$.
\index{automaton (of an expression)!Thompson --}%
\index{Thompson|see{automaton (of...)}}%
% This construction corresponds indeed to another way of
% defining the standard automaton:
It is remarkable that this construction corresponds indeed to another 
way of defining the standard automaton.

\ifcheat\smallskipneg\fi%
\begin{proposition}
\label{:pro:Tho-sta}
The \emph{backward closure} of~$\Thom{\Ed}$ is equal to~$\Stan{\Ed}$. 
\end{proposition}
% \ifcheat\clearpage\fi%

\begin{figure}[ht]
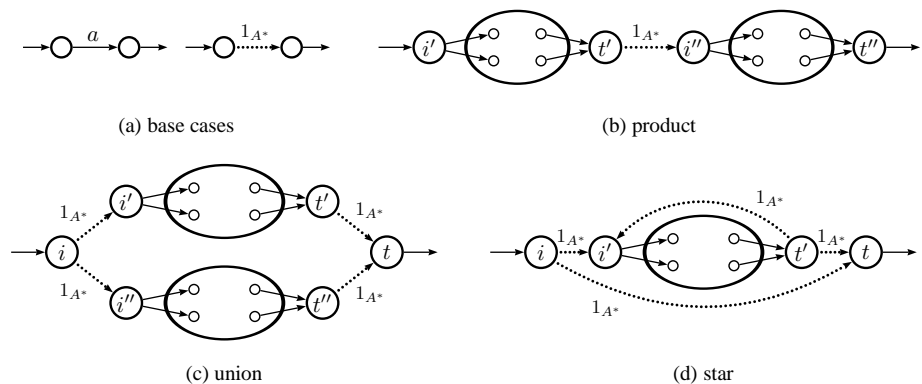

\centering
\subfigure[base cases]{% of the construction
\VCDraw[0.9]{%
\begin{VCPicture}{(-1.4,-1.2)(3.4,1.2)}
% etats
\SmallState
\State{(0,0)}{A}\State{(2,0)}{B}
\Initial{A}\Final{B}
% transitions 
\EdgeL{A}{B}{a}
\end{VCPicture}%
}
\VCDraw[0.9]{%
\begin{VCPicture}{(-1.4,-1.2)(3.4,1.2)}
% etats
\SmallState
\State{(0,0)}{A}\State{(2,0)}{B}
\Initial{A}\Final{B}
% transitions 
\TransSpont\EdgeL{A}{B}{\unAe}
\end{VCPicture}%
}
}
\e \
\subfigure[product]{%
\VCDraw[0.9]{%
\begin{VCPicture}{(-8,-1.2)(8,1.2)}
% premier automate
\rput(-3.9,0){%
% Automate central
\FixStateDiameter{2cm}%
\scalebox{\SQRTwoRatio}{\State{(0,0)}{A}}%
% \scalebox{\GoldMeanRatio}{\State{(0,0)}{A}}%
% param etats
\MediumState%
% etat initial
\rput(-0.7,0){\State[i']{(-1.9,0)}{B}%
\VSState{(0,0.4)}{B1}\VSState{(0,-0.4)}{B2}}%
% etat final
\rput(0.7,0){\State[t']{(1.9,0)}{C}%
\VSState{(0,0.4)}{C1}\VSState{(0,-0.4)}{C2}}%
}% 
% \StateVar[t'=i'']{(0,0)}{C}%
% transitions 
% \TransSpont
\EdgeL{B}{B1}{}\EdgeL{B}{B2}{}
\EdgeL{C1}{C}{}\EdgeL{C2}{C}{}
% deuxime automate
\rput(3.9,0){%
% Automate central
\FixStateDiameter{2cm}%
\scalebox{\SQRTwoRatio}{\State{(0,0)}{AA}}%
% param etats
\MediumState%
% etat initial
\rput(-0.7,0){\State[i'']{(-1.9,0)}{BB}%
\VSState{(0,0.4)}{BB1}\VSState{(0,-0.4)}{BB2}}%
% etat final
\rput(0.7,0){\State[t'']{(1.9,0)}{CC}%
\VSState{(0,0.4)}{CC1}\VSState{(0,-0.4)}{CC2}}%
}%
\Initial{B}\Final{CC}%
% transitions 
\EdgeL{BB}{BB1}{}\EdgeL{BB}{BB2}{}
\EdgeL{CC1}{CC}{}\EdgeL{CC2}{CC}{}
\TransSpont
\EdgeL{C}{BB}{\unAe}
\end{VCPicture}%
}%
}
\\
\medskipneg
\subfigure[union]{%
\VCDraw[0.9]{%
\begin{VCPicture}{(-6.2,-2.5)(6.2,2.5)}
% premier automate
\rput(0,1.5){%
% Automate central
\FixStateDiameter{2cm}%
\scalebox{\GoldMeanRatio}{\State{(0,0)}{A}}%
% param etats
\MediumState%
% etat initial
\rput(-0.9,0){\State[i']{(-2,0)}{B}%
\VSState{(0,0.4)}{B1}\VSState{(0,-0.4)}{B2}}%
% etat final
\rput(0.9,0){\State[t']{(2,0)}{C}%
\VSState{(0,0.4)}{C1}\VSState{(0,-0.4)}{C2}}%
% transitions 
% \TransSpont
\EdgeL{B}{B1}{}\EdgeL{B}{B2}{}
\EdgeL{C1}{C}{}\EdgeL{C2}{C}{}
}% 
% deuxime automate
\rput(0,-1.5){%
% Automate central
\FixStateDiameter{2cm}%
\scalebox{\GoldMeanRatio}{\State{(0,0)}{AA}}%
% param etats
\MediumState%
% etat initial
\rput(-0.9,0){\State[i'']{(-2,0)}{BB}%
\VSState{(0,0.4)}{BB1}\VSState{(0,-0.4)}{BB2}}%
% etat final
\rput(0.9,0){\State[t'']{(2,0)}{CC}%
\VSState{(0,0.4)}{CC1}\VSState{(0,-0.4)}{CC2}}%
% transitions 
\EdgeL{BB}{BB1}{}\EdgeL{BB}{BB2}{}
\EdgeL{CC1}{CC}{}\EdgeL{CC2}{CC}{}
}%
\State[i]{(-4.8,0)}{I}\State[t]{(4.8,0)}{T}
\Initial{I}\Final{T}%
\TransSpont
\EdgeL{I}{B}{\unAe}\EdgeR{I}{BB}{\unAe}
\EdgeL{C}{T}{\unAe}\EdgeR{CC}{T}{\unAe}
\end{VCPicture}%
}%
}
\ee
\subfigure[star]{%
\VCDraw[0.9]{%
\begin{VCPicture}{(-6.2,-2.5)(6.2,2.5)}
% Automate central
\FixStateDiameter{2cm}%
\scalebox{\GoldMeanRatio}{\State{(0,0)}{A}}%
% param etats
\MediumState%
% etat initial
\rput(-0.9,0){\State[i']{(-2,0)}{B}%
\VSState{(0,0.4)}{B1}\VSState{(0,-0.4)}{B2}}%
% etat final
\rput(0.9,0){\State[t']{(2,0)}{C}%
\VSState{(0,0.4)}{C1}\VSState{(0,-0.4)}{C2}}%
% transitions 
% \TransSpont
\EdgeL{B}{B1}{}\EdgeL{B}{B2}{}
\EdgeL{C1}{C}{}\EdgeL{C2}{C}{}
% nouvel etat initial
\rput(-4.8,0){\State[i]{(0,0)}{BB}}%
\Initial{BB}%
% nouvel etat final
\rput(4.8,0){\State[t]{(0,0)}{CC}}%
\Final{CC}%
% transitions 
\TransSpont
\EdgeL{BB}{B}{\unAe}\EdgeL{C}{CC}{\unAe}
\LArcR[.2]{BB}{CC}{\unAe}
\ChgLArcAngle{45}
\LArcR[.2]{C}{B}{\unAe}
\end{VCPicture}%
}%
}
\medskipneg \medskipneg
\caption{Thompson's construction}
\label{:fig:Tho-1}
\end{figure}

\longonly{%
\figur{Co2-Tho} shows the
construction applied to the
expression~$\Ed_{2}= (a^{*} b + b\xmd b^{*} a )^{*}\msp $.
}

\longonly{%
\begin{figure}[ht]
\centering
\VCDraw{%
\begin{VCPicture}{(-6,-3)(16,3)}
\SmallState
% etats
\VCPut{(0,1)}{%
\State{(0,0)}{A}\State{(2,0)}{B}\State{(4,0)}{C}
\State{(6,0)}{D}\State{(8,0)}{DD}\State{(10,0)}{E}
}%
\VCPut{(0,-1)}{%
\State{(-2,0)}{L}
\State{(0,0)}{FF}\State{(2,0)}{F}\State{(4,0)}{G}
\State{(6,0)}{H}\State{(8,0)}{J}\State{(10,0)}{JJ}
\State{(12,0)}{K}}%
\State{(-3,0)}{I}\State{(13,0)}{T}
\State{(-5,0)}{II}\State{(15,0)}{TT}
\Initial{II}\Final{TT}
% transitions 
\EdgeL{B}{C}{a}\EdgeL{DD}{E}{b}
\EdgeL{L}{FF}{b}\EdgeR{G}{H}{b}\EdgeL{JJ}{K}{a}
\TransSpont
\Edge{II}{I}\Edge{I}{A}\Edge{I}{L}\Edge{A}{B}\Edge{C}{D}
\Edge{E}{T}\Edge{T}{TT}\Edge{F}{G}\Edge{H}{J}\Edge{K}{T}
\LArcL{A}{D}{}\LArcR{F}{J}{}
\LArcL{C}{B}{}\LArcR{H}{G}{}
% transitions oubliŽes
\Edge{D}{DD}\Edge{FF}{F}\Edge{J}{JJ}
\ChgLArcAngle{40}\ChgLArcCurvature{0.65}
\LArcR{II}{TT}{}
\ChgLArcAngle{50}\ChgLArcCurvature{0.75}
\LArcR{T}{I}{}
\RstLArcAngle\RstLArcCurvature 
\end{VCPicture}
}
\caption{The automaton $\msp \Thom{\Ed_{2}}\msp $}
\label{:fig:Co2-Tho}
\end{figure}
}

\longonly{%
\subsubsection{Loop complexity of standard automata}
\label{:sec:sta-hei-sta}

We end this section with the proof of 
\propo{Egg-con}:
\textsl{With every rational expression~$\Ed$,
we can associate an automaton equivalent to~$\Ed$ and whose loop 
complexity is equal to the star height of~$\Ed$.}

The sketches at \figur{sta-aut-ope}~(a) and~(b) make clear that 
if~$\Ac$ and~$\Ac'$ are standard automata,
it holds
\begin{equation}
\jsEnla{(\Ac+\Ac')}  = \jsEnla{(\Ac\cdot\Ac')} 
                     = \max\{\jsEnla{\Ac},\jsEnla{\Ac'}\}  
\eqpnt
\notag
\end{equation}
It follows then from \defin{loo-cpl} that the loop complexity of the 
sum and product of  
standard automata is equal to the star height of the expressions 
for the sum and product, provided equality hold for the operands.

The same relation does not holds for the star of a standard 
automaton, as seen on the example shown at \figur{sta-hei-sta-1}:
the loop complexity of the automaton is not necessarily incremented 
by the star operation.
In the opposite way, the star operation (on standard automata)
may well increase the loop complexity by more than~$1$, as  
shown by~$\Stan{\Ed_{1}}$.
Hence, it is not true that
$\msp \jsEnla{\Stan{\Ed}}=\jsHaut{\Ed}\msp$ holds.

\begin{figure}[ht]
\centering
\setlength{\lga}{2cm}\setlength{\lgb}{1.5cm}
\renewcommand{\ForthBackEdgeOffset}{2.8}
\VCDraw{%
\begin{VCPicture}{(-1,-1)(5,1)}
\SmallState
\State{(0,0)}{A}
\State{(\lga,0)}{B}    
\State{(2\lga,0)}{C}
\Initial[w]{A}\Final[e]{C}
\EdgeL{A}{B}{a}\ArcL{B}{C}{a}\ArcL{C}{B}{a}
\end{VCPicture}%
}%
\eee
\VCDraw{%
\begin{VCPicture}{(-1,-1)(5,1)}
\SmallState
\State{(0,0)}{A}
\State{(\lga,0)}{B}    
\State{(2\lga,0)}{C}
\Initial[w]{A}\Final[s]{A}\Final[e]{C}
\EdgeL{A}{B}{a}\ArcL{B}{C}{a}\ArcL{C}{B}{a}
\end{VCPicture}%
}%

\caption{The standard automaton~$\Ac_{3}$ of~$a\xmd(a^{2})^{*}$ 
and~${\Ac_{3}^{*}}$}
\label{:fig:sta-hei-sta-1}
\end{figure}
}%

% \shortonly{%
% We replace the star operation 
% on standard automata by a more elaborate one.
% A standard automaton is \emph{normalised} if it has only one 
% final state and if this final state is not the origin of any 
% transition.
% \index{normalised standard automaton}%
% \index{standard automaton!normalised --}%
% \index{automaton! normalised standard --}%
% Any standard automaton~$\Ac$ may be transformed into 
% an equivalent normalised one, which we write~$\jsNorm{\Ac}$, 
% and which has the same loop complexity as~$\Ac$.
% }
\longonly{%
In order to circumvent this difficulty, we replace the star operation 
on standard automata by a more elaborate one.
A standard automaton is \emph{normalised} if it has only one 
final state and if this final state is not the origin of any 
transition.
\index{normalised standard automaton}%
\index{standard automaton!normalised --}%
\index{automaton! normalised standard --}%
An obvious construction transforms any standard automaton~$\Ac$ into 
an equivalent normalised one, which we write~$\jsNorm{\Ac}$, 
and we have:

\begin{property}
    \label{:pty:enl-nor}
\ee
$\msp\jsEnla{\jsNorm{\Ac}} = \jsEnla{\Ac}\msp$.
\end{property}
}%

% \shortonly{%
% We further write~$\Zero{\Ac}$ for the standard automaton~$\Ac$ in 
% which the initial state is not final.
% The automaton
% $\msp\left(\Zero{\jsNorm{\Ac}}\matmul\Zero{\Ac}\right)\msp$ has a 
% cut-vertex~$t$.
% Finally, let 
% $\msp\Bc = \left(\Zero{\jsNorm{\Ac}}\matmul\Zero{\Ac}\right)^{*}\msp$ 
% and~$\Bc'$ the automaton in which~$t$ has been made final;
% this~$\Bc'$ answers the problem:
% $\msp\CompAuto{\Bc}=\CompAuto{\Ac}^{*}\msp$
% and
% \NeMLV%
% $\msp \jsEnla{\Bc'} = \jsEnla{\Ac } + 1\msp$.
% }%

\longonly{%
We further write~$\Zero{\Ac}$ for the standard automaton~$\Ac$ in 
which the initial state is not final.
The automaton
$\msp\left(\Zero{\jsNorm{\Ac}}\matmul\Zero{\Ac}\right)\msp$ has a 
cut-vertex~$t$.
Finally, let 
$\msp\Bc = \left(\Zero{\jsNorm{\Ac}}\matmul\Zero{\Ac}\right)^{*}\msp$ 
and~$\Bc'$ the automaton in which~$t$ has been made final.
Clearly,
$\msp\CompAuto{\Bc}=\Ac^{*}\msp$ and 
\propo{Egg-con} is established with the proof of the following 
lemma.

\begin{lemma}
    \label{:lem:egg-rec}
$\jsEnla{\Bc'} = \jsEnla{\Ac } + 1$ .
\end{lemma}

\begin{proof}
The automaton~$\Bc'$ without its initial state~$i$ is a ball; by definition, we have
\begin{equation}
    \jsEnla{\Bc'} = 
       1 + \min \Defi{\jsEnla{\Bc' \bk \{i,s\}}}{s \in \Bc'\bk i}
\eqpnt
\label{:equ:egg-rec}
\end{equation}
The final state~$t$ of~$\jsNorm{\Ac}$ is a \emph{cut vertex} (of the 
underlying graph) of~$\Zero{\jsNorm{\Ac}}\cdot\Zero{\Ac}$
If we set~$s=t$ in~\equnm{egg-rec}, the balls 
of~$\Bc''=\Bc'\bk\{i,s\}$ are those of~$\Zero{\jsNorm{\Ac}}$
and~$\Zero{\Ac}$ and hence 
$\msp\jsEnla{\Bc'}\jsleq\jsEnla{\Ac}+1\msp$.
If we take an arbitrary~$s$ in~$\Zero{\jsNorm{\Ac}}$ or~$\Zero{\Ac}$, 
$\Bc''$ contains either~$\One{\Ac}$ or~$\Zero{\jsNorm{\Ac}}$ and
$\msp\jsEnla{\Bc'}\jsgeq\jsEnla{\Ac}+1\msp$.
\end{proof}
}

\longonly{%
\begin{figure}[ht]
\centering
\setlength{\lga}{2cm}\setlength{\lgb}{1cm}
\renewcommand{\ForthBackEdgeOffset}{2.8}
\VCDraw{%
\begin{VCPicture}{(-1,-1.5)(7,1.5)}
\SmallState
\State{(\lga,0)}{B}    
\State{(2\lga,-\lgb)}{C}
\MediumState
\State[i]{(0,0)}{A}
\State[t]{(3\lga,0)}{D}
\Initial[w]{A}\Final[e]{D}
\EdgeL{A}{B}{a}\ArcL{B}{C}{a}\ArcL{C}{B}{a}\LArcL{B}{D}{a}
\end{VCPicture}%
}%
\eee
\VCDraw{%
\begin{VCPicture}{(-1,-1.5)(13,1.5)}
\MediumState
\State[i]{(0,0)}{A}
\State[t]{(3\lga,0)}{AA}
\SmallState
\State{(\lga,0)}{B}    
\State{(2\lga,-\lgb)}{C}
\VCPut{(3\lga,0)}{%
\State{(\lga,0)}{BB}    
\State{(2\lga,-\lgb)}{CC}
\State{(3\lga,0)}{DD}
}
\Initial[w]{A}\Final[s]{A}%\Initial[w]{AA}
\Final[s]{AA}\Final[e]{DD}
\EdgeL{A}{B}{a}\ArcL{B}{C}{a}\ArcL{C}{B}{a}\ArcL{B}{AA}{a}
\EdgeL{AA}{BB}{a}\ArcL{BB}{CC}{a}\ArcL{CC}{BB}{a}
\ArcL{BB}{DD}{a}\LArcR{DD}{B}{a}
\end{VCPicture}%
}%

\caption{The standard automata~$\jsNorm{{\Ac_{3}}}$ and~$\Bc'_{3}$ 
buit from  
$\msp\left(\Zero{\jsNorm{{\Ac_{3}}}}\cdot
           \Zero{\Ac_{3}}\right)^{*}\msp$}
\label{:fig:sta-hei-sta-2}%
\end{figure}
}%

\ifcheat\smallskipneg\fi%
\ifcheat\smallskipneg\fi%
\ifcheat\smallskipneg\fi%
\subsection{The derived-term automaton of an expression}
\label{:sec:der-aut-exp}%

Let us first recall the \emph{(left) quotient} operation on 
languages: 
\begin{equation}
     \fa L\in\jsPart{\Ae}\quantvrg\fa u\in\Ae\quantsp
     u^{-1}L = \Defi{v\in\Ae}{u\xmd v\in L}
     \eqpnt
     \eee
     \notag
\end{equation}
The quotient is a \emph{(right) action} of~$\Ae$ on~$\jsPart{\Ae}$:
\shortlong{%
$\msp(u\xmd v)^{-1}L = v^{-1}\left(u^{-1}L\right)\msp$.
}{%
\begin{equation}
    \fa L\in\jsPart{\Ae}\quantvrg\fa u,v\in\Ae\quantsp
    (u\xmd v)^{-1}L = v^{-1}\left(u^{-1}L\right)
    \eqpnt
    \eee
    \label{:equ:quo-act}
\end{equation}
}
A fundamental, and characteristic, property of rational 
\index{quotient!of a language}%
\emph{languages} --- which is  
another way to express that they are recognisable --- is that they 
have a \emph{finite number of quotients}.

The principle of the construction we present in this section, and 
which we call \emph{derivation}, is to 
transfer the quotient on languages to an operation on the expressions.
First introduced by Brzozowski \cite{Brzo64}, the definition of the
derivation of an expression~$\Ed$ has been
% slightly, but smartly, 
modified by Antimirov \cite{Anti96} (\cf Notes) 
and yields a non-deterministic 
automaton $\DerT{\Ed}$, which  
we propose to call the \emph{derived-term automaton} of~$\Ed$.
% By essence, this construction concerns expressions over \emph{free 
This construction concerns thus expressions over \emph{free 
\index{free monoid}%
monoids} only. 
In the sequel, $\Ed$ is a rational expression over~$\Ae$.

\begin{definition}[Brzozowski--Antimirov \cite{Anti96}]
    \label{:def:exp-der-b}%
The \emph{derivation}
\index{derivation (of an expression)}%
\index{expression!derivation of --|see{derivation}}%
of~$\Ed$ with respect to a letter~$a$ of~$A$,
denoted by~$\ExpDer{\Ed}$, is \emph{a set} of rational expressions
over~$\Ae$, inductively defined by:
{\iesf
{\allowdisplaybreaks
\begin{align}
\ExpDer{\zed}  \msp=\msp \ExpDer{\und} \msp=\msp \emptyset ,
\qquad \forall b \in A  \quantsmsp
& \ExpDer{b}  \msp=\msp \left \{
\begin{array}{cl}
\{1\} & \quad \text{if \quad}  b = a \msp, \\
\emptyset & \quad \text{otherwise},
\end{array} \right . \eqvrg \ee
\label{:equ:der-b-1}
%      \notag 
	 \\[1ex]
     \ExpDer{(\Fd \autplus \Gd)} & \msp=\msp
        \ExpDer{\Fd} \cupsp \ExpDer{\Gd}
   \eqvrg\label{:equ:der-b-2}\\[1ex]
   \ExpDer{(\Fd \autprod  \Gd)} & \msp=\msp
      \left(\ExpDer{\Fd}\right) \autprod    \Gd
         \cupsp    \TermCst{\Fd} \, \ExpDer{\Gd}
    \eqvrg\label{:equ:der-b-3}\\[1ex]
    \ExpDer{(\Fd^{*})} & \msp=\msp
       \left(\ExpDer{\Fd}\right)\autprod \Fd^{*}
          \eqpnt\label{:equ:der-b-4}
  \end{align}}%
}%
\end{definition}

\equat{der-b-3} should be understood with the convention that
the product~$x\xmd X$ of a set~$X$ by a Boolean value~$x$ is~$X$  
if $x=\und$ and~$\emptyset$ if $x=\zed$.
The induction involved in
Equations~\equnmnosp{der-b-2}--\equnmnosp{der-b-4}
should be interpreted by extending derivation additively (as are 
always derivation operators) and  by distributing (on the right) 
the~$\autprod$ operator over sets as well.
Finally, every operation on rational expressions is computed modulo 
the trivial identities~$\IdRT$, but not modulo the natural 
identies~$\IdRN$ 
--- nor the idempotent identites~$\IdRI$ and~$\IdRJ$.

\begin{definition}
\label{:def:exp-der-wrd}%
The \emph{derivation} of~$\Ed$ with respect to a non-empty word~$v$ 
of~$\Ae$,  
denoted by~$\msp\longonly{\displaystyle}{\ExpDer[v]{\Ed}}\msp$, 
is \emph{the set of rational expressions over~$\Ae$}, 
defined by~\equnm{der-b-2}--\equnm{der-b-4} 
for letters in~$A$ 
and by induction on the length of~$v$ by:
{\iesf
\begin{equation}
\forall u \in A^+ \quantvrg \forall a \in A \quantsp
\ExpDer[ua]{\Ed} = \ExpDer{ \left(\ExpDer[u]{\Ed}\right)}
\eqpnt 
\eee\eee
\label{:equ:der-b-5}
\end{equation}}
\end{definition}

\shortlong{%
The derivation of expressions is parallel 
% corresponds 
to the quotient of languages and we have:
{\iesf
\begin{equation}
\forall \Ed \in \RatEA \quantvrg \forall u \in \Ap \quantsp
\CompExpr{\ExpDer[u]{\Ed}} \xmd = u^{-1}\CompExpr{\Ed}
\eqpnt
\eee \eee \eee \e 
\label{:equ:der-b-quo-2}
\end{equation}}
}{%
The derivation of expressions is indeed parallel to the quotient of 
languages as we have the following property.

\begin{property}
\label{:pty:quo-lan}
\begin{align}
\forall L, K\subseteq \Ae \quantvrg \forall a \in A \quantsp
a^{-1}(L\cup K) & = a^{-1}L \cupsp a^{-1}K
\eqvrg
\eee\eee\ee
\notag
\\
a^{-1}(L\xmd K) & = (a^{-1}L) \cupsp \TermCst{L}\xmd a^{-1}K
\eqvrg
\notag
\\
a^{-1}(L^{*}) & = (a^{-1}L)\xmd L^{*}
\eqpnt
\notag
\end{align}
\end{property}

It follows, by induction on the depth of the expression:
\begin{equation}
\forall \Ed \in \RatEA \quantvrg \forall a \in A \quantsp
\CompExpr{\ExpDer{\Ed}} = a^{-1}\CompExpr{\Ed}
\eee\eee
\label{:equ:der-b-quo}
\end{equation}
which in turn implies, by induction on the length of words and the 
parallel between~\equnm{der-b-5} and~\equnm{quo-act}:
\begin{equation}
\forall \Ed \in \RatEA \quantvrg \forall u \in \Ap \quantsp
\CompExpr{\ExpDer[u]{\Ed}} = u^{-1}\CompExpr{\Ed}
\eqpnt
\eee\eee
\label{:equ:der-b-quo-2}
\end{equation}
In particular, we have:
\begin{property}
\label{:pty:cst-trm-der}
% If
% $\msp\Ed\not=\zed\msp$,
% then
% $\msp\TermCst{\Ed}=\und\msp$
% or 
If~$\Ed$ is not a constant,
there exists~$u$ in~$\Ap$ such that
$\msp\TermCst{\ExpDer[u]{\Ed}}=\und\msp$.
\end{property}
}%

\ifcheat\smallskipneg\fi%
\ifcheat\smallskipneg\fi%
\begin{example}
    \label{:exa:der-aut-1}%
The derivation of~$\msp\Ed_{1}=(a^{*}b+b\xmd b^{*}a)^{*}\msp$ (\cf 
\exemp{sta-aut-1}) yields:
{\iesf
\begin{gather}
\ExpDer[a]{\Ed_{1}}=\ExpDer[aa]{\Ed_{1}}=\{a^{*}b\xmd\Ed_{1}\}
\EqVrgInt
\ExpDer[b]{(\Ed_{1})^{*}}=\{\Ed_{1},b^{*}a\xmd\Ed_{1}\}
\EqVrgInt
\notag	\\
\ExpDer[b]{a^{*}b\xmd\Ed_{1}}=\{\Ed_{1}\}
\EqVrgInt
\ExpDer[a]{(b^{*}a\xmd\Ed_{1})^{*}}=\{\Ed_{1}\}
\EqVrgInt
\ExpDer[b]{(b^{*}a\xmd\Ed_{1})^{*}}=\{b^{*}a\xmd\Ed_{1}\}
\EqPnt
\notag	
\end{gather}}
\end{example}

\subsubsection{The derived-term automaton}
\label{:ssc:der-ter-aut}

Derivation thus associates a pair of an expression and a word with a 
set of expressions. 
We now turn this map into an automaton.

\begin{definition}
\label{:def:der-ter}%
We call \emph{true derived term} of~$\Ed$ every expression that 
belongs to ${\longonly{\displaystyle}{\ExpDer[w]{\Ed}}}$
for some word~$w$ of~$\Ap$;
\index{derived term (of an expression)!true --}%
we write  
$\msp \displaystyle{\TruDerTer{\Ed}}\msp$ 
for the set of true derived terms of~$\Ed$:
\begin{equation}
   \textstyle{\TruDerTer{\Ed} = \bigcup_{w\in A^{+}}\ExpDer[w]{\Ed}} \eqpnt
    \label{:equ:tdt-def}
\end{equation}
The set
$\msp\DerTer{\Ed} = \TruDerTer{\Ed} \cupsp \{\Ed \}\msp$
\index{derived term (of an expression)}%
is the set of \emph{derived terms} of~$\Ed$.
\end{definition}

\begin{example}[\exemp{der-aut-1} cont.]
    \label{:exa:der-aut-2}%
$\msp
\DerTer{\Ed_{1}}= \{\Ed_{1}, a^{*}b\xmd\Ed_{1}, b^{*}a\xmd\Ed_{1}\}
\msp$.
\end{example}

The sets of derived terms and the rational operations are related by 
the following equations, from which most of the subsequent properties 
will be derived. 

\begin{proposition}%[\cite{LombSaka05a}]rational \msp
\label{:pro:tru-der-ind}
Let~$\Fd$ and~$\Gd$ be two expressions.
Then,
$\TruDerTer{\Fd \autplus \Gd} 
   = \TruDerTer{\Fd} \cupssp  \TruDerTer{\Gd}$,
$\TruDerTer{\Fd \autprod  \Gd} 
    =\left(\TruDerTer{\Fd}\right)\autprod\Gd
         \cupssp   \TruDerTer{\Gd}$, and 
$\msp\TruDerTer{\Fd^{*}} 
     = \left(\TruDerTer{\Fd}\right)\autprod \Fd^{*}\msp$
hold.
\end{proposition}

Starting from
$\msp\TruDerTer{\zed}=\TruDerTer{\und}=\emptyset\msp$
and $\msp\TruDerTer{a}=\{\und\}\msp$ for every~$a$ in~$A$,
$\TruDerTer{\Ed}$ \emph{can be computed} from \propo{tru-der-ind}
by induction on~$\Dpth{\Ed}$ and \emph{without reference to 
the derivation operation} (\cf the \emph{prebases} in~\cite{Mirk66} 
and \defin{wei-der-ter} below).
It follows in particular that
$\msp\jsCard{\TruDerTer{\Ed}}\leq\LiteLgth{\Ed}\msp$
and thus:

\begin{corollary}
\label{:cor:der-ter-car}
\ee 
$\msp\jsCard{\DerTer{\Ed}}\leq\LiteLgth{\Ed}+1\msp$.
\end{corollary}

The computation of~$\DerTer{\Ed}$ is a $\EtAs$-map, as expressed 
by the following.

\begin{definition}[Antimirov \cite{Anti96}]
\label{:def:aut-der}% a rational expression over~$\Ae$ finite 
The \emph{derived-term automaton} of~$\Ed$ 
\index{derived-term|see{automaton (of...)}}%
\index{automaton (of an expression)!derived-term --}%
is the automaton~$\DTAut{\Ed}$
whose set of states is $\DerTer{\Ed}$ and whose transitions are
defined by:
\begin{conditionsiii}
    \item  if~$\Kd$ and~$\Kd'$ are derived terms of~$\Ed$ 
	and~$\xmd a\xmd$ a letter of~$A$, 
	then $\msp(\Kd,a,\Kd')\msp$ is a transition if and only if
            $\Kd'$ belongs to $\ExpDer{\Kd}$;
            
    \item  the initial state is~$\Ed$;

    \item a derived term $\Kd$ is final if and only if 
$\TermCst{\Kd}=1$.  
\end{conditionsiii}
\end{definition}

\begin{theorem}[\cite{Anti96}]
    \label{:the:der-ter-aut}
For any rational expression~$\Ed$,
$\msp\CompExpr{\Ed} = \CompAuto{\DTAut{\Ed}} \msp$.
\end{theorem}

\begin{example}[\exemp{der-aut-2} cont.]
    \label{:exa:der-aut-3}%
The automaton~$\DTAut{\Ed_{1}}$ is shown at \figur{der-aut-3}.
\end{example}

\begin{figure}[ht]
\setlength{\lga}{4cm}\setlength{\lgb}{1.5cm}%
	\SmallPicture%
\centering
\VCDraw{%
\begin{VCPicture}{(-5,-1.4)(5.2,1.2)}
% etats
\LargeState
\StateVar[\xmd a^{*} b\xmd\Ed_{2}]{(-5,0)}{A}
\State[\Ed_{2}]{(0,0)}{B}
\StateVar[\xmd b^{*} a\xmd\Ed_{2}]{(5,0)}{C}
\Initial[nw]{B}\Final[ne]{B}
% transitions
\ArcR{A}{B}{b}\ArcR{B}{A}{a}
\ArcL{B}{C}{b}\ArcL{C}{B}{a}
\LoopS{B}{b}
\VarLoopOn
\LoopS{A}{a}\LoopS{C}{b}
\VarLoopOff
\end{VCPicture}
}%
\caption{The automaton~$\DTAut{\Ed_{1}}$. }
\medskipneg
\label{:fig:der-aut-3}
\end{figure}

\subsubsection{Relationship with the standard automaton}
\label{:sec:sta-der-aut}

The constructions of the standard and derived-term automata of an 
expression are of different nature.
But both arise from the same inner structure of the expression by two 
inductive processes, and the two automata have 
a structural likeness  which yields another proof of 
\corol{der-ter-car}:

\begin{theorem}[\cite{ChamZiad02}]
    \label{:the:der-ter-sta}
For any rational expression~$\Ed$, 
$\DTAut{\Ed}$ is a quotient of $\Stan{\Ed}$.
\end{theorem}

\subsubsection{Derivation and bracketing}
\label{:sec:der-bra}

The derivation operator is 
sensitive to the bracketing of expressions;
on the other hand, it does  
% commute to the associativity identity~\equnm{ass-ide}. 
commute to the associativity identity~$\IdRAs$. 

\begin{example}
    \label{:exa:der-ter-bra}
Let
$\msp a \xmd b \xmd ( c \xmd (a\xmd b) )^*\msp$  
be an expression which is not completely bracketed.
The derivation
of the two expressions obtained by different bracketings yields:
\begin{align*}
\DerTer{\strut a \xmd ( b \xmd ( c \xmd (a\xmd b) )^* )} &= 
\{a \xmd ( b \xmd ( c \xmd (a\xmd b) )^* )\setvrg b \xmd ( c \xmd (a\xmd b) )^*,
( c \xmd (a\xmd b) )^*\setvrg (a\xmd b) \xmd ( c \xmd (a\xmd b) )^* \} \eqpnt \\ 
\DerTer{\strut(a\xmd b) \xmd ( c \xmd (a\xmd b) )^* } &= 
\{(a\xmd b) \xmd ( c \xmd (a\xmd b) )^* \setvrg b \xmd ( c \xmd (a\xmd b) )^*,
( c \xmd (a\xmd b) )^*  \} \eqpnt
\end{align*}
\end{example}

More precisely, we have the following.

\begin{proposition}[\cite{AngrEtAl10}] 
    \label{:pro:der-ter-bra}
Let $\Ed$, $\Fd$ and $\Gd$ be three rational expressions. 
Then:
{\small
\begin{equation}
\jsCard{\strut \DerTer{\strut(\Ed \cdot \Fd) \cdot \Gd} } \jsleq 
\jsCard{\strut \DerTer{\strut\Ed \cdot ( \Fd \cdot \Gd )} }
% \notag
% \\
\ \text{ and }\ 
        \IdRAs\dedtxt \DerTer{\strut(\Ed \cdot \Fd) \cdot \Gd}
          \equiv   \DerTer{\strut\Ed \cdot ( \Fd \cdot \Gd )}
\msp.
\notag
% \label{:equ:der-ter-bra} 
\end{equation}}
\end{proposition}

%%%%%%  AE Section 5  %%%%%%%%%%
%%%             100921                   %%%
%%%%%%%%%%%%%%%%%%%%%%%%%%%%%%%%%%%%%%%%%%%%
\section{Changing the monoid}
\label{:sec:cha-mon}

Most of what has been presented so far extends without 
problems from languages to subsets of arbitrary monoids, from 
expressions over a free monoid to expressions over such monoids. 
We run over definitions and statements to transform them accordingly. 
The main difference will be that rational and recognisable sets do 
not coincide anymore, making the link between finite 
automata and rational expressions even tighter, and ruling out 
quotient and derivation that refer to 
the recognisable `side' of rational languages.

Non-free monoids of interest in the field of computer 
science and automata theory are, among others,
direct products of free monoids (for relations between words),
free commutative monoids (for counting purpose),
partially commutative, or trace, monoids (for modelling 
concurrent or parallel computations),
free groups and polycyclic monoids (in relation with pushdown 
automata).

In the sequel, $M$ is a monoid, and~$\unM$ its identity element.

\ifcheat\smallskipneg\fi%
\ifcheat\smallskipneg\fi%
\subsection{Rationality}
\label{:sec:rat-mon}%

\paragraph{Rational sets and expressions}
\label{:sec:rat-exp-mon}%

Product and star are defined in~$\jsPart{M}$ as in~$\jsPart{\Ae}$ and 
the set of \emph{rational subsets} of~$M$, denoted by~$\RatM$,
\index{rational!subset}%
is the smallest subset 
of~$\jsPart{M}$ which contains the finite sets 
(including the empty set)
and which is closed under union, product, and star. 

\emph{Rational expressions over~$M$} are defined as those 
over~$\Ae$, with the only difference that the \emph{atoms} are the 
elements of~$M$; their set is denoted by~$\RatEM$.
We also write~$\CompExpr{\Ed}$ for the subset \emph{denoted} by an 
expression~$\Ed$. 
Two expressions are equivalent if they denote the same subset and we 
have the same statement as \propo{rat-exp-fm}: 

\ifcheat\smallskipneg\fi%
\begin{proposition}
    \label{:pro:rat-exp-mon}%
A subset of~$M$ is rational if and only if it is denoted 
by a  rational expression over~$M$.
\end{proposition}

\ifcheat\smallskipneg\fi%
A subset~$G$ of~$M$ is a \emph{generating set} if
$\msp M=G^{*}\msp$. 
The direct part of \propo{rat-exp-mon} may be restated with more 
precision as: 
\emph{any rational subset of~$M$ is denoted by a rational expression 
whose atoms are taken in any generating set}.
It follows from the converse part that a rational subset of~$M$ is 
\index{generating set|see{monoid}}%
\index{monoid!generating set of --}%
\index{monoid!finitely generated --}%
contained in a finitely generated submonoid. 

\ifcheat\smallskipneg\fi%
\paragraph{Finite automata}

An \emph{automaton over~$M$}, 
\index{automaton}%
denoted by $\msp \Ac= \aut{Q,M,E,I,T}\msp$,
is defined like an automaton over~$\Ae$, with the only difference that
the transitions are labelled by elements of~$M$:
$\msp E\subseteq Q\x M\x Q\msp$.
Then,~$\Ac$ is \emph{finite} if~$E$ is finite.

The \emph{subset accepted} by~$\Ac$, called the 
\index{automaton!subset accepted by --}%
\emph{behaviour} of~$\Ac$ and denoted by~$\CompAuto{\Ac}$ as above, 
\index{automaton!behaviour of --}%
is the set of labels of \emph{successful computations}:
$\msp \CompAuto{\Ac}= 
\Defi{m\in M}{\ext i\in I,\ext t\in T\quantsmsp i\pathaut{m}{\Ac}t}$.

\ifcheat\smallskipneg\fi%
\paragraph{The fundamental theorem of finite automata}

In this setting, the statement appears more clearly different from 
Kleene's theorem.
Its first appearance\footnote{%
   Hidden in a footnote!}
seems to be in Elgot and Mezei's 
\index{Fundamental theorem!of finite automata}%
\index{Kleene's theorem}%
paper on rational relations.

\begin{theorem}[\cite{ElgoMeze65}]
    \label{:the:fun-the-mon}%
A subset of a monoid~$M$ is rational if and only if it is the 
behaviour of a finite automaton over~$M$ whose labels are taken in 
any generating set of~$M$.
\end{theorem}

There is not much to change in Propositions~\ref{:pro:phi-map-fm} 
and~\ref{:pro:psi-map-fm} 
% to \propo{phi-map-fm} and \propo{} nor to their proofs 
to establish \theor{fun-the-mon}.

\begin{proposition}[$\AtEs$-maps]
    \label{:pro:phi-map-mon}%~$\Ed$
For every finite automaton~$\Ac$ over~$M$, there exist rational 
expressions over~$M$ which denote~$\CompAuto{\Ac}$\xmd.
\end{proposition}

\ifcheat\smallskipneg\fi%
All four methods described in \secti{aut-exp} apply for 
arbitrary~$M$, even if their formal proof may be slightly different 
(Arden's lemma does not hold anymore).

\begin{proposition}[$\EtAs$-maps]
    \label{:pro:psi-map-mon}%~$\Ac$
For every rational expression~$\Ed$ over~$M$, there exist 
finite automata over~$M$ whose behaviour is equal to~$\CompExpr{\Ed}$\xmd.
\end{proposition}

\ifcheat\smallskipneg\fi%
Here again, the algorithms and results described in 
\secti{sta-aut-exp} for the construction of the standard automaton, 
\index{automaton (of an expression)!standard --}%
\index{automaton (of an expression)!Thompson --}%
Thompson automaton, \etc pass over to expressions 
over~$M$.
On the contrary, quotients in~$M$ define recognisable subsets of~$M$ 
and not rational ones (see below) and derivation of expressions 
\index{derivation (of an expression)}%
over~$M$ does not make sense anymore.

\ifcheat\smallskipneg\fi%
\ifcheat\smallskipneg\fi%
\subsection{Recognisability}

\defin{rec-lan} may be rephrased \textit{verbatim} for 
arbitrary monoids.
\emph{A subset~$P$ of~$M$ is said to be \emph{recognised} by a morphism
$\msp\alpha\colon M\rightarrow N\msp$
if
$\msp P= \alpha^{-1}(\alpha(P))\msp$.
A subset of~$M$ is \emph{recognisable} if it is recognised by a 
\index{recognisable!subset}%
morphism from~$M$ into a \emph{finite} monoid.
The set of \emph{recognisable subsets} of~$M$ is
denoted by~$\RecM$.}

\ifcheat\smallskipneg\fi%
\paragraph{Recognisable and rational subsets}

We can then
reproduce almost \textit{verbatim} the converse part of the proof of 
\theor{kle-fm}.
% ~$M$ be a finitely generated monoid and
Let~$P$ be in~$\Rec M$, 
recognised by a morphism~$\alpha$.
We replace the alphabet~$A$ by any generating set~$G$ of~$M$ in the 
construction of the automaton~$\Ac_{\alpha}$. 
If~$M$ is finitely generated, $G$ is finite, so is~$\Ac_{\alpha}$ 
and~$P$ is rational by \theor{fun-the-mon}:

\begin{proposition}[McKnight \cite{MKni64}]
    \label{:pro:rec-fin-gen}%
If~$M$ is finitely generated, then $\msp\RecM\subseteq\RatM\msp$.
\end{proposition}

\ifcheat\smallskipneg\fi%
On the other hand, the first part of the quoted 
proof does not generalise to non-free monoids and the inclusion in 
\propo{rec-fin-gen} is strict in general.
For instance, the set
$\msp(a,c)^{*}=\big((a,1)(1,c)\big)^{*}\msp$
is a rational subset of~$a^{*}\x c^{*}$
(where the product is formed component wise).
It is accepted by a two-state automaton which 
induces a map~$\mu$ from the generating set of~$a^{*}\x c^{*}$ 
into~$\B^{2\x2}$:
\begin{equation}
    \mu\big((a,1)\big) = \redmatu{\matricedd{0}{1}{0}{0}}
    \ee\text{and}\ee
    \mu\big((1,c)\big) = \redmatu{\matricedd{0}{0}{1}{0}}
\eqpnt
\notag
\end{equation}
But this map does not define a \emph{morphism} from~$a^{*}\x c^{*}$ 
into~$\B^{2\x2}$.

\ifcheat\smallskipneg\fi%
\paragraph{Decision problems for rational sets}

In general, $\RatM$ is not a Boolean algebra.
This is also accompanied with 
undecidability results.
The undecidability of Post Correspondence Problem, easily expressed 
in terms of monoid morphisms, implies for instance:

\ifcheat\smallskipneg\fi%
\begin{theorem}[Rabin--Scott \cite{RabiScot59}]
    \label{:the:int-ind}%
It is undecidable whether the intersection of two rational sets of 
$\{a,b\}^{*}\x\{c,d\}^{*}$ is empty or not.
\end{theorem}

\ifcheat\smallskipneg\fi%
From which one deduce:

\ifcheat\smallskipneg\fi%
\begin{theorem}[Fischer--Rosenberg \cite{FiscRose68}]
    \label{:the:equ-ind}%
The equivalence of finite automata, and hence of rational 
expressions, over $\{a,b\}^{*}\x\{c,d\}^{*}$ is undecidable.
\end{theorem}

In contrast, the cases where~$\RatM$ is an effective Boolean algebra 
--- such as when~$M$ is a (finitely generated) \emph{free commutative 
monoid} \cite{GinsSpan66} or \emph{free group} \cite{Flie71} --- play 
a key role in model-checking issues which involve counters, or 
pushdown automata.

%%%%%%  AE Section 6  %%%%%%%%%%
%%%             111112                   %%%
%%%%%%%%%%%%%%%%%%%%%%%%%%%%%%%%%%%%%%%%%%%%
\ifcheat\smallskipneg\fi%
\ifcheat\smallskipneg\fi%
\ifcheat\smallskipneg\fi%
\ifcheat\smallskipneg\fi%
\section{Introducing weights} %multiplicity
\label{:sec:int-mul}

Most of the statements about automata and expressions established in
the previous sections extend again without much difficulties in the
\emph{weighted case}, as we have taken care to formulate them
adequately.
There are two questions though that should be settled first in order
to set up the framework of this generalisation. 
First, the definition of the \emph{star operator} requires some
mathematical apparatus to be meaningful.  Second, the definition of
\emph{weighted expressions} has to be tuned in such a way that former
computations such as the \emph{derivation} remain valid.\footnote{%
   The definition of the \emph{behaviour of weighted automata}
   also conceals a problem due to the existence of \emph{spontaneous} or
   $\epsilon$-transitions).
   This is out of the scope of this chapter
   where we focus on the relationships between automata and 
   expressions. 
   All usual definitions eventually allow to establish that every
   automaton whose behaviour is defined is equivalent to a \emph{proper}
   automaton.
   This is how we define a weighted automaton and where we
   begin our presentation.
   We thus save a significant amount of foundation results.  
   On this subject,
%    and the various ways to deal with it, 
   we refer to 
   other chapters of 
   this handbook (Chapters~\ref{DK:chp:DK} and~\ref{ZE:chp:ZE}) and
   other works
   (\cite{SaloSoit77,BersReut84,KuicSalo86,Saka03,DrosEtAl09Edhb,LombSaka13}).}

\ifcheat\smallskipneg\fi%
\ifcheat\smallskipneg\fi%
\subsection{Weighted languages, automata, and expressions}
\label{:sec:wei-lae}%

\subsubsection{The series semiring}
\label{:sec:rat-ser}%

The \emph{weights}, with which we enrich the
languages or subsets of monoids are taken in a \emph{semiring}, so as
\index{semiring}%
to give the set of \emph{series} we build the desired structure.
We are interested in weights as they actually appear in the
modelisation of phenomena that we want to be able to describe (and not
because they fullfil some axioms).
These are the classical numerical semirings~$\N$,
$\Z$,
$\Q$,
\etc,
the less classical~$\StruSA{\Z\cup\etainfty,\min,+}$,
\etc 
None of them are \emph{Conway semirings} (\cf \CTchp{ZE}), 
$\N$
is a \emph{quasi-Conway semiring}
but not the others.
In the  sequel, $\K$ is a semiring.
The unweighted case corresponds to~$\K=\B$ and will 
\index{Conway|see{semiring}}%
\index{semiring!Conway --}%
\index{semiring!quasi-Conway --}%
be refer to as \emph{the Boolean case}.

As in the Boolean case, free monoids give rise to results 
which do not hold in non-free ones
(the Kleene--Sch\"utzenberger theorem).
But not all non-free monoids  allow to easily define series with 
weights in arbitrary semirings.
We restrict ourselves to \emph{graded monoids}, \ie which are 
equipped with a \emph{length function}.
They behave exactly like the free monoids as far as the construction 
of series is concerned, they cover many monoids that are considered in computer 
science, and they are sufficient to make clear the difference between 
the free and non-free cases as far as rationality is concerned.
\index{monoid!graded --}%
\index{monoid!finitely generated --}%
In the  sequel, $M$ is a finitely generated graded monoid.

\ifcheat\smallskipneg\fi%
\paragraph{Series}

Any \emph{map}~$s$ from~$M$ to~$\K$ is a \emph{formal power series} 
(\emph{series} for short)
over~$M$ with coefficients in~$\K$.
The image by~$s$ of an element~$m$ in~$M$
is written $\msp\bra{s,m}\msp$
and is called the \emph{coefficient of~$m$ in~$s$}.
\index{series}%
\index{power series|see{series}}%
\index{series!coefficient in a --}%
\index{coefficient|see{series}}%
The set of these series, written $\KM$, is equipped with 
the (left and right) \emph{`exterior' multiplications},
\index{multiplication (exterior)}%
the pointwise \emph{addition}, 
and the \emph{(Cauchy) product}: for every~$m$ in~$M$,
$\msp
    \bra{ s \xmd t ,m }  =    
 \sum _{{uv=m}}\bra{ s ,u }\bra{ t ,v }
\msp$.
As~$M$ is graded, the  product is 
well-defined, and the three operations make~$\KM$ a semiring
(\cf \CTchp{DK}).

The \emph{support} of a series~$s$ is the subset of elements 
of~$M$ whose coefficient in~$s$ is not~$\zeK$.
\index{support|see{series}}%
\index{series!support of --}%
\index{polynomial}%
A series with finite support is a \emph{polynomial}; the set of 
polynomials over~$M$ with coefficients in~$\K$ is written~$\KPM$.

\ifcheat\smallskipneg\fi%
\paragraph{Topology}

The following definition of the star as
an \emph{infinite sum} calls for the definition of a \emph{topology} 
on~$\KM$.  
The semirings~$\K$ we consider are equipped with a topology defined by a 
\emph{distance}, whether it is a 
\emph{discrete topology} ($\N$, $\Z$, 
$\StruSA{\Z\cup\etainfty,\min,+}$, \etc) or  
a more classical one ($\Q$, $\R$, another $\SerSAnMon{\Lb}{N}$, \etc). 
Since~$M$ is graded (and finitely generated) it is easy to derive 
% from the distance on~$\K$ 
a distance which defines on~$\KM$
the \emph{simple convergence topology}: 
\begin{equation}
    \text{$s_{n}\msp$ \emph{converges to} $\msp s\msp$
    if, and only if, 
    for all~$m$ in~$M$,
    $\msp\bra{s_{n},m}\msp$ \emph{converges to} $\msp\bra{s,m}\msp$.} 
    \notag 
\end{equation}
Along the same line, a family of series~$\{s_i\}_{i \in I}$ is 
\emph{summable} if for  
every~$m$ in~$M$ the family $\{\bra{s_i,m}\}_{i \in I}$ is summable (in~$\K$). 
An obvious case of summability is when for every~$m$ in~$M$ 
there is only a finite number of indices~$i$ such
that~$\bra{s_i,m}$ is different from~$\zeK$, in which case the 
family~$\{s_i\}_{i \in I}$ is said to be \emph{locally finite}.

All quoted semirings that we consider are \emph{topological 
semirings}, that is, not only equipped with a topology, but their 
semiring operations are \emph{continuous}.
We also use silently in the sequel the following identification:
if~$Q$ is a finite set, $\KMQQ$, the semiring of $Q\x Q$-matrices with 
entries in~$\KM$ \emph{is isomorphic} to~$\KQQM$, the semiring of 
series on~$M$ with coefficients in~$\KQQ$.

\ifcheat\smallskipneg\fi%
\paragraph{Star}
The star, denoted~$t^*$, of an element~$t$ in an arbitrary 
topological semiring~$\T$ (not only in a  
semiring of series) is defined if the family $\{t^n\}_{n\in\N}$ is 
summable and in this case, $\msp t^* = \sum_{n \in \N} t^n\msp$
and~$t$ is said to be \emph{starable}.
\index{starable}%
If~$t^{*}$ is defined, then
$\msp t^* = \unT + t\xmd t^* =\unT + t^*t \msp$
hold.
If moreover~$\T$ is a \emph{ring}, 
this
can be written
$\msp (1 - t)\xmd t^* = t^*(1 - t) = 1\msp $
and $\msp t^*\msp $ is the \emph{inverse} of $\msp 1 - t\msp $.
Generally in semirings, the star of an element may be viewed as a 
substitute of taking the inverse 
in a poor structure that has no inverse.
Hence is the name \emph{rational} given to objects that can be 
computed with the star.
 
The \emph{constant term} of a series~$s$ is 
\index{constant term!of a series}%
\index{series!constant term of --}%
the coefficient of the identity of~$M$:
$\msp\TermCst{s} = \bra{s,\unM }\msp$.
A series is \emph{proper} if its constant term is zero.
\index{proper|see{series}}%
\index{series!proper --}%
If $\msp s\msp$ is proper, the family $\{s^n\}_{n\in\N}$
is locally finite \emph{since~$M$ is graded} and 
the star of a proper series of~$\KM$ is thus always defined.

\ifcheat\smallskipneg\fi%[Arden]
\begin{lemma}
    \label{:lem:Ard-wei}%
Let~$s$ and~$t$ be two series in~$\KM$.
If~$s^{*}$ is defined, then $\msp s^{*}t\msp $ is the \emph{unique} 
\index{Arden's lemma}%
solution of the equation 
$\msp 
        \mathrm{X} = s \xmd \mathrm{X} + t  
\msp$.
\end{lemma}

\ifcheat\smallskipneg\fi%
\ifcheat\smallskipneg\fi%
\subsubsection{Rational series and expressions}
\label{:sec:rat-ser-exp}

The \emph{rational operations} on~$\KM $ are:
the two \emph{exterior multiplications} by elements of~$\K$, %the 
the \emph{addition}, the \emph{product},
and the \emph{star} which is not defined everywhere.
A subset $\Ec$ of~$\KM$ is \emph{closed under star} if 
for every~$s$ in~$\Ec$ such that $s^{*}$ is
defined then $s^{*}$
belongs to~$\Ec$.
% \index{rational!operation}%
% \index{operation!rational}%
\index{rational!closure}%
\index{rational!series}%
\index{series!rational --}%
The \emph{rational closure} of a set~$\Ec$, written~$\KRat\Ec$, is 
the \emph{smallest} subset of~$\KM$ closed under 
the rational operations and
which contains~$\Ec$.
The set of ($\K$-)\emph{rational series}, written~$\KRat M$, is the rational 
closure of~$\KPM$. 

\ifcheat\smallskipneg\fi%
\paragraph{Weighted rational expressions}

A rational expression on~$M$ with weight in~$\K$ 
--- a \emph{weighted expression}  --- 
is defined by completing \defin{rat-exp-fm} with \emph{two} 
operations \emph{for every~$k$ in~$\K$}:
if~$\Ed$ is an expression, then so are 
$\msp(k\xmd\Ed)\msp$ and~$\msp(\Ed\xmd k)$. %\msp
\index{expression!weighted rational --}%
The set of weighted rational expressions is written~$\msp\KREM$. %\msp
As for the languages, we write~$\msp\CompExpr{\Ed}\msp$ 
for the \emph{series denoted by~$\Ed$}, 
\index{expression!series denoted by --}%
\index{series!denoted|see{expression}}%
with the supplementary equations:
$\msp\CompExpr{(k\xmd\Ed)}= k\xmd\CompExpr{\Ed}\msp$ and
$\msp\CompExpr{(\Ed\xmd k)} = {\CompExpr{\Ed}}\xmd k$. %\msp

The \emph{constant term}~$\msp\TermCst{\Ed}\msp$ is defined as in 
\index{series!constant term of --}%
\defin{con-ter-exp} but for the last equation 
[$\TermCst{\autstar{\Fd}}=1$] which is replaced by:
$\msp`\TermCst{\autstar{\Fd}}=\autstar{\TermCst{\Fd}}\msp$ 
\emph{if the latter is defined}'.
An expression is \emph{valid} if its constant term is defined.
\index{expression!valid --}%
As~$M$ is graded, 
$\msp \TermCst{\Ed} = \bra{\CompExpr{\Ed},\unM}\msp$ 
holds for every valid weighted rational expression~$\Ed$.
Finally, the following holds:

\ifcheat\smallskipneg\fi%
\begin{proposition}
\label{:pro:rat-ser-1}%
A series of\/~$\KM$ is rational if and only if it is denoted by
a valid rational $\K$-expression over~$M$.
\end{proposition}

\ifcheat\smallskipneg\fi%
In this framework, we reformulate \lemme{Ard-wei} as:

\ifcheat\smallskipneg\fi%
\begin{corollary}
\label{:cor:Ard-wei}%
Let~$\Ud$ and~$\Vd$ be two expressions in~$\KREM$.
If\/~$(\TermCst{\Ud})^{*}$ is defined,\linebreak 
then~$\Ud^{*}\xmd\Vd$ denotes the unique 
solution of the equation 
$\msp\mathrm{X}=\CompExpr{\Ud}\xmd\mathrm{X}+\CompExpr{\Vd}  \msp$.
\end{corollary}

\ifcheat\smallskipneg\fi%
\ifcheat\smallskipneg\fi%
\subsubsection{Weighted automata and the fundamental theorem}
\label{:sec:wei-aut-fta}

An automaton~$\Ac$ over~$M$ with weight 
in~$\K$, a \emph{$\K$-automaton} for short, still written 
$\msp\Ac=\aut{Q,M,E,I,T}\msp$,
is an automaton where the \emph{sets} of initial and final states are 
replaced with \emph{maps} from~$Q$ to~$\K$, that is, every state has 
an initial and a final weight, and where the set~$E$ of transitions 
is contained in~$Q\x \K\x (M\bk\unM)\x Q$, that is, every transition is 
labelled with
a \emph{monomial} 
in~$\KPM$, different from a constant term.
\index{automaton!weighted --}%
The automaton~$\Ac$ is \emph{finite} if~$E$ is finite.

Alternatively, the same automaton is (more often) written
$\msp\Ac=\autiet\msp$, with the convention taken at \secti{rat-rec}: 
$E$ is the \emph{transition matrix} of~$\Ac$, a $Q\x Q$-matrix whose 
$(p,q)$-entry is the sum of the labels of all transitions from~$p$ 
to~$q$, and~$I$ and~$T$ are vectors in~$\K^{Q}$.
In this setting, $\Ac$ is finite if every entry of~$E$ is a 
polynomial of~$\KPM$.

The \emph{label} of a computation in~$\Ac$ is, as above, the product 
of the labels of the transitions that form the computation, multiplied 
(on the left) by the initial weight of the origin and (on the right) 
by the final weight of the end of the computation.
With the definition we have taken for automata (no transition 
labelled with a constant term), and because~$M$ is graded, the family 
of labels of all transitions of~$\Ac$ is \emph{summable} and
the \emph{series accepted} by~$\Ac$, also called \emph{behaviour} 
\index{automaton!series accepted by --}%
\index{automaton!behaviour of --}%
of~$\Ac$ and written~$\CompAuto{\Ac}\xmd$, is its sum.
\index{Fundamental theorem!of finite automata}%
The fundamental theorem of automata then reads:

\ifcheat\smallskipneg\fi%
\begin{theorem}
    \label{:the:fta-wgt-mon}%
Let~$M$ be a graded monoid.
A series of~\/$\KM$ is rational if and only if it is 
the behaviour of a finite $\K$-automaton over~$M$.
\end{theorem}

\ifcheat\medskipneg\fi
\ifcheat\smallskipneg\fi%
\subsubsection{Recognisable series}
\label{:sec:rec-ser}%

The distinction between \emph{rational} and \emph{recognisable} 
carries over from subsets of a monoid~$M$ to series over~$M$.
The equivalence between automata over free monoids and matrix 
representation (\cf \sbsct{sec-ste}) paves the way to the 
definition of recognisability.

A \emph{$\K$-representation} of~$M$ of dimension~$Q$ is a 
\index{representation}%
triple~$\lmn$ 
where~$\mu\colon M\rightarrow\K^{Q\x Q}$ is a morphism,
% where~$\mu$ is a morphism from~$M$ into~$\K^{Q\x Q}$
and~$\lambda$ and~$\nu$ are two vectors of~$\K^{Q}$.
The representation~$\lmn$ realises the series
$\msp s= \sum_{m\in M} (\lambda\cdot\mu(m)\cdot\nu)\xmd m \msp$;
a series in~$\KM$ is \emph{recognisable} if it is realised by a 
\index{series!recognisable --}%
representation and the set of recognisable series is denoted 
by~$\KRec M$.
The family of rational and of recognisable series are distinct in 
general.
A proof which is very similar to the one given 
at~\sbsct{sec-ste}, and which is \emph{independent from~$\K$}, yields the 
following. 

\ifcheat\smallskipneg\fi%
\begin{theorem}[Kleene--Sch{\"u}tzenberger]
    \label{:the:kle-sch-fm}%
    If~$A$ is finite, then
\index{Kleene's theorem}%
    $\msp\KRat\Ae=\KRec\Ae\msp$.
\end{theorem}

\ifcheat\smallskipneg\fi%
\ifcheat\smallskipneg\fi%
\ifcheat\smallskipneg\fi%
%%%%%%%%%%%%%%%%%%%%%%%%%
\subsection{From automata to expressions: the $\AtEs$-maps}
\label{:sec:aut-exp-mul}%

With the definition taken for a $\K$-automaton
$\msp\Ac=\autiet\msp$,
every entry of~$E$ is a \emph{proper polynomial} of~$\KPM$,
$E$ is in~$\KPM^{Q\x Q}$, 
hence a proper polynomial of~$\PolSAnMon{\K^{Q\x Q}}{M}$,
and~$E^{*}$ is well-defined.
\lemme{mat-mul-gra} generalises to $\K$-automata 
and
$\msp\CompAuto{\Ac} = I\matmul E^{*}\matmul T\msp$
holds.

In every respect, the weighted case is similar to the Boolean one.
The direct part of \theor{fta-wgt-mon} follows from the generalised 
statement of \propo{sta-mat}:

\ifcheat\smallskipneg\fi%
\begin{proposition}
    \label{:pro:sta-mat-wgt}%~$\Ed$
    The entries of the star of a proper matrix~$E$ of~$\KM^{Q\x Q}$ 
    belong to the rational closure of the entries of~$E$.
\end{proposition}

\ifcheat\smallskipneg\fi%
The same algorithms as those presented at \secti{aut-exp}: the 
state-elimination and system-solution methods, 
the McNaughton--Yamada and 
recursive algorithms,
\index{method!state-elimination}%
\index{method!system-solution}%
\index{McNaughton--Yamada algorithm}%
\index{method!recursive}%
establish the\linebreak 
weighted version of \propo{phi-map-fm}:

\begin{proposition}
    \label{:pro:phi-map-wgt}%~$\Ed$[$\AtEs$-maps]$\K$-
Let~$M$ be a graded monoid.
For every finite $\K$-automaton~$\Ac$ over~$M$, there exist rational 
expressions over~$M$ which denote~$\CompAuto{\Ac}$.
\end{proposition}

\ifcheat\smallskipneg\fi%
If the algorithms are the same, one has to establish nevertheless 
their correctness in this new and more complex framework.
We develop the case of the system-solution method, the other ones 
could be treated in the same way.
To begin with, 
we have to 
enrich the set of \emph{trivial identities} in order to set up the 
definition of \emph{reduced weighted expressions}, which in turn is 
\index{expression!reduced --}%
necessary to define computations on expressions.
The set~$\Tmbf$ as defined at \secti{Ard-rat-ide} is now denoted 
as~$\Tmbf_{\mathsf{u}}$:
\begin{equation}
\Ed \autplus \zed \equiv \Ed \EqVrgInt \zed \autplus \Ed \equiv \Ed \EqVrgInt
\Ed \autprod \zed \equiv \zed \EqVrgInt \zed \autprod \Ed \equiv \zed 
\EqVrgInt
\Ed \autprod \und \equiv \Ed \EqVrgInt \und \autprod \Ed \equiv \Ed  \EqVrgInt
\autstar \zed \equiv \und
\hspace{.4em}
\tag{$\Tmbf_{\mathsf{u}}$}	
\end{equation}
and augmented with three other sets of identities:
\begin{gather}
{\zeK}\xmd\Ed  \equiv \zed   \EqVrgInt \Ed\xmd{\zeK} \equiv \zed \EqVrgInt
{k}\xmd\zed \equiv \zed \EqVrgInt\zed\xmd{k} \equiv \zed \EqVrgInt
{\unK}\xmd\Ed  \equiv \Ed \EqVrgInt \Ed\xmd{\unK}  \equiv  \Ed
\tag{$\Tmbf_{\K}$}
\\[.7ex]
{k}\xmd({h}\xmd\Ed)  \equiv {kh}\xmd\Ed \EqVrgInt
(\Ed\xmd{k})\xmd{h}  \equiv \Ed\xmd{kh}  \EqVrgInt
({k}\xmd\Ed)\xmd{h} \equiv {k}\xmd(\Ed\xmd{h})
\tag{$\Ambf_{\K}$}
\\[.7ex]
  \und\xmd{k}  \equiv  {k}\xmd\und \EqVrgInt
  \Ed\autprod({k}\xmd\und)   \equiv  \Ed\xmd{k} \EqVrgInt
  ({k}\xmd\und)\autprod\Ed   \equiv  {k}\xmd\Ed
\tag{$\Umbf_{\K}$}
\end{gather}
From now on, all computations on weighted expressions are performed 
modulo the \emph{trivial identities}
\index{identities!trivial --}%
$\msp
\Tmbf= \Tmbf_{\mathsf{u}}\land\Tmbf_{\K}\land\Ambf_{\K}\land\Umbf_{\K}
\msp$.
Besides the trivial identities, the 
\emph{natural identities}
\index{identities!natural --}%
$\msp
\Nmbf= \Ambf\land\Dmbf\land\Cmbf
\msp$
hold on the expressions of~$\KREM$ for any~$\K$ and 
(graded)~$M$, 
and, in contrast, the identities~$\Imbf$ 
and~$\Jmbf$ that are special to~$\jsPart{M}$ do not hold anymore.

\ifcheat\smallskipneg\fi%
\paragraph{The system-solution method}
starts from a proper automaton
\index{method!system-solution}%
$\msp\Ac=\autiet\msp$
of dimension~$Q$ whose behaviour is
$\msp\CompAuto{\Ac}=I\matmul V\msp$
where $\msp V= E^{*}\matmul T\msp$
is a vector in~$\KM^{Q}$.
\lemme{Ard-wei} easily generalises and as~$E$ is proper (in~$\KQQM$),
\index{Arden's lemma}%
$V$~is the unique solution of the equation
$\msp \mathrm{X} = E \xmd \mathrm{X} + T\msp$which we rewrite as a 
system of 
$\jsCard{Q}$ equations:
\begin{equation}
\textstyle{\fa p\in Q\quantsp
V_{p} = \sum_{q\in Q}\CompExpr{\Ed_{p,q}}\xmd V_{q} + \CompExpr{T_{p}\xmd\und} 
\eee
}
\label{:equ:wei-res-1}%
\end{equation}
where the~$V_{p}$ are the `unknowns',
where the entries~$E_{p,q}$, which 
are linear combinations of elements of~$M$, are considered as 
expressions and denoted as such
and where~$\CompExpr{T_{p}\xmd\und}$ is the series reduced to the 
monomial~$T_{p}\xmd\unM$.
The system~\equnm{wei-res-1} may be solved by successive 
\emph{elimination} of the unknowns, by means of \corol{Ard-wei}.
When all unknowns~$V_{q}$ have been eliminated following an 
order~$\omega$ on~$Q$, the computation yields an 
expression that we denote by~$\srm{\Ac}$, 
as in \sbsct{sys-res-met}, 
and 
$\msp\CompAuto{\Ac}=\CompExpr{\srm{\Ac}}\msp$
holds.

The parallel with the Boolean case can be carried on:
given a $\K$-automaton~$\Ac$ of dimension~$Q$, an ordering~$\omega$, and 
a recursive division~$\tau$ 
on~$Q$, the expressions~$\sem{\Ac}$, $\mny{\Ac}$, and~$\rcm{\Ac}$ 
that all denote~$\CompAuto{\Ac}\xmd$ are 
computed by the state-elimination method, the McNaughton--Yamada and 
recursive algorithms respectively.
The results on the comparison between these expressions also extend 
to the weighted case.

\ifcheat\smallskipneg\fi%
\begin{proposition}
    \label{:pro:wei-eli-equ}%
For every order~$\omega$ on~$Q$,
$\msp\sem{\Ac} = \srm{\Ac}\msp$ holds.
\end{proposition}

\ifcheat\smallskipneg\fi%
\ifcheat\smallskipneg\fi%
\begin{proposition}
    \label{:pro:wei-com-mya-sem}
For every order~$\omega$ on~$Q$,
$\msp\IdRN\land\IdRA \dedtxt\mny{\Ac_{p,q}}\equiv\sem{\Ac_{p,q}}\msp$
holds.
\end{proposition}

\ifcheat\smallskipneg\fi%
\theor{ord-eli-met} also extends to the weighted case
(and it is now clear why it was important that identities~$\IdRI$ 
and~$\IdRJ$ do not play a role in that result).
% with the important change that the identity~$\Imbf$ does not come 
% into play anymore. 

\ifcheat\smallskipneg\fi%
\begin{theorem}
    \label{:the:wei-ord-eli-met}%
Let~$\omega $ and~$\omega' $ be two orders on the set of states of a 
$\K$-automaton~$\Ac$.\linebreak  
Then,
$\msp
       \IdRN \land \IdRS \land \IdRP \e \dedtxt 
        \sem{\Ac} \e \equiv \e \sem[\omega']{\Ac}
\msp$ holds.
\end{theorem}

\ifcheat\smallskipneg\fi%
\ifcheat\smallskipneg\fi%
\ifcheat\smallskipneg\fi%
\subsection{From expressions to automata: the $\EtAs$-maps}
\label{:sec:exp-aut-mul}%

\ifcheat\smallskipneg\fi%
\subsubsection{The standard automaton of a weighted expression}

The definition of a \emph{standard} weighted automaton is the same as 
the one of a standard automaton for the Boolean case:
a unique initial state on which the initial map takes the value~$\unK$ 
and which is not the end of any transition.
\index{automaton!standard weighted --}%
Such an automaton may thus be represented as in \figur{sta-aut} and every 
weighted automaton is equivalent to, and may be turned into, a 
standard one.

As in the Boolean case, operations are defined on standard weighted 
automata that are parallel to the rational weighted operators.
With the notation of \figur{sta-aut}, the operators
$\Ac+\Bc$ and ${\Ac}\matmul{\Bc}$ are given by~\equnm{sta-aut-sum} 
and~\equnm{sta-aut-pro}, 
$k\xmd\Ac$ and~$\Ac\xmd k$ by

\ifcheat\smallskipneg\fi%
\ifcheat\smallskipneg\fi%
\begin{equation}
    k\xmd{\Ac} =
\autsk{\redmatu{\lignedblvs{1}{0}},
          \redmatu{\matriceddblvs{0}{k\xmd J}{0}{F}},
          \redmatu{\vecteurdblvs{k\xmd c}{U}}}
		  ,\;
    {\Ac}\xmd k=
\autsk{\redmatu{\lignedblvs{1}{0}},
          \redmatu{\matriceddblvs{0}{J}{0}{F}},
          \redmatu{\vecteurdblvs{c\xmd k}{U\xmd k}}}
,
     \notag
\end{equation}
and ${\Ac}^{*}$, which is defined when~$c^{*}$ is defined, 
by the following modification of~\equnm{sta-aut-sta}: 

\ifcheat\smallskipneg\fi%
\begin{equation}
    {\Ac}^{*} =
\aut{\redmatu{\lignedblvs{1}{0}},
     \redmatu{\matriceddblvs{0}{c^{*}\xmd J}{0}{H}},
     \redmatu{\vecteurdblvs{c^{*}}{U\xmd c^{*}}}}\eqvrg
\label{:equ:wei-sta-aut-sta}
\tag{\ref{:equ:sta-aut-sta}'}
\end{equation}
where $\msp H= U\matmul c^{*}\xmd J + F\msp$.
As in \secti{sta-aut-exp},
these operations allow to associate with every weighted 
expression~$\Ed$ and by induction on its depth,  
a standard weighted automaton~$\Stan{\Ed}$ which we call 
\emph{the} standard automaton of~$\Ed$.
\index{automaton (of an expression)!standard --}%
\index{automaton (of an expression)!Glushkov --}%
Straightforward computations show that
$\msp\CompAuto{(k\xmd\Ac)}= k\xmd\CompAuto{\Ac}\msp$,
$\msp\CompAuto{(\Ac\xmd k)}= \CompAuto{\Ac}\xmd k\msp$,
$\msp\CompAuto{(\Ac+\Bc)}= \CompAuto{\Ac}+\CompAuto{\Bc}\msp$,
$\msp\CompAuto{(\Ac\matmul\Bc)}=\CompAuto{\Ac}\matmul\CompAuto{\Bc}\msp$
and
$\msp\CompAuto{(\Ac^{*})}=\CompAuto{\Ac}^{*}$. %\msp
From which one concludes that the construction of~$\Stan{\Ed}$ is a 
$\EtAs$-map:

\ifcheat\smallskipneg\fi%
\begin{proposition}[\cite{CaroFlou00,LombSaka05a}]
\label{:pro:wei-sta-aut-exp}
If~\/$\Ed$ is a weighted expression over~$\Ae$, then
$\msp\CompAuto{\Stan{\Ed}}=\CompExpr{\Ed}\msp$.
\end{proposition}

\ifcheat\smallskipneg\fi%
The automaton~$\Stan{\Ed}$ has~$\LitL{\Ed}+1$ states. 
Computing~$\Stan{\Ed}$ from~\equnm{sta-aut-sum}, \equnm{sta-aut-pro} 
and~\equnm{wei-sta-aut-sta} is \emph{cubic} in~$\LitL{\Ed}$ and a 
\emph{star-normal form} for weighted expressions is something that does not 
\index{star-normal form!of an expression}%
seem to exist in the general case.
\figur{aut-sta-1} shows the standard $\Q$-automaton $\Sc_{\Ed_{3}}$ 
associated with 
$\msp\Ed_{3}=(\fracts{1}{6}\xmd a^{*}+\fracts{1}{3}\xmd b^{*})^{*}\msp$
and $\Z$-automaton $\Sc_{\Ed_{4}} $associated with
$\msp\Ed_{4}=(1-a)\xmd a^{*}$. %\msp
% the standard 

\begin{figure}[htbp]
\centering
\setlength{\lga}{4cm}%
\VCDraw{%
\begin{VCPicture}{(-1.4,-1)(9.4,1.6)}
% states
\SmallState
\State{(0,0)}{A}
\State{(\lga,0)}{B}\State{(2\lga,0)}{C}
\InitialL{w}{A}{1}\FinalR{s}{A}{2}
\FinalL{s}{B}{2}\FinalL{s}{C}{2}
% transitions  
\EdgeR{A}{B}{\fracts{1}{3}\xmd a}
\VArcL[.1]{arcangle=37.5,ncurv=.9}{A}{C}{\fracts{2}{3}\xmd b}
\ArcL{C}{B}{\fracts{1}{3}\xmd a}
\ArcL{B}{C}{\fracts{2}{3}\xmd b}
\LoopN[.2]{B}{\fracts{4}{3}\xmd a}
\LoopN[.75]{C}{\fracts{5}{3}\xmd b}
\end{VCPicture}%
}
\setlength{\lga}{3cm}%
\eee
\VCDraw{%
\begin{VCPicture}{(-1.4,-1)(7.4,1.6)}
% states
\SmallState
\State{(0,0)}{A}
\State{(\lga,0)}{B}\State{(2\lga,0)}{C}
\InitialL{w}{A}{1}\FinalR{s}{A}{1}
\FinalL{s}{B}{1}\FinalL{e}{C}{1}
% transitions  
\EdgeR{A}{B}{-\xmd a}
\EdgeR{B}{C}{a}
\LArcL{A}{C}{a}
\LoopN{C}{a}
\end{VCPicture}%
}
\caption{The $\Q$-automaton $\Sc_{\Ed_{3}}$ and 
the $\Z$-automaton $\Sc_{\Ed_{4}}$}
\label{:fig:aut-sta-1}%
\end{figure}

% \ifcheat\smallskipneg\fi%
It is the necessary definition of~$k\xmd\Ac$ and~$\Ac\xmd k$ that 
rules out the equivalence
$\msp k\xmd m\equiv m\xmd k\msp$,
\index{identities!trivial --}%
with~$m$ in~$M$, from the set of trivial identities.

\ifcheat\smallskipneg\fi%
\subsubsection{The derived-term automaton of a weighted expression}

The \emph{(left) quotient} operation also extends from languages to 
series:
\index{quotient!of a series}%
for every~$s$ in~$\KA$, and every~$u$ in~$\Ae$,
$\msp u^{-1}s\msp$
is defined by
$\msp\bra{u^{-1}s,v} = \bra{s,u\xmd v}\msp$
for every~$v$ in~$\Ae$.
The quotient is a \emph{(right) action} of~$\Ae$ on~$\KA$:
$\msp(u\xmd v)^{-1}s = v^{-1}\left(u^{-1}s\right)$. %\msp

In contrast with the Boolean case, a series in~$\KRat\Ae$ may have an 
infinite number of distinct quotients.
However, the quotient operation allows to express a characteristic 
property of rational series. 
Let us call \emph{stable} 
\index{stable}%
a subset~$U$ of~$\KA$ closed under quotient.
% Then, a characterisation, due to Jacob~\cite{JacoG75}, reads:
Then, a characterisation due to Jacob reads:
\emph{a series of~$\KA$ is rational if and only if it is contained 
in a finitely generated stable submodule of~$\KA$}, \cf 
\cite{BersReut84,Saka09bhb}.

\ifcheat\smallskipneg\fi%
\paragraph{Derivation}

The \emph{derivation} of  weighted rational expressions implements 
the lifting of the quotient of series to the level of expressions.
It yields an effective version of the characterisation quoted above.

In the sequel, addition in~$\K$
is written~$\Kadd $ to
distinguish it from the~$+$ operator in expressions.
The set of \emph{(left) linear combinations} of
$\K$-expressions with coefficients in~$\K$ 
is denoted, by abuse, by~$\KKREA$.
In the following,~$[k\,\Ed]$ or~$k\,\Ed$ is a monomial 
in~$\KKREA$  
whereas~$(k\,\Ed)$ is an expression in~$\KRA$.
An external right multiplication
on~$\KKREA$ by an expression and by a 
scalar is needed in the sequel. 
It is first defined  
on monomials by
$\msp ([k\,\Ed]\cdot\Fd) \equiv k\, (\Ed\cdot\Fd)\msp$
and
$\msp ([k\,\Ed]\, k') \equiv k\, (\Ed\, k')\msp$
and then extended to~$\KKREA$ by linearity.

\ifcheat\smallskipneg\fi%
\begin{definition}[\cite{LombSaka05a}]
    \label{:def:wei-exp-der}%
The \emph{derivation} of~$\Ed$  in~$\KRA$ with respect to~$a$ in~$A$, 
\index{derivation (of an expression)!K@$\K$- --}%
denoted by~${\ExpDer{\Ed}}$, 
is a linear combination of expressions in~$\KRA$
defined 
by~\equnm{der-b-1} for the base cases and inductively 
by the following formulas.
{\iesf
{\allowdisplaybreaks
\begin{align}
\ExpDer{(k\, \Ed)}  = k\, \ExpDer{\Ed}
\EqVrgInt
\ExpDer{(\Ed\, k)}  
= 
\left(\left[\ExpDer{\Ed}\right]\, k\right)
&  
\EqVrgInt
\ExpDer{(\Ed \mathsf{+} \Fd)} 
  = \ExpDer{\Ed} \oplus  \ExpDer{\Fd} \eqvrg
\notag
\\[1ex]
\ExpDer{(\Ed \cdot  \Fd)} 
= \left(\left[\ExpDer{\Ed}\right] \cdot \Fd\right) \oplus  
   \TermCst{\Ed} \, \ExpDer{\Fd}
\EqVrgInt
&
\text{ and } \e
\ExpDer{(\Ed^{*})} 
= \TermCst{\Ed}^{*}\, 
   \left(\left[\ExpDer{\Ed}\right]\cdot (\Ed^{*})\right)
   \eqpnt
   \e 
\notag
\end{align}
}%
}%
\end{definition}

The last equation
is defined only if~$\Ed^*$
is a \emph{valid} expression.
The derivation of an expression with respect to a 
\emph{word}~$u$ is defined by induction on the length of~$u$:
for every~$u$ in~$A^{+}$ and every~$a$ in~$A$,
$\msp\ExpDer[ua]{\Ed} = \ExpDer{ \left(\ExpDer[u]{\Ed}\right)}\msp$
and the definition of derivation is consistent 
with that of quotient of series
since for every~$u$ in~$A^{+}$,
$\msp\CompExpr{\ExpDer[u]{(\Ed)}} = u^{-1}\CompExpr{\Ed}\msp$
holds.

\ifcheat\smallskipneg\fi%
\paragraph{The derived-term automaton}

At \sbsct{der-ter-aut}, we have defined the \emph{derived terms} of a 
(Boolean) expression as the expressions that occur in a derivation 
of that expression.
\propo{tru-der-ind} then established properties that allow to compute 
these derived terms, without derivation.
For the  weighted case, we  take the same properties as 
the definition.

\ifcheat\smallskipneg\fi%
\begin{definition}[\cite{LombSaka05a}]
    \label{:def:wei-der-ter}
The set~$\TruDerTer{\Ed}$ of \emph{true derived terms} 
\index{derived term (of an expression)!true --}%
of~$\Ed$ in~$\KRA$ 
is inductively defined by:  
$\TruDerTer{k\,\Ed}=\TruDerTer{\Ed}$,\e
$\TruDerTer{\Ed\,k}=\left(\TruDerTer{\Ed}\,k\right)$,\e 
$ \TruDerTer{\Ed \autplus \Fd} 
   = \TruDerTer{\Ed} \cupssp  \TruDerTer{\Fd}$,\e
$\TruDerTer{\Ed \autprod  \Fd} 
    =\left(\TruDerTer{\Ed}\right)\autprod\Fd
         \cupssp   \TruDerTer{\Fd} $,\e
$\TruDerTer{\Ed^{*}} 
     = \left(\TruDerTer{\Ed}\right)\autprod \Ed^{*}$,
 starting from the base cases
$\msp\TruDerTer{\zed}=\TruDerTer{\und}=\emptyset\msp$,
and
$\msp\TruDerTer{a}=\{\und\}\msp$ for every~$a$ in~$A$.
\end{definition}

\ifcheat\smallskipneg\fi%
$\TruDerTer{\Ed}$ is a 
set of \emph{unitary} monomials of~$\KKREA$,
with~$\jsCard{\TruDerTer{\Ed}}\leq\LitL{\Ed}\msp$.
The set of \emph{derived terms} of~$\Ed$
\index{derived term (of an expression)}%
is
$\msp\DerTer{\Ed}=\TruDerTer{\Ed}\cup\{\Ed\}\msp$.
\theor{ant-mul} insures consistency between
Definitions~\ref{:def:wei-exp-der} and~\ref{:def:wei-der-ter};
the usefulness of the latter follows from \theor{th-cov}.

\ifcheat\smallskipneg\fi%
\begin{theorem}[\cite{Rutt03,LombSaka05a}]
    \label{:the:ant-mul}%
Let~$\Ed$ be in~$\KRA$ and $\msp\DerTer{\Ed}=\{\Kd_1,...,\Kd_n\}$.
For every~$a$ in~$ A$, there exist an $n\x n$-matrix~$\mu(a)$ 
with entries in~$\K$ such that
\begin{equation}
\textstyle{\forall i \in [n] \quantsp
   \ExpDer{\Kd_{i}} = 
   \Ksum_{j\in [n]} \mu(a)_{i,j}\xmd \Kd_{j} \EqPnt
   \eee
   \notag
}%    
\end{equation}
\end{theorem}

\ifcheat\smallskipneg\fi%
\ifcheat\smallskipneg\fi%
The derivation of an expression~$\Ed$ in~$\KRA$ with respect 
to every word in~$A^{+}$ is thus a linear combination of derived 
terms of~$\Ed$. 
Hence the derived terms of an expression denote 
\emph{the generators of a stable submodule} 
that contains the series denoted by the expression.
\theor{ant-mul} yields the \emph{the derived-term automaton} of~$\Ed$,
\index{automaton (of an expression)!derived term --}%
$\msp\Ac_{\Ed}=\aut{I,X,T}\msp$, of dimension~$\DerTer{\Ed}$,
with
$\msp I_{}  = \unK\msp$ if~$\Kd_{i} = \Ed$ and~$\zeK$ otherwise,
$\msp X =   \Ksum_{a\in A} \mu(a)\xmd a \msp$, and
$\msp T_{j}  = \TermCst{\Kd_{j}}\msp$.
The $\K$-derivation is another $\EtAs$-map since
$\msp\CompExpr{\Ac_{\Ed}} = \CompExpr{\Ed}\msp$
holds. 

Morphisms and quotients of (Boolean) automata are generalised to
\emph{Out-morphisms} 
% \emph{$\K$-coverings} and \emph{$\K$-quotients} 
and \emph{quotients}
of $\K$-automata (\cf 
\cite{Saka09bhb,BealEtAl06}).
\theor{der-ter-sta} is then extended to the weighted case.

\ifcheat\smallskipneg\fi%
\begin{theorem}[\cite{LombSaka05a}]
    \label{:the:th-cov}%$\K$-
Let~$\Ed$ be in~$\KRA$.
Then~$\Ac_{\Ed}$ is a quotient of~$\Sc_{\Ed}$. 
\end{theorem}

\ifcheat\smallskipneg\fi%
\ifcheat\smallskipneg\fi%
\ifcheat\smallskipneg\fi%
\begin{remark}
\label{:rem:der-ter}%
This statement is a justification for \defin{wei-der-ter}.
The monomials that appear in the derivations of an 
expression~$\Ed$ are in~$\DerTer{\Ed}$.
The converse is not 
necessarily true when~$\K$ is not a positive semiring: some derived 
terms may never occur in a derivation,
as it can be observed for instance on the $\Z$-expression
$\msp\Ed_{4} = (1-a)\xmd a^{*}\msp$ (\cf \figur{ter-der-2}).
With a definition of derived terms based on derivation only, 
\theor{th-cov} would not hold anymore. 
\end{remark}

\ifcheat\smallskipneg\fi%
\ifcheat\smallskipneg\fi%
\begin{figure}[htbp]
\centering
\setlength{\lga}{3cm}%
\ifcheat\smallskipneg\fi%
\VCDraw{%
\begin{VCPicture}{(-1.4,-1)(7.4,1.6)}
% states
\SmallState
\State{(0,0)}{A}
\State{(\lga,0)}{B}\State{(2\lga,0)}{C}
\InitialL{w}{A}{1}\FinalR{s}{A}{1}
\FinalL{s}{B}{1}\FinalL{e}{C}{1}
% transitions  
\EdgeR{A}{B}{-\xmd a}
\EdgeR{B}{C}{a}
\LArcL{A}{C}{a}
\LoopN{C}{a}
\end{VCPicture}%
}
\eee
\VCDraw{%
\begin{VCPicture}{(-4,-1)(7,1.6)}
\LargeState
\SetAOSLengthCoef{1}%
\State[\Ed_{2}]{(0,0)}{A}\State[a^{*}]{(\lga,0)}{C}
\Initial{A}\Final[s]{A}\Final{C}
\LoopN{C}{a}
\end{VCPicture}}
\caption{The $\Z$-automaton $\Sc_{\Ed_{4}}$ and its $\Z$-quotient 
$\Ac_{\Ed_{4}}$} 
\label{:fig:ter-der-2}
\end{figure}

% \ifcheat\medskipneg\fi
% \ifcheat\smallskipneg\fi%
% \subsection{Recognisable series}
% \label{:sec:rec-ser}%
% 
% The distinction between \emph{rational} and \emph{recognisable} 
% carries over from subsets of a monoid~$M$ to series over~$M$.
% The equivalence between automata over free monoids and matrix 
% representation (\cf \sbsct{sec-ste}) paves the way to the 
% definition of recognisability.
% 
% A \emph{$\K$-representation} of~$M$ of dimension~$Q$ is a 
% \index{representation}%
% triple~$\lmn$ 
% where~$\mu\colon M\rightarrow\K^{Q\x Q}$ is a morphism,
% % where~$\mu$ is a morphism from~$M$ into~$\K^{Q\x Q}$
% and~$\lambda$ and~$\nu$ are two vectors of~$\K^{Q}$.
% The representation~$\lmn$ realises the series
% $\msp s= \sum_{m\in M} (\lambda\cdot\mu(m)\cdot\nu)\xmd m \msp$;
% a series in~$\KM$ is \emph{recognisable} if it is realised by a 
% \index{series!recognisable --}%
% representation and the set of recognisable series is denoted 
% by~$\KRec M$.
% The family of rational and of recognisable series are distinct in 
% general, and, a proof which is very similar to the one given 
% at~\sbsct{sec-ste}, and \emph{independent from~$\K$}, yields the 
% following. 
% 
% \ifcheat\smallskipneg\fi%
% \begin{theorem}[Kleene--Sch{\"u}tzenberger]
%     \label{:the:kle-sch-fm}%
%     If~$A$ is finite, then
% \index{Kleene's theorem}%
%     $\msp\KRat\Ae=\KRec\Ae\msp$.
% \end{theorem}
% 

%%%%%%  AE Section Notes  %%%%%%%%%%
%%%             100721                   %%%
%%%%%%%%%%%%%%%%%%%%%%%%%%%%%%%%%%%%%%%%%%%%
\ifcheat\smallskipneg\fi%
\ifcheat\smallskipneg\fi%
\ifcheat\smallskipneg\fi%
\ifcheat\smallskipneg\fi%
\section{Notes}
\label{:sec:not}

{\iesf
Most of the material presented in this chapter has appeared in 
previous work of the author \cite{Saka03,Saka05,Saka09bhb}.
\shortonly{A detailed version of this chapter is to be found at
\texttt{http://arxiv.org/abs/xxx.yyy}}
% available on-line. % \cite{Saka11}., including proofs,

\ifcheat\smallskipneg\fi%
\paragraph{\secti{new-loo}. New look at Kleene's theorem}%
A detailled history of the development of ideas at the beginning of 
the theory of automata is given in~\cite{PerrD93}.
Berstel \cite{Bers79} attributes to Eilenberg the idea of
\index{recognisable!subset}%
\index{rational!subset}%
distinguishing the family of recognisable from that of rational sets.

Besides the already quoted Elgot and Mezei's paper \cite{ElgoMeze65}, 
other authors have certainly noticed the equality of expressiveness 
of automata and expressions beyond free monoids.
It is part of Walljasper's thesis \cite{Wall70}; it can be found in 
Eilenberg's treatise \cite{Eile74}.
The splitting of Kleene's theorem has been proposed in~\cite{Saka87b}.

\paragraph{\secti{aut-exp}. From automata to expressions}%
First note that this section is mostly of 
theoretical interest: for which practical purpose would one exchange an 
automaton for an expression?

\textit{Identities.\ } As mentioned, the axiomatisation of rational 
expressions, even  
hinting at bibliographic references, is out of the scope of this 
chapter.
Conway showed that besides the identities~$\Smbf$ and~$\Pmbf$ (that 
are at the basis of the definition of the 
so-called `Conway semirings', \cf~\CTchp{ZE}), each finite simple 
\index{semiring!Conway --}%
\emph{group} gives rise to an identity that is independent from the 
others~\cite{Conw71}.
Krob, who showed that this set of identities is complete, 
coined~$\Smbf$ and~$\Pmbf$ the \emph{aperiodic 
\index{identities!aperiodic --}%
identities}~\cite{Krob91}. 

\textit{State-elimination method.\ }
\index{method!state-elimination}%
The example~$\Dc_{3}$ of \figur{sta-eli-exa} is easily generalised so as 
to find an exponential gap between the length of expressions for two 
distinct orders.
The search for short expressions is 
performed by heuristics; 
as reported in~\cite{GrubEtAl09},
the naive one, modified or not as in~\cite{DelgMora04},
appears to be good (\cf~\CTchp{GHK} for more information on the 
subject).

\textit{McNaughton--Yamada algorithm\ } is the implementation in the 
\index{McNaughton--Yamada algorithm}%
semiring of languages of the comtemporary Floyd--Roy--Warshall 
algorithms (in the Boolean or tropical 
semirings)~\cite{Floy62,Roy59,Wars62}.

\textit{Star height.\ } 
% The definition of star height of expressions  leads to the
% one of 
The star height of a rational language~$L$ is \emph{the minimum of
the star heights of the rational expressions that denote~$L$}.
\index{star height!of a rational language}%
Whether the star height of a language is effectively computable has 
been a long standing open problem until it was positively solved 
\index{star height!problem}%
first by 
K.~Hashiguchi~\cite{Hash88} and then by D.~Kirsten~\cite{Kirs05}.

\ifcheat\smallskipneg\fi%
\ifcheat\miniskipneg\fi%
\paragraph{\secti{exp-aut}. From expressions to automata}%
The presentation of the standard automaton given here is not 
\index{automaton (of an expression)!standard --}%
the classical one, and not only for the chosen name.
The recursive definition, also used in~\cite{FischEtAl10} for 
instance, avoids the 
definition of 
\texttt{First}, \texttt{Last}, and 
\texttt{Follow} functions that are built in most papers on the 
subject.
\index{automaton (of an expression)!Glushkov --}%
\index{automaton (of an expression)!position --}%
% The definition of these functions has been modified almost all 
% in~\cite{IlieYu03} in order to build a possibly smaller automaton, 
% called the \emph{follow automaton} of the expression.
Based on these functions, other automata may be defined: \eg 
in~\cite{MNauYama60} they are used to compute directly the 
\emph{determinisation} of~$\Stan{\Ed}$,
in~\cite{IlieYu03} positions with the same image by 
\texttt{Follow}    
% the \texttt{Follow} function  
are merged, giving rise to a possibly smaller automaton, called
\index{follow automaton|see{automaton (of...)}}%
\index{automaton (of an expression)!follow --}%
\emph{follow automaton}.

% \textit{Derived term automaton.\ } definition of
Attributing derivation to Brzozowski and Antimirov together is an 
unusual but sensible foreshortening.
Original Brzozowski's \emph{derivatives}~\cite{Brzo64} 
\index{derivative (of an expression)}%
are obtained by replacing~`$\cup$' by a~`$+$' in~\equnm{der-b-2} 
and~\equnm{der-b-3}.
Derivatives are then \emph{expressions}, and there is a finite number of 
them, modulo the~$\Ambf$, $\Cmbf$, and~$\Imbf$ identities.
By replacing the~`$+$' by a~`$\cup$' in Brzozowski's definition, 
Antimirov~\cite{Anti96} changed the derivatives into a \emph{set of 
expressions},  
which he called \emph{partial derivatives}, as they are `parts' of 
\index{derivative (of an expression)!partial --}%
derivatives.
As they are applied to union of sets, and not to expressions, 
% the necessary use of 
the~$\Ambf$, $\Cmbf$, and~$\Imbf$ identities come for free, and are 
no longer necessary to insure the finiteness of derived terms.

% Derived terms were called \emph{partial derivatives} by Antimirov 
% in~\cite{Anti96}. 
% There are at least two reasons for this renaming: first the
% (Brzozowski) derivative expressions of~$\Ed$ were already `partial'
% inasmuch as they are the result of a derivation with respect to
% \emph{one} of the letters, and because `partial derivatives' further
% overloads an established mathematical term.

A common technique for defining $\EtAs$-maps has been the 
\emph{linearisation}~$\Edl$ of the expression~$\Ed$, \ie making all 
letters in~$\Ed$ distinct by indexing them by their position in~$\Ed$ 
(\eg~\cite{MNauYama60,IlieYu03}).
Berry--Sethi~\cite{BerrSeth86} showed that the (Brzozowski) 
derivatives of~$\Edl$  
coincide with the states of~$\Stan{\Ed}$, whereas 
Berstel--Pin~\cite{BersPin96} observed that~$\CompExpr{\Ed}$ is a \emph{local} 
language~$\ExprLine{L}$ and interpreted Berry-Sethi's result as the 
construction of \emph{the} deterministic automaton canonically 
associated with~$\ExprLine{L}$.

The similarity between Mirkin's prebases~\cite{Mirk66} and Antimirov's derived
\index{prebase}%
terms was noted by Champarnaud--Ziadi~\cite{ChamZiad01}, who called 
\emph{equation  automaton} the derived-term automaton.
\index{equation automaton|see{automaton  (of...)}}%
\index{automaton (of an expression)!equation --}%
\index{automaton (of an expression)!derived term --}%

Allauzen--Mohri have generalised \propo{Tho-sta} and 
\theor{der-ter-sta} and computed~$\DTAut{\Ed}$ and the follow 
automaton of~$\Ed$ from~$\Thom{\Ed}$ by quotient and elimination of 
marked spontaneous transitions~\cite{AllaMohr06}.

In~\cite{CaroZiad00}, an algorithm is given which is a 
kind of converse of a $\EtAs$-map: it recognises if an automaton is 
the standard automaton~$\Stan{\Ed}$ of an expression~$\Ed$ and, in
\index{star-normal form!expression in --}%
this case, computes such an~$\Ed$ in \emph{star-normal form}.
The problem of inverting a $\AtEs$-map has been given a partial answer 
in~\cite{LombSaka05bbis}: 
it is possible to compute~$\Ac$ from~$\sem{\Ac}$ for certain~$\Ac$ 
(and any~$\omega$);
this has lead to the definition of a variant of the derivation: the 
\emph{broken derivation}, that has been further studied 
\index{derivation (of an expression)!broken --}%
\index{broken derivation|see{derivation}}%
in~\cite{AngrEtAl10}.

\ifcheat\smallskipneg\fi%
\ifcheat\miniskipneg\fi%
\paragraph{\secti{cha-mon}. Changing the monoid}%
\propo{rec-fin-gen}
leads naturally to consider monoids~$M$ in which
$\Rat M = \Rec M\msp$ holds, and which one could call \emph{Kleene 
monoids}. 
\index{Kleene monoid}%
\index{monoid!Kleene --}%
In~\cite{Saka87c} was defined the family of \emph{rational monoids} 
which contains all previously known examples of Kleene monoids;  
still the inclusion is strict~\cite{PellSaka90}.
\index{rational!monoid}%
\index{monoid!rational --}%
\emph{Commutative} Kleene monoids, as well as finitely generated submonoids 
of~$\Rat a^{\!*}$ 
are rational monoids~\cite{Rupe91,AfonKhaz10}.
% \ifcheat\clearpage\fi%

\ifcheat\smallskipneg\fi%
\paragraph{\secti{int-mul}. Introducing weights}%

If the definition of rational (and algebraic) series in non-commu\-ting 
variables as generalisation of rational (and context-free) languages 
on one hand-side, as well as the formalisation of rational 
expressions on the other, date back to the beginning of automata 
theory, the formalisation of \emph{weighted} rational expressions 
\index{expression!weighted rational --}%
seem to have appeared in various papers in the years~2000 only 
\cite{CaroFlou00,Rutt03,LombSaka05a}.
A satisfactory definition of trivial identities for weighted 
expressions proves to be tricky and has evolved in the publications of 
the author. 

By replacing quotient and derivation by \emph{co-induction},
\index{co-induction}%
Rutten formulated the equivalent of \theor{ant-mul}~\cite{Rutt03}. 

Krob~\cite{Krob92} and Berstel--Reutenauer~\cite{BersReut08} have 
considered `weighted rational expression'
slightly different from those expression dealt with in this chapter.
% chapter, with restrictions on the weight semirings, complete for the 
% former, commutative for the latter.
With their \emph{differentiation} and \emph{derivation}, they have 
tackled different problems than the construction of $\EtAs$-maps.
% , extension of complete systems of 
% identities for the former, representation of \emph{recognisable} 
% series over free partially commutative monoids for the latter.

\ifcheat\smallskipneg\fi%
\paragraph{Acknowledgements}%

The author is grateful to Z.~{\'E}sik and to J.~Brzozowski who read a 
first draft of this chapter and made numerous and helpful remarks.
P.~Gastin, A.~Demaille, and H.~Gr\"uber sent corrections on the first 
version.
The careful reading of the final version by A.~Szilard has been very 
encouraging and most helpful, and is heartily
acknowledged. 
% The unbounded patience of Jean-{\'E}ric Pin is also gratefully 
% acknowledged. 

}%

%%%%%%%%  AE Section References  %%%%%%%%%%%
%%%             150114                   %%%
%%%%%%%%%%%%%%%%%%%%%%%%%%%%%%%%%%%%%%%%%%%%

\newcommand{\BibDir}{./bibinputs/}%

\bibliographystyle{abbrv}
\addcontentsline{toc}{section}{References}
\ifcheat\smallskipneg\fi%
\ifcheat\smallskipneg\fi%
\ifcheat\smallskipneg\fi%
\ifcheat\smallskipneg\fi%
\begin{small}
\bibliography{\BibDir AMAHB-abbrevs,%
              \BibDir Alexandrie-abbrevs,%
              \BibDir Alexandrie-AC,%
              \BibDir Alexandrie-DF,%
              \BibDir Alexandrie-GL,%
			  \BibDir Alexandrie-MR,%
              \BibDir Alexandrie-SZ}
\end{small}
\ifhdbk
\newpage{\pagestyle{empty}\cleardoublepage}
\else
\clearpage 
\fi
%%%%%%%%%%%%%

\ifhdbk\else%%%%%%  AE Section Index  %%%%%%%%%%
%%%             120211                   %%%
%%%%%%%%%%%%%%%%%%%%%%%%%%%%%%%%%%%%%%%%%%%%
% \clearpage 
\addcontentsline{toc}{section}{Index}
\markright{\indexname}\markboth{\indexname}{\indexname}
\printindex
%%%%%%%%%%%%%
 
\fi
%==================================================

%==================================================
\end{document}